\documentclass[preprint,review,8pt,3p]{elsarticle}

\usepackage{graphicx}
\usepackage{amssymb}
\usepackage{paralist}
\usepackage{amsmath}
\usepackage{amsthm}
\usepackage{float}
\usepackage{multirow}
\usepackage[table]{xcolor}% http://ctan.org/pkg/xcolor
\usepackage{tikz}
\usepackage{csquotes}
\biboptions{sort&compress}
\usepackage{rotating}
\usepackage{url}
\usepackage{doi}
\usepackage{color}
\usepackage{graphicx}
\usepackage{subcaption}
\usepackage{bbm}
\usepackage{nomencl}
\makenomenclature

\newcommand{\mc}{\mathcal}

\newcommand{\1}{\mathbbm{1}}

\allowdisplaybreaks
\newtheorem{assum}{Assumption}
\newtheorem{definition}{Definition}

\newtheorem{thm}{Theorem}
\newtheorem{prop}{Proposition}

\newtheorem{example}{Example}

\hypersetup{
	colorlinks   = true, %Colours links instead of ugly boxes
	urlcolor     = blue, %Colour for external hyperlinks
	linkcolor    = blue, %Colour of internal links
	citecolor    = blue %Colour of citations
}

\journal{Transportation Research Part B: Methodological}

\begin{document}
\begin{frontmatter}


\title{Risk-aware Urban Air Mobility Network Design with Overflow Redundancy}



\author{Qinshuang Wei\footnote{Corresponding Author. Postdoctoral Fellow, Oden Institute for Computational Engineering and Sciences, {\tt\small qinshuang@utexas.edu}}, Zhenyu Gao\footnote{Postdoctoral Fellow, Department of Aerospace Engineering and Engineering Mechanics, {\tt\small zhenyu.gao@utexas.edu}}, John-Paul Clarke\footnote{Professor and Ernest Cockrell, Jr. Memorial Chair in Engineering, Department of Aerospace Engineering and Engineering Mechanics, {\tt\small johnpaul@utexas.edu}}, and Ufuk Topcu\footnote{Professor, Department of Aerospace Engineering and Engineering Mechanics and Oden Institute for Computational Engineering and Sciences, {\tt\small utopcu@utexas.edu}}}

\address{The University of Texas at Austin, Austin, Texas 78712-1221, United States}

\begin{abstract}

%Urban Air Mobility (UAM) is a transformative air transport concept for operating multiple services at lower altitudes within or traversing metropolitan areas. 
Urban air mobility (UAM), as envisioned by aviation professionals, will transport passengers and cargo at low altitudes within urban and suburban areas. 
%Among the challenges in UAM infrastructure design, properly managing overflow traffic under physical and operational disruptions is critical to ensure the system's safety and efficiency. 
To operate in urban environments, precise air traffic management, in particular the management of traffic overflows due to physical and operational disruptions will be critical to ensuring system safety and efficiency. 
To this end, we propose UAM network design with reserve capacity, i.e., a design where alternative landing options and flight corridors are explicitly considered as a means of improving contingency management. Similar redundancy considerations are incorporated in the design of many critical infrastructures, yet remain unexploited in the air transportation literature. 
In our methodology, we first model how disruptions to a given UAM network might impact on the nominal traffic flow and how this flow might be re-accommodated on an extended network with reserve capacity. Then, through an optimization problem, we select the locations and capacities for the backup vertiports with the maximal expected throughput of the extended network over all possible disruption scenarios, while the throughput is the maximal amount of flights that the network can accommodate per unit of time.  
We show that we can obtain the solution for the corresponding bi-level and bi-linear optimization problem by solving a mixed-integer linear program. We demonstrate our methodology in the case study using networks from Milwaukee, Atlanta, and Dallas--Fort Worth metropolitan areas and show how the throughput and flexibility of the UAM networks with reserve capacity can outcompete those without.

\end{abstract}


\end{frontmatter}

% \linenumbers

% Nomenclature

% \nomenclature[1]{\(\mc V\)}{The set of nodes (vertiports)}
% \nomenclature[2]{\(\mc E\)}{The set of edges (flight corridors)}
% \nomenclature[3]{\(\mc R\)}{The set of routes (paths)}
% \nomenclature[4]{\(\mc V^b\)}{The set of all backup nodes (vertiports)}
% \nomenclature[5]{\(\mc E^b\)}{The set of all backup edges (corridors)}
% \nomenclature[6]{\(\mc V^c\)}{The set of chosen backup nodes (vertiports)}
% \nomenclature[7]{\(\mc E^c\)}{The set of chosen backup edges (corridors)}
% \nomenclature{\(x_R\)}{flow through route $R$}
% \nomenclature{\(R\)}{Route(Path)}
% \nomenclature{\(c_i\)}{Actual capacity of link or node $i$}
% \nomenclature{\(C_i\)}{Maximal capacity of link or node $i$}
% \nomenclature{\(C\)}{Capacity vector for all links and nodes}
% \nomenclature{\(Z\)}{Vector of decision variable}
% \nomenclature{\(K\)}{Vector of cost for vertiports construction}

% \nomenclature{\(\mathbb{R}\)}{The set of all real numbers}
% \nomenclature{\(\mathbb{N}\)}{The set of all natural numbers}

% \printnomenclature

\section{Introduction}\label{sec:intro}

Urban air mobility (UAM) is a transformative air transportation concept that will primarily utilize electric vertical take-Off and landing (eVTOL) aircraft and unmanned aerial systems (UAS) at lower altitudes within or traversing metropolitan areas for application scenes such as passenger mobility, cargo delivery, infrastructure monitoring, public safety, and emergency services~\citep{cohen2021urban}. By expanding the urban transportation into the sky, UAM will enable the transport of passengers and cargo more effectively and serve the public good, especially in the currently underserved local and regional environments.

Large-scale deployment of UAM requires solutions to a range of technological and operational challenges to ensure system safety, efficiency, and sustainability~\citep{garrow2021urban,cohen2021urban,gao2023noise,wei2022safe,wei2021scheduling}. Among the challenges in air traffic management and infrastructure design, air traffic overflow is a critical factor. The term overflow refers to a situation where a quantity exceeds the limitation or capacity of a given system. Properly managing overflow can help avoid unfavorable consequences such as system inefficiency, loss, and failure, and is therefore a common consideration in the design of critical infrastructures, such as intensive care units in hospitals~\citep{litvak2008overflow}, telecommunication systems~\citep{glkabowski2008modeling}, and sewage discharges~\citep{joseph2014minimization}. In air transportation, air traffic overflow occurs when capacities at nodes (vertiports) and links (flight corridors) are reduced due to physical and operational disruptions such as weather, power outage, incidents, and temporal traffic control. Effective accommodation of overflow air traffic and aircraft in a UAM network is of vital significance to two aspects of operations. First, any air transportation infrastructure is safety-critical in nature. When certain disruptions happen to a vertiport or flight corridor, the affected aircraft in operation must have alternative options to safely divert and land. Second, from the economic perspective, even small local spill-over effects under disruption could cascade and result in network-scale delays, cancellations, and even network failure, which translate into direct economic and welfare losses. These observations highlight the importance of managing air traffic overflow via network design with improved redundancy and flexibility.

To date, network design with redundancy in functionality for risk and contingency management has not been fully exploited in the air transportation. However, the ground transportation and supply chain engineering communities have tackled network design with redundancy consideration from various angles. Most related research efforts focus on finding optimal network expansion policies to improve network redundancy, which mainly includes two aspects: (1) network capacity~\citep{zhang2012enhancing,kerner2016maximization} or reserve capacity~\citep{chen1999capacity}, and (2) travel alternative diversity~\citep{xu2018redundancy}. Such a network design problem is commonly formulated as a bi-level optimization problem, where the upper-level and lower-level sub-problems optimizes the network design objective and traffic flow conditions, respectively. Network design under uncertainty~\citep{sumalee2009evaluation} further utilizes bi-level stochastic program~\citep{liu2009twostage,unnikrishnan2009freight} to take into account stochastic demand and/or capacity. On the other hand, network design for UAM operations has become an active research stream in recent years. People have proposed methods to design UAM networks that can best serve passenger demand~\citep{wu2021integrated,kai2022vertiport}, minimize traffic congestion~\citep{yu2023vertiport}, and plan for various operational strategies~\citep{willey2021uamnetwork,wei2023dynamic}. Albeit these fruitful advances, design with redundancy consideration for risk and contingency management is still a noteworthy research gap in UAM network.

\begin{figure}[htbp]
	\centering
        \includegraphics[width=0.45\textwidth]{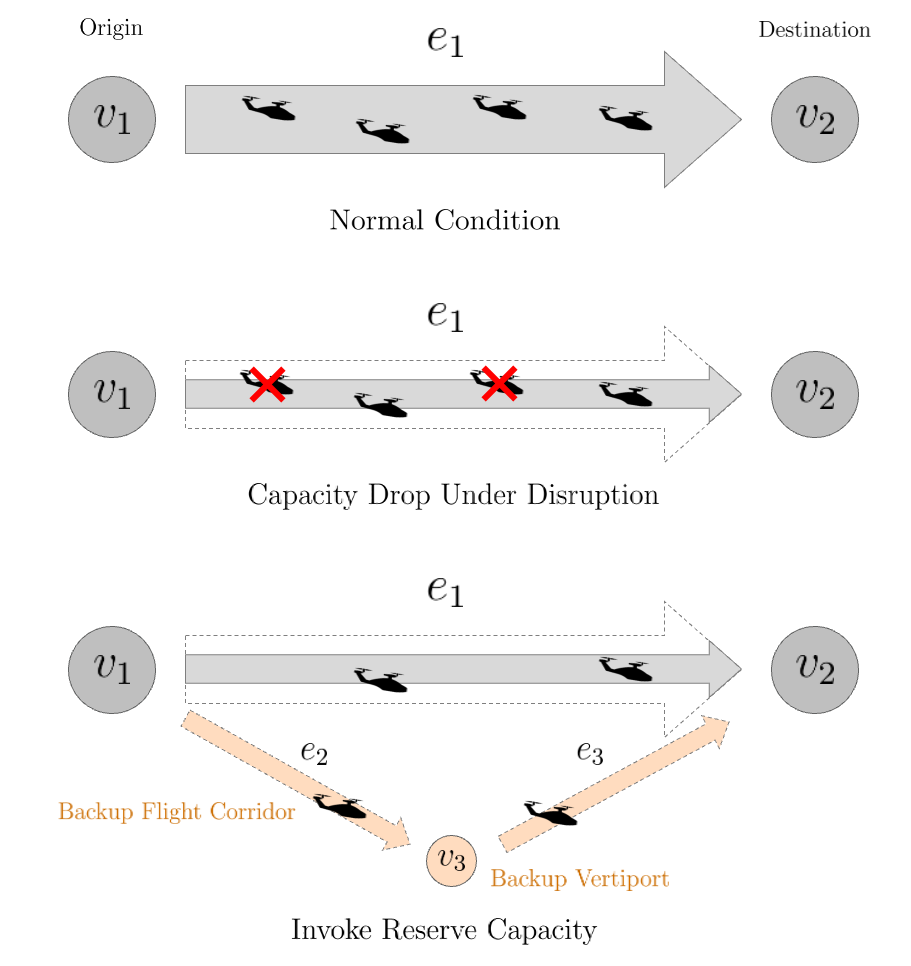}
	\caption{The concept of utilizing reserve capacity for operations under disruption. The three graphs demonstrate the network with an origin-destination pair under normal condition (top), disrupted condition with overflow (middle), and disrupted condition with reserve capacity (bottom). }
	\label{fig:Concept}
\end{figure}

We propose the design of an air mobility network with a small amount of reserve capacity to better manage overflow and thus mitigate the impact of disruption over the entire network. The core idea is to provide alternative flight routes and landing options to reallocate aircraft in case of service disruptions and emergencies. Figure~\ref{fig:Concept} uses a simple example to illustrate the design idea. In Figure~\ref{fig:Concept}, the top figure illustrates an origin-destination (O-D) pair under normal condition; the middle figure illustrates the scenario when the capacity of the flight corridor drops to half the original capacity due to disruption; the bottom figure illustrates how the additional reserve capacity accommodates the overflow traffic during disruption, utilizing a backup vertiport and two backup flight corridors. Because an air transportation network is a node-based spatial network~\citep{yoo2016airtrans}, at the kernel of this risk-aware air mobility network design problem is a network design problem (NDP) considering two types of nodes with permanent infrastructures---the larger ordinary vertiports, and the smaller backup vertiports. The design plans for the optimal locations and capacities of the backup vertiports to maximize the expected throughput of the network, considering all possible disruption scenarios, budget constraints, and cost-effective analysis, where we define the throughput of an air transportation network as the maximal amount of flights that the network can accommodate per unit of time. Overall, we summarize our four primary contributions as follows:
\begin{enumerate}
    \item \textit{Proposing a novel concept of air transportation network design with reserve capacity.} To the best of our knowledge, the concept of air transportation network with both regular and reserve functionalities has not emerged in the current literature. In addition to a regular air mobility network design driven by passenger demand and operational efficiency, we allocate limited resources (budget) and place a small amount of less-utilized capacity in reserve throughout the network. Such reserve capacity provides alternative landing sites and flight corridors for the network to adapt to disruptions. Figure~\ref{fig:Demos} displays an illustration of the design idea on a real-world UAM network.
    \item \textit{Modeling a risk-aware air mobility network with reserve capacity.} We consider potential disruptions to an air mobility network by characterizing the epistemic uncertainties---in the capacities of vertiports and travel corridors---with discrete probability distributions. We then model a extended (supplementary) network for operations under disruptions, where we enable the additional reserve capacity by constructing backup vertiports on a set of candidate locations.
    \item \textit{Formulating, simplifying, and solving a bi-level optimization problem that optimally plan backup vertiports under limited budget.} To maximize the extended network's expected throughput, we formulate a bi-level optimization problem for selecting the best locations and capacities of backup vertiports under budget constraints and cost-benefit trade-off, and show that we can find the optimal solution by solving a mixed-integer programming problem with bi-linear constraints. We then further simplify the problem into a mixed-integer linear program (MILP) with a manageable size of variables.
    \item \textit{Demonstrating the proposed design and its benefits via three representative use cases.} We apply our proposed approach on UAM network designs for Milwaukee, Atlanta, and Dallas--Fort Worth metropolitan areas, which have different network scales and topologies. The three networks contains 7, 11, and 15 ordinary vertiports, 12, 58, and 64 ordinary flight corridors, as well as 24, 26, and 45 candidate backup nodes, respectively. We generate design outcomes at different budget and cost-benefit valuation levels, comprehensively evaluate the design outcomes through visualizations and quantitative metrics, and discuss insights obtained from the results.
\end{enumerate}

\begin{figure}[htbp]
	\centering
        \includegraphics[width=0.3\textwidth]{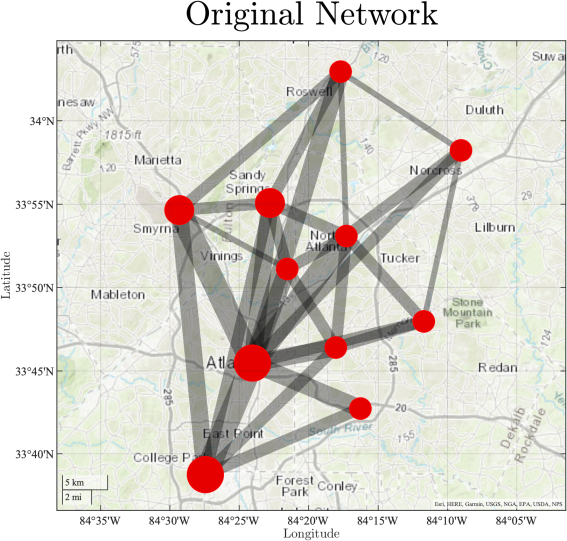}
        \hspace*{1.45cm}
        \includegraphics[width=0.3\textwidth]{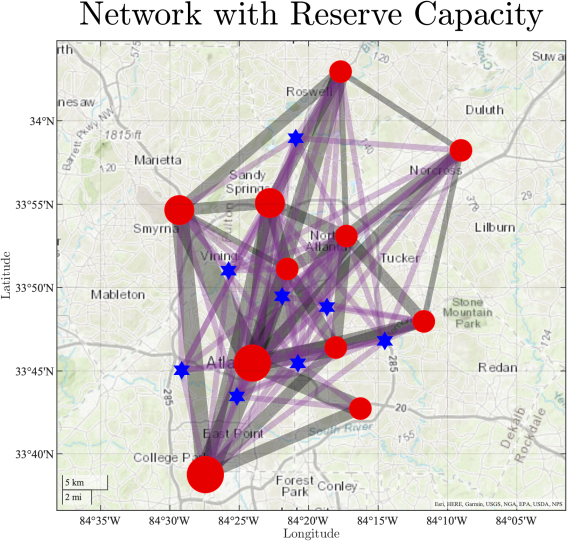}
	\caption{Illustrations of the original UAM network (left) and network with reserve capacity -- alternative flight routes and landing sites (right).}
	\label{fig:Demos}
\end{figure}

We come up with the first concept of air transportation network with both ordinary and backup functionalities and consider it a crucial step to the design of risk-aware air transportation networks outside of UAM, such as other detailed branches of advanced air mobility (AAM). The method is particularly useful in improving the redundancy of a demand-driven air transportation network. The remainder of the paper is organized as follows. Section~\ref{sec:litreview} reviews literature in three relevant streams and identifies research gap. Section~\ref{sec:network} introduces the two air transportation network models. Section~\ref{sec:selection} formulates and solves the mathematical program to find the optimal network designs. Section~\ref{sec:casestudy} applies the proposed approach to three representative use cases and evaluates the results. Section~\ref{sec:remarks} discusses the limitations and extensions of the study before Section~\ref{sec:conclusion} concludes the paper.

\section{Background and Literature Review}\label{sec:litreview}

This work is related to three streams of literature in the domains of transportation and aeronautical engineering: (1) network redundancy and capacity, (2) network design problem, and (3) air transportation and UAM network. This section serves to review the key concepts and recent progresses of these three research areas and identify the research gaps to be addressed.

\subsection{Network Redundancy and Capacity}

Network redundancy and capacity are two key concepts in this work. Research efforts to quantify the redundancy and capacity of a network, as well as design of networks with improved redundancy and capacity have been active in the past two decades among the ground transportation communities. Many related concepts in the literature are designed to quantify a network or system's flexibility~\citep{morlok2004measuring,chen2011modeling}, i.e., the ability of a network to adapt to external changes, while maintaining satisfactory system performance. External changes are uncontrolled factors that affect the system, including changes in level of demand or use, shifts in spatial traffic patterns, infrastructure loss and degradation, etc. System performance is characterized by parameters such as level of service, maintainability, and profitability. Overall, network redundancy refers to the extent to which elements or systems of a network are substitutable~\citep{xu2018redundancy}. In the context of UAM network, redundancy is the existence of alternative landing sites and flight corridors between origins and destinations that can serve the same purpose and result in less serious consequences in case of a degradation in some part of the system. In supply chain design problems, it can also refer to under-utilized capacity/resource in reserve to be used in case of disruption~\citep{jansuwan2021redundancy,pavlov2019optimization}. When quantifying network redundancy, both network capacity and travel alternative diversity should be taken into account~\citep{jansuwan2021redundancy}.

Network capacity is the maximum admittable traffic flow that can be accommodated by the network, without violating the the system’s resources (fleet, terminal and link capacities, node capacities~\citep{dui2022maintenance}, etc.) under network equilibrium~\citep{zhang2012enhancing}. It can also be interpreted as the maximum throughput of the network at which flows on the bottleneck links just reach their capacity~\citep{yang1998paradox,kerner2016maximization}, but no queues develop behind the bottlenecks. Network capacity can be computed as the largest possible sum of O-D flows that can be handled by the network and can be estimated by the MAXCAP model, formulated as maximize $\sum_{r=1}^R x_r$ subject to limitations of link capacity, node capacity, available fleet, conservation of flow, and routing options available in the system, where $x_r$ is the traffic volume (flights/hr) on O-D pair $r$. Network reserve capacity~\citep{chen1999capacity,morlok2004measuring,lucker2019roles} is another important concept in the literature. It measures the difference between the current traffic and the maximum traffic of a transportation network. It is defined as the largest multiplier $\mu$ applied to an existing O-D demand matrix that can be allocated to a transportation network in a user-optimal way without violating the link capacities or exceeding a pre-specified volume to capacity ratio (level of service). The problem is written as maximize $\mu$ subject to $\nu_a (\mu \boldsymbol{q}) \leq C_a, \forall a \in A$, where $\nu_a (\mu \boldsymbol{q})$ is the equilibrium flow on link a with the demands of all O-D pairs being uniformly scaled by $\mu$ times the based O-D demands $\boldsymbol{q}$ and $C_a$ is the capacity of link $a$. In~\citep{chen2011modeling}, $\nu_a (\mu \boldsymbol{q})$ is obtained by solving a user equilibrium problem. Given the network topology, known link capacities, and O-D demands, this problem is deterministic. However, stochasticity could also exist in sources such as link capacity~\citep{kato2023estimation} and demand~\citep{bayram2023hub}.

Travel alternative diversity refers to the number of effective paths between a specific O-D pair~\citep{xu2018redundancy}. In a more strict criterion, an alternative path does not share any of the links or nodes except for the origin and destination nodes. Some literature~\citep{millerhooks2012} mention network resilience as a function of the number of reliable paths between all O-D pairs. A reliable path must also meet a given level of service (e.g., travel time, distance)~\citep{morlok2004measuring}. For example, in~\citep{kerner2016maximization}, the alternative routes are defined as possible routes from origin $O_i$ and destination $D_j$ for which the maximum difference between travel times in free flow conditions does not exceed a given threshold value $\Delta T_{ij}^{(\text{th})}$.

\subsection{Network Design Problem}

% How to formulate
% Connection to our formulation and problem nature
% Design under uncertainty
% How to solve

Network design problem (NDP) has a long history in the transportation domain~\citep{magnanti1984network}. Our problem of air transportation network design with redundancy in functionality has a similar formulation with the capacity expansion problem in the literature~\citep{mathew2009expansion,liu2021estimation}. In this formulation, a regular UAM network, which is a completely business-driven design for on-demand passenger service, is prepared as a base network. It is designed in accordance with urban structure, demand forecast, ground transportation and infrastructure, etc. Its node and link capacities are restricted by factors such as geographic conditions, safety considerations, and societal impacts. The capacity expansion problem starts with the identification of a set of candidate locations for backup nodes, which are selected in accordance with geographic restrictions, infrastructure, and societal impacts~\citep{gao2023noise}, and have much smaller capacities than the regular nodes. Because building and maintaining such reserve infrastructures bring additional costs, this capacity expansion problem becomes finding optimal network expansion policies by choosing optimal decision variables (for facility location and resource allocation) in terms of capacity enhancement values under budget constraints and/or cost-benefit analysis.

Such an NDP is often formulated as a bi-level optimization problem in the literature. The upper-level sub-problem optimizes the network design objective such as network capacity, reserve capacity, and total system cost, and the decision variables are dependent on the actual problem. The lower-level sub-problem determines link flow through equilibrium conditions, which can be obtained by route choice behavior, traffic equilibrium model,  user equilibrium (e.g., game theory), and traffic routing and assignment. In the design of ground transportation network, drivers' route choice behaviors must be explicitly considered. In some other types of networks (e.g., communication), traffic have no preference in choosing paths and can be assigned to take paths that lead to maximum flow in the network. In most cases, such bi-level formulations of NDP are nonconvex and non-differentiable.

Network design under uncertainty~\citep{sumalee2009evaluation,govindan2020integrated,arani2021lateral,xiang2022robust} considers stochastic demand, supply, and/or network capacity in the design problem. The most commonly used criterion in network design under uncertainty is the mean performance (expectation). In general, expectation-based strategies are risk-neutral and tend to perform well in the long run in a repetitive environment~\citep{huang2022optimal}. Robust network design not only maximizes the expected future system performance but also minimizes the variance of system performance~\citep{unnikrishnan2009freight}. Robust approaches are introduced to handle decision making under environments of extreme uncertainty and focuses on the worst-case scenario, which is usually more conservative than techniques focusing on expectation. Common formulations for network design under uncertainty include bi-level stochastic program~\citep{unnikrishnan2009freight,liu2009twostage,millerhooks2012}, expected value model, mean-variance model, chance-constrained model, and minimax model~\citep{chen2011review}.

\subsection{Air Transportation and UAM Network}

% Uniqueness of air transport network compared to ground transport in the previous two aspects
% Uniqueness of UAM network compared to normal air transport network in the need of reserve capability

As an air transportation system, the nature of UAM network is node-based spatial network with node capacity~\citep{yoo2016airtrans}. Such a network has three main features: (1) spatial network: the nodes are located in a space equipped with a metric; (2) node-based network: the cost of link connection is relatively low; and (3) functional network: each node has a capacity to function. The idea of strategic network capacity expansion can be found in some ground transportation literature, where optimal decisions are made in the NDP to either enhance the capacity of existing links or introduce new links~\citep{yang1998paradox,cats2015planning}. Air transportation, however, is a unique problem such that its network capacity is expanded through introducing additional nodes into the network, like a facility location problem. These backup/overflow nodes then drive alternative links to complete a network with reserve functionalities. 

Network design for UAM operations has become active especially in the past five years. Researchers from various disciplines (aerospace engineering, operations research, transportation, urban planning, etc.) have been tackling this problem from various angles. In some earlier works, \citep{holden2016uber} and \citep{rajendran2019insights} apply k-means clustering and its variant to optimize the selection of vertiport locations. More recent works have started to utilize more advanced optimization setups in UAM network design. Among some representative works in recent years, \citep{wu2021integrated} and \citep{kai2022vertiport} design UAM network with optimal vertiport planning to serve passenger demand. \citep{yu2023vertiport} determines the locations and capacities of vertiports through minimizing the traffic congestion in a hybrid air-ground transportation network. \citep{willey2021uamnetwork} identifies vertiport locations by taking into account vehicle limitations and operational strategies such as multi-leg trips. \citep{wang2023uam} investigates UAM network structure and transportation benefits for several major cities in the world. 

\subsection{Research Gap}

Through reviewing literature from the three most relevant research streams, we identify two research gaps which have not been adequately addressed up to this point. First of all, studies revolving around network redundancy, network capacity, and network expansion/enhancement have been found only in ground transportation literature. Second, most UAM network design works have focused on objectives such as passenger demand, travel time, and operational strategy, without considering redundancy for risk and contingency management. The lack of studies on similar problems in air transportation could be ascribed to two main reasons: (1) the already dense networks of commercial and general aviation networks for diversion and landing (especially in places where the aviation industry is highly developed); and (2) the unrealistically high costs of building and maintaining reserve infrastructures (airports with runway) for commercial and general aviation. UAM network, on the other hand, is a paradigm case to research and design air transportation network with redundancy because of the concrete needs and manageable costs. Building or repurposing a handful of small backup vertiports throughout the city could significantly benefit the system's safety and continuity of operation under disruptive events. This study aims to fill in these two gaps through an approach to optimally determine the number, locations, and capacities of backup vertiports to enhance the throughput and flexibility of the entire UAM network.

\section{Network Models}\label{sec:network}

This section formally introduces the network models involved in this work. We first introduce a risk-aware network model that depicts the UAM network with possible disruptions, referred to as the \emph{original network}, followed by definitions of the \emph{extended network} that supplements the original network upon disruption with additional landing options and flight corridors.

\subsection{Notations}
In the rest of the paper, we denote all real numbers as $\mathbb{R}$, positive real number as $\mathbb{R}_{>0}$, nonnegative real number as $\mathbb{R}_{\geq 0}$, natural number including $0$ as $\mathbb{N}$, and natural number without $0$ as $\mathbb{N}_{>0}$. For any matrix $A$, we use $A_+$ to denote a matrix of the same size as $A$ and while each entry of $A_+$ is the absolute value of the corresponding entry in $A$, i.e., $[A_+]_{ij}  = \lvert[A]_{ij}\rvert$. The notation $\1_{ K}$ for any $K\in\mathbb{N}_{>0}$ represents a $K \times 1$ vector of ones. The notation $\mathbf{0}^{K_1 \times K_2}$ for any $K_1, K_2 \in\mathbb{N}_{>0}$ represents a $K_1 \times K_2$ matrix of zeros. We also define $\vec{e}_{K}(i)$ for any $i, K\in\mathbb{N}_{>0}$ and $i \leq K$ as an elementary vector of length $K$ with $1$ only at the $i$'th entry and $0$ elsewhere.

% We let $d_{x_1,x_2}$ represent the Euclidean distance between the two points. 

\subsection{Risk-Aware Original Network Model}
\label{sec:original_model}
We model an air transportation network with a directed graph $\mc G = (\mc V, \mc E)$, where $\mc V = \{v_1,v_2, \cdots, v_{N_{\mc V}}\}$ is the set of all nodes and $\mc E= \{e_1,e_2, \cdots, e_{N_{\mc E}}\}$ is the set of all links in the network, with $N_{\mc V}$ (resp., $N_{\mc E}$ denoting the number of nodes (resp. links). Nodes are physical landing sites, also called vertiports, for UAM aircraft to land or takeoff; links are corridors of airspace connecting nodes that allow the flights to travel through. For any link $e = (v_i,v_j) \in \mc E$ where $v_i,v_j \in \mc V$, we denote $\sigma_e = v_j$ (resp., $\tau_e = v_i$) as the head (resp., tail) of $e$. % We denote the number of links (resp. nodes) as $N_{\mc E}$ (resp., $N_{\mc V}$.
We further define the index set $\mc J_{\mc V} = \{1,2,\cdots,{N_{\mc V}}\}$ (resp., $\mc J_{\mc E} = \{1,2,\cdots,{N_{\mc E}}\}$) for the nodes (resp., links). A route is a path in the network connecting multiple nodes along its route.
% We assume that each flight is assigned a route and the flight has to land at each vertiport along its route.

Each node $v \in \mc V$ (resp., link $e \in \mc E$) has capacity $c_v \in \mathbb{R}_{\geq 0}$ (resp., $c_e \in \mathbb{R}_{\geq 0}$). 
In real-world operations, disruption events, such as extreme weather, may lead to nondeterministic capacity degradations of vertiports and flight corridors. We therefore assume that the capacity of each link or node $ev \in \mc E \cup \mc V$, $c_{ev} \in [0,C_{ev}]$ for some positive number $C_{ev}$. A node or link $ev$ is \emph{not disturbed} if $c_{ev} = C_{ev}$, and $ev$ is \emph{disturbed} if $c_{ev} < C_{ev}$. The \emph{network is disturbed} if any of the nodes or links in the network is disturbed. %We denote the vector of node capacities $C^{\mc V}= \{c_v\}_{v\in \mc V}$, the vector of link capacities $C^{\mc E}= \{c_e\}_{e\in \mc E}$, and the vector of all capacities $C= \{c_i\}_{i\in \mc V \cup \mc E}$.

%We assume that the network will be disturbed with probability $P_{dis}$. For convenience, we assume that the probability that each node or link is disturbed is identical, i.e., $P^{ev}_{dis} := Prob\{c_{ev} < C_{ev}\} = P_{dis}/{(\lvert \mc E \rvert+\lvert \mc V \rvert)}$ for all $ev \in \mc E \cup \mc V$. 
Our model is based on the following assumptions.
\begin{assum} [Disturbed Network]\
\label{assum:disturb}
\begin{enumerate}
    \item The network is disturbed with probability $P_{dis} \in [0,1]$.
    \item Only one link or node $i \in \mc E \cup \mc V$ will be disturbed at a time.
    \item When a node or link $ev\in \mc E\cup \mc V$ is disturbed, $c_{ev} \in \{c_{ev,1}, c_{ev,2}, \cdots, c_{ev,k_{ev}}\}$ with $k_{ev} \geq 1$, where $ C_{ev} > c_{ev,1} > c_{ev,2}> \cdots > c_{ev,k_{ev}} \geq 0$ and $k_{ev}$ represents the total number of values the capacity of disturbed $ev$ may take.
    \item $Prob\{c_{ev} = c_{ev,k}\} = p_{ev,k} \in [0,P_{dis}]$ for all $k = 1, \cdots, k_{ev}$, such that $\sum_{ev \in \mc E\cup \mc V} \sum_{k = 1}^{k_{ev}} p_{ev,k} = P_{dis}$.
    \item  The distributions of the disturbed capacity $c_{ev}$ given that $ev$ is disturbed are independent for all $ev \in \mc E \cup \mc V$.
\end{enumerate}
\end{assum}

The third item in Assumption~\ref{assum:disturb} demonstrate the discretized capacity values of the disturbed links and nodes. The second and fifth items in Assumption~\ref{assum:disturb} are key assumptions in this work which are also supported by some existing works in the literature. For example, the assumption that only one link or node fails at a time can be found in~\citep{sumalee2009evaluation}. The assumption that capacity degradation of each network link or node is treated independently can be found in \citep{chen1999capacity} and \citep{peeta2010pre}. 

%The traffic flow in the network can be interpreted as the amount of flights traveling through an link per unit of time.
\begin{assum}[Traffic Flow]\
\label{assum:flow}
\begin{enumerate}
    \item The amount of flights traveling through each link per unit of time never exceeds its capacity.
%    \item The flights are only allowed to travel along a set of routes $\mc R$, and for all $R = \{e^R_{\ell} \}_{\ell =1}^{k_R} \in \mc R$, $\tau_{e^R_1} \in S$ and $\sigma_{e^R_{k_R}} \in T$.
    \item The total incoming and outgoing amount of flights per unit of time at each node never exceeds its capacity.
    \item The traffic is at equilibrium, so that the amount of flights entering and exiting a link per unit of time is the same.
  %  \item The incoming demands are always larger than the amount that the network can handle, i.e., we can regard $n_\ell \to \infty$ for all $\ell \in \mc J$.
\end{enumerate}
\end{assum}

%Based on the Assumption~\ref{assum:disturb}, we now define the probability distribution of $c_{ev}$ when the element $ev\in \mc E\cup \mc V$ is disabled: we assume the distribution is discrete with finite support, such that $c_{ev} \in S_{ev} = \{c_{ev,1}, c_{ev,2}, \cdots, c_{ev,k_{ev}}\}$ with $k_{ev} \geq 1$; $Prob\{c_{ev} = c_{ev,k}\} = p_{ev,k} \in [0,P^{ev}_{dis}]$, such that $\sum_{k = 1}^{k_{ev}} p_{ev,k} = P^{ev}_{dis}$ .% According to Assumption~\ref{assum:disturb}, $\sum_{ev \in \mc E\cup \mc V} \sum_{k = 1}^{k_{ev}} p_{ev,k} = P_{dis}$.
According to Assumption~\ref{assum:disturb}, the probability that the link or node $ev \in \mc E \cup \mc V$ is disturbed can be computed as $P^{ev}_{dis} := \sum_{k = 1}^{k_{ev}} p_{ev,k}$, and thus the probability that $ev$ is not disturbed is $Prob\{c_{ev} = C_{ev}\} = 1-P^{ev}_{dis}$. We therefore let $c_{ev,0} = C_{ev}$ and $p_{ev,0} = 1-P^{ev}_{dis}$.
For every $ev \in \mc E \cup \mc V$, we denote the vector of all possible capacities $\vec{C}^{ev}= \{c_{ev,k}\}_{k=0}^{k_{ev}}$, and the vector of all corresponding probabilities $\vec{P}^{ev}= \{p_{ev,k}\}_{k=0}^{k_{ev}}$. We let tuples $C^{\mc E} = (\vec{C}^{e_1},\vec{C}^{e_2},\cdots,\vec{C}^{e_{N_{\mc E}}})$, $C^{\mc V} = (\vec{C}^{v_1},\vec{C}^{v_2},\cdots,\vec{C}^{v_{N_{\mc V}}})$, $P^{\mc E} = (\vec{P}^{e_1},\vec{P}^{e_2},\cdots,\vec{P}^{e_{N_{\mc E}}})$, and $P^{\mc V} = (\vec{P}^{v_1},\vec{P}^{v_2},\cdots,\vec{P}^{v_{N_{\mc V}}})$ represent the possible capacities and corresponding probabilities for all links and nodes.

We now define the risk-aware air transportation network.

\begin{definition}[Risk-Aware Air Transportation Network]\ 
\label{def:network_ori}
A risk-aware air transportation network $\mc N$ is a tuple $\mc N = (\mc G,P_{dis},C^{\mc E},C^{\mc V},P^{\mc E},P^{\mc V})$, where $\mc G = (\mc V, \mc E),P_{dis},C^{\mc E},C^{\mc V},P^{\mc E},P^{\mc V}$ are the network graph, the network disturbed probability, the set of all possible node and link capacities, and their corresponding probabilities as defined above.
\end{definition}

We now define the vector of link and node capacities in a disturbed network.
We let $C^{\mc E}(e_i,k)$ represent the vector of link capacities when a link $e_i$ is disturbed and the corresponding capacity drops to $c_{e_i,k}$ for some $i \in \mc J_{\mc E}$ and $k = 1, \cdots, k_{e_i}$. Hence $C^{\mc E}(e_i,k)$ is a vector of length $N_{\mc E}$ such that $[C^{\mc E}(e_i,k)]_l = C_{e_l}$ for all $l \neq i$ and $[C^{\mc E}(e_i,k)]_i = c_{e_i,k}$. 
Similarly, we let $C^{\mc V}(v_j,k)$ represent the vector of node capacities when a node $v_i$ is disturbed and the corresponding capacity drops to $c_{v_j,k}$ for some $j \in \mc J_{\mc V}$ and $k = 1, \cdots, k_{v_j}$. Hence $C^{\mc V}(v_j,k)$ is a vector of length $N_{\mc V}$ such that $[C^{\mc V}(v_j,k)]_l = C_{v_l}$  for all $l \neq j$ and $[C^{\mc V}(v_j,k)]_j = c_{v_j,k}$.
We also define the vectors of undisturbed link and node capacities, $C^{\mc E}_0 = \{C_{e_i}\}_{i = 1}^{N_{\mc E}}$ and $C^{\mc V}_0 = \{C_{v_j}\}_{j = 1}^{N_{\mc V}}$. For any $k$, if a node $v \in \mc V$ is disturbed, we let $C^{\mc E}(v,k) = C^{\mc E}_0$, if a link $e \in \mc E$ is disturbed, we let $C^{\mc V}(e,k) = C^{\mc V}_0$.
%and for convenience, we let  (resp., $C^{\mc V}(ev,k) = C^{\mc V}_0$) for any $ev \notin \mc E$ (resp., $ev \notin \mc V$) for any $k$. 

\begin{definition}[Momentary Air Transportation Network]\
    \label{def:moment_ATN}
A \emph{Momentary Air Transportation Network} is a tuple $\mc N^M = (\mc G,\vec{C}_1,\vec{C}_2)$, where $\mc G= (\mc V, \mc E)$ is the network graph as defined above, $\vec{C}_1 \in \mathbb{R}_{\geq 0}^{N_\mc E}$ and $\vec{C}_2 \in \mathbb{R}_{\geq 0}^{N_\mc V}$. A momentary air transportation network $\mc N^M = (\mc G,\vec{C}_1,\vec{C}_2)$ is a disturbed or undisturbed \emph{Moment} of an air transportation network $\mc N = (\mc G,P_{dis},C^{\mc E},C^{\mc V},P^{\mc E},P^{\mc V})$ where $\mc G = (\mc V, \mc E)$ if $\vec{C}_1,\vec{C}_2$ are the undisturbed or disturbed vector of capacities for the links and nodes, i.e., 
\begin{enumerate}
    \item $\mc N^M$ is a \emph{disturbed moment} of $\mc N$ if $\vec{C}_1 = C^{\mc E}(ev,k)$ and $\vec{C}_2 = C^{\mc V}(ev,k)$ for some $ev \in \mc V \cup \mc E$ and $k = 1, \cdots, k_{ev}$. Formally, we let $\mc N(ev,k):= (\mc G,C^{\mc E}(ev,k),C^{\mc V}(ev,k))$;
    \item  $\mc N^M$ is an \emph{undisturbed moment} of $\mc N$ if $\vec{C}_1 = C^{\mc E}_0$ and $\vec{C}_2 =C^{\mc V}_0$. Formally, we let $\mc N(0,0) = (\mc G,C^{\mc E}_0,C^{\mc V}_0)$,
\end{enumerate}
where $C^{\mc E}_0, C^{\mc V}_0, C^{\mc E}(ev,k)$, and $C^{\mc V}(ev,k)$ are the vectors of undisturbed and disturbed link and node capacities defined as above.
\end{definition}
%$\mc N(ev,k) = (\mc G,C^{\mc E}(ev,k),C^{\mc V}(ev,k))$

% \begin{definition}[Realization of Air Transportation Network]\ 
% \label{def:network_realization}
% Given an air transportation network $\mc N =(\mc G,P_{dis},C^{\mc E},C^{\mc V},P^{\mc E},P^{\mc V})$, where $\mc G = (\mc V, \mc E)$, % where $C^{\mc E} = (\vec{C}^{e_1},\vec{C}^{e_2},\cdots,\vec{C}^{e_{N_{\mc E}}})$ and $C^{\mc V} = (\vec{C}^{v_1},\vec{C}^{v_2},\cdots,\vec{C}^{v_{N_{\mc V}}})$, then 
% \begin{enumerate}
%     \item a \emph{disturbed realization} of $\mc N$ is a tuple $\mc N(ev,k) = (\mc G,C^{\mc E}(ev,k),C^{\mc V}(ev,k))$, where $ev \in \mc E \cup \mc V$ and $k = 1, \cdots, k_{ev}$. % or $\mc N(v,k_2) = (\mc G,C^{\mc E}_0,C^{\mc V}(v,k_2))$, where $e \in \mc E$, $v \in \mc V$, $k_1 = 1, \cdots, k_{e}$, and $k_2 = 1, \cdots, k_{v}$.
%     \item an \emph{undisturbed realization} of $\mc N$ is the tuple $\mc N(0,0) = (\mc N,  C^{\mc E}_0, C^{\mc V}_0)$,
% \end{enumerate}
% where $C^{\mc E}_0, C^{\mc V}_0, C^{\mc E}(e,k_1),C^{\mc V}(v,k_2)$ are the vectors of undisturbed and disturbed link and node capacities defined as above.
% \end{definition}

Consider an air transportation network $\mc N$ with graph $(\mc V, \mc E)$, and a disturbed or undisturbed moment of it $\mc N(i,k)$ defined as in Definition~\ref{def:moment_ATN}, where $i=0$ and $k =0$, or $i \in \mc E \cup \mc V$ and $k = 1, \cdots, k_i$ we then introduce the concepts of incidence matrices, demand, flow matrix, and throughput.

%According to Assumption~\ref{assum:model}, we associate a demand of the network---a request of traveling through the network---directly with a route $R\in \mc R$. We further assume that the demands associated with each route is infinite, that is, for each route, there are always infinite amount of travelers request to travel through the route.

\underline{Incidence matrix}: We let $E$ represent the $N_{\mc V} \times N_{\mc E}$ incidence matrix of the graph, that is, for all $i \in \mc J_{\mc V}$ and $j \in \mc J_{\mc E}$,
\begin{align}
    [E]_{ij} = \begin{cases} -1 &\text{ if } \tau_{e_j} = v_i\\
                            1 &\text{ if } \sigma_{e_j} = v_i\\
                            0 &\text{ otherwise} \,. 
    \end{cases}
\end{align}

\underline{Demands:} We associate a demand---a request of traveling through the network---with an O-D pair $(o,d)$, where $o\in \mc V$ (resp., $d\in \mc V$) is the origin (resp., destination) of the trip and $o \neq d$. We then denote the set of all demands as $S = \{s_\ell = (o_\ell,d_\ell)\}_{\ell =1}^{N_S}$, where $N_S$ is the number of O-D pairs with travel demand. We let $\mc J_S = \{1,2,\cdots, N_S\}$ represent the index set of the demands.
%We associate a demand---a request of traveling through the network---with an origin-destination pair (O-D pair). We denote all demands with the same O-D pair as $s_{o,d} = (o,d,n)$, where $o\in \mc V$ (resp., $d\in \mc V$) is the origin (resp., destination) and $n$ is the amount of flights request to travel from $o$ to $d$ per unit time. We assume that $o \neq d$. We then denote the set of all demands as $S = \{s_\ell = (o_\ell,d_\ell,n_\ell)\}_{\ell =1}^{N_S}$, where $N_S$ is the number of different O-D pairs. We let the $\mc J_S = \{1,2,\cdots, N_S\}$ be the index set of the O-D pairs.

To formulate the demands into the matrix form, we let $\Delta_1$ (resp., $\Delta_2$) be the $N_\mc V \times N_S$ destination (resp., origin) indicating matrix such that 
\begin{align}
    [\Delta_1]_{i\ell} = \begin{cases} 1 &\text{ if } d_\ell= v_i\\
                            0 &\text{ otherwise} \,. 
    \end{cases}
\end{align}
\begin{align}
    [\Delta_2]_{i\ell} = \begin{cases} -1 &\text{ if } o_\ell = v_i\\
                            0 &\text{ otherwise} \,,
    \end{cases}
\end{align}

We denote the amount of demand the momentary network $\mc N(i,k)$ will fulfill for any O-D pair $(o_\ell,d_\ell)$ per unit of time as $n_{\ell}(i,k)$ for all $\ell \in \mc J_S$.
We then represent the amount of flights with all different destinations (resp., origins) with the $N_\mc V \times N_S$ matrices $D_1(i,k)$ and $D_2(i,k)$, such that $[D_1(i,k)]_{j\ell} = [\Delta_1]_{j\ell}\cdot n_\ell(i,k)$ and $[D_2(i,k)]_{j\ell} = [\Delta_2]_{j\ell}\cdot n_\ell(i,k)$. Therefore, the total number of fulfilled demands in the network can be computed as
\begin{equation}
\label{eq:n2D}
    \sum_{\ell =1}^{N_S} n_{\ell}(i,k)= \1^T_{ N_S} D_1(i,k)^T \1_{ N_\mc V}\, ,
\end{equation}
with non-negative $n_{\ell}(i,k)$ for all $\ell \in \mc J_S$, and $D_1(i,k), D_2(i,k), \Delta_1,\Delta_2$ must satisfy the constraints below: 
\begin{align}
\label{eq:D_constr}
\begin{split}    
    & D_1(i,k)^T \1_{N_\mc V} = -D_2(i,k)^T \1_{ N_\mc V}\\
    & D_1(i,k) \leq  \Delta_1 M ,
    \quad  -D_2(i,k) \leq  -\Delta_2 M, 
    \quad D_1(i,k) \geq 0,
    \quad  D_2(i,k) \leq 0 \,,
\end{split}
\end{align}
where $M$ can be any large enough positive real number.

%Before introducing the throughput of the network, we first notice that the traffic flow through a route equals the traffic flow induced by the route at each node or link along the route. 

% Before we formally define the throughput, we first come up with a few matrices. We denote the $m \times m$ identity matrix as $I_{m}$.%, and the $m \times 1$ ones vector as $\1_m$ for any $m \in \mathbb{N}_{>0}$.

%To ease the computation, we reformulate the maximization problem with matrix form.

\underline{Flow matrix}: 
We let $X(i,k)$ be a $N_{\mc E} \times N_S$ flow matrix for all links and demands, where each entry $[X(i,k)]_{j\ell}$ represents the amount of flights assigned to demand $s_{\ell}$ traveling through the link $e_j$ per unit of time for all $j\in \mc J_{\mc E}$ and $\ell \in \mc J_S$. According to the flow conservation law, $X(i,k)$ must satisfies
\begin{equation}
\label{eq:flow_cons}
    EX(i,k) = D_1(i,k)+D_2(i,k), 
    \quad X(i,k) \geq 0
\end{equation}
%Similarly, we denote the vector of demands for all O-D pairs as $\vec{n}$, which is a $\lvert \mc J\rvert \times 1$ vector and $\vec{n}_\ell= n_\ell$ for all $\ell \in \mc J$.

\begin{example}
\label{ex:ori}
\begin{figure}
    \centering    
    \begin{tikzpicture}
    
    \node[draw, circle] (1) at (0, 0) {$v_1$};
    \node[text width=1cm] at (0.3,0.6) {$10$};
    
    \node[draw, circle] (2) at (2, 1) {$v_2$};
    \node[text width=1cm] at (2.3,1.6) {$15$};

    \node[draw, circle] (3) at (2, -1) {$v_3$};
    \node[text width=1cm] at (2.4,-0.4) {$15$};

    \node[draw, circle] (4) at (4, 0) {$v_4$};
    \node[text width=1cm] at (4.3,0.6) {$10$};

    \draw[->] (1) --  node[above] {$e_1$} node[below] {$8$}(2);
    \draw[->] (1) --  node[below] {$e_2$} node[above] {$4$}(3);
    \draw[->] (2) --  node[above] {$e_3$} node[below] {$4$}(4);
    \draw[->] (3) --  node[below] {$e_4$} node[above] {$8$}(4);
    % \draw[->] (4) --  node[below] {$e_5$}(2);
    % \draw[->] (3) -- node[above] {$[4,8]$} (6);
    % \draw[->] (3) -- node[below,xshift=-2mm, yshift=1mm] {$[5,6]$} (5);
    % \draw[->] (4) -- node[below,xshift=3mm, yshift=1mm] {$[1,5]$} (5); 
    % \draw[->] (5) -- node[above,xshift=3mm, yshift=-1mm] {$[3,6]$} (7);
    \end{tikzpicture}
    \caption{A notional air transportation network with 4 nodes and 4 links, with undisturbed capacities labeled beside the corresponding node or link.}
    \label{fig:ex_ori_network}
\end{figure}
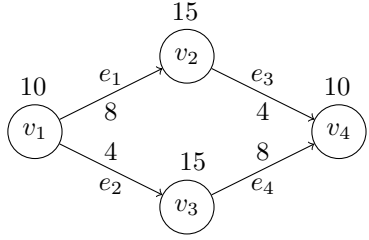

We use Fig.~\ref{fig:ex_ori_network} to illustrate our concepts of the air transportation network. We have the set of nodes $\mc V = \{v_1, v_2, v_3, v_4\}$, the set of links $\mc E = \{e_1, e_2, e_3, e_4\}$, and the undisturbed capacities of all nodes and links labeled in the figure.
We can hence obtain $E$, $C^\mc V_0$, and $C^\mc E_0$ directly,
\begin{equation*}
E =  \begin{bmatrix}
  -1 & -1 &  0 &  0  \\
   1 &  0 & -1 &  0  \\
   0 &  1 &  0 & -1 \\
   0 &  0 &  1 &  1  
\end{bmatrix} , \quad
C^\mc V_0 =  \begin{bmatrix}
   10\\
   15 \\
   15 \\
   10
\end{bmatrix} \, \quad
C^\mc E_0 =  \begin{bmatrix}
   8\\
   4 \\
   4 \\
   8
\end{bmatrix}
\end{equation*}

We then provide a set of O-D pairs with demand, $S = \{s_1 = (v_1,v_2), s_2 = (v_1,v_4), s_3 = (v_3,v_4)\}$, so that
\begin{equation*}
\Delta_1 =  \begin{bmatrix}
   0 & 0 & 0 \\
   1 & 0 & 0 \\
   0 & 0 & 0 \\
   0 & 1 & 1
\end{bmatrix}, \quad
\Delta_2 =  \begin{bmatrix}
   -1 & -1 & 0 \\
    0 &  0 &  0 \\
    0 &  0 & -1 \\
    0 &  0 & 0 
\end{bmatrix}.
\end{equation*}

We further construct a series of possible vectors of all possible disturbed capacities for all nodes and links, and their corresponding probability vectors as follows:
\begin{equation*}
\vec{C}^{v_1} = \vec{C}^{v_4} = \begin{bmatrix}
   10\\
    5 \\
   0
\end{bmatrix} , \quad
\vec{P}^{v_1} = \vec{P}^{v_4} =  \begin{bmatrix}
   0.9\\
   0.05 \\
   0.05
\end{bmatrix} ,\quad
\vec{C}^{v_2} =  \vec{C}^{v_3} =  \begin{bmatrix}
   15\\
   10\\
   5
\end{bmatrix} , \quad
\vec{P}^{v_2} = \vec{P}^{v_3} = \begin{bmatrix}
   0.85\\
   0.1\\
   0.05\\
\end{bmatrix}\,, 
\end{equation*}
\begin{equation*}
\vec{C}^{e_1} = \vec{C}^{e_4} = \begin{bmatrix}
   8\\
   4 \\
   0
\end{bmatrix} , \quad
\vec{P}^{e_1} = \vec{P}^{e_4} =  \begin{bmatrix}
   0.9\\
   0.05 \\
   0.05
\end{bmatrix} ,\quad
\vec{C}^{e_2} =  \vec{C}^{e_3} =  \begin{bmatrix}
   4\\
   2
\end{bmatrix} , \quad
\vec{P}^{e_2} = \vec{P}^{e_3} = \begin{bmatrix}
   0.95\\
   0.05\\
\end{bmatrix}\,.
\end{equation*}
Implicitly, $P_{dis} = 0.8$. We then provide three possible flow matrices that conforms the node and link capacity constraints mentioned in Assumption~\ref{assum:flow} in the scenarios when the network is undisturbed, when node $v_4$ is disturbed such that $c_{v_4} = c_{v_4,1} = 5$, and when link $e_3$ is disturbed such that $c_{e_3} = c_{e_3,1} = 2$ as below,
\begin{equation*}
X(0,0) =   \begin{bmatrix}
    5 & 3 & 0 \\
    0 & 2 & 0 \\
    0 & 3 & 0 \\
    0 & 2 & 5 
\end{bmatrix}, \quad
X(v_4, 1) =   \begin{bmatrix}
    5 & 3 & 0 \\
    0 & 1 & 0 \\
    0 & 3 & 0 \\
    0 & 1 & 1 
\end{bmatrix},\quad
X(e_3,1) =  \begin{bmatrix}
    5 & 2 & 0 \\
    0 & 3 & 0 \\
    0 & 2 & 0 \\
    0 & 3 & 5 
\end{bmatrix} \,.
\end{equation*}
According to~\eqref{eq:D_constr}--\eqref{eq:flow_cons}, we can infer that
\begin{equation*}
D_1(0,0) =   \begin{bmatrix}
    0 & 0 & 0 \\
    5 & 0 & 0 \\
    0 & 0 & 0 \\
    0 & 5 & 5 
\end{bmatrix}, \quad
D_1(v_4, 1) =   \begin{bmatrix}
    0 & 0 & 0 \\
    5 & 0 & 0 \\
    0 & 0 & 0 \\
    0 & 4 & 1 
\end{bmatrix},\quad
D_1(e_3,1) =  \begin{bmatrix}
    0 & 0 & 0 \\
    5 & 0 & 0 \\
    0 & 0 & 0 \\
    0 & 5 & 5 
\end{bmatrix} \,.
\end{equation*}
\end{example}

We then define the throughput of a momentary network as the maximal total amount of flights that the network can accommodate per unit of time while satisfying~\eqref{eq:n2D}--\eqref{eq:flow_cons} and the Assumption~\ref{assum:flow}. 
\begin{definition}[Throughput]
Given a risk-aware air transportation network $\mc N  =(\mc G,P_{dis},C^{\mc E},C^{\mc V},P^{\mc E},P^{\mc V})$, where $\mc G = (\mc V, \mc E)$, we let $i = \{0\}\cup \mc E \cup \mc V$, and let $k = 0$ if $i=0$, $k =1,\cdots, k_{i}$ if $ i \in \mc E \cup \mc V$. As a result, $\mc N(i,k)$ is a disturbed moment of $\mc N$ if $i \in \mc E \cup \mc V$ and $ k = 1,\cdots, k_{i}$, an undisturbed moment of $\mc N$ if $i =0$ and $ k = 0$.

The \emph{throughput} of the momentary air transportation network, $\mc N(i,k)= (\mc G,C^{\mc E}(i,k),C^{\mc V}(i,k))$, denoted as $S^{\ast}(\mc N(i,k))$, is the maximal amount of demands that the network can fulfill per unit of time while satisfying flow reservation law and the node and link capacity constraints, i.e., 
% \begin{align}
%  \label{eq:tp_mat}   
%  \begin{split}
%  S^{\ast}(\mc N(i,k)) = \max_{X,D_1,D_2}  \1^T_{ N_S}& D_1^T \1_{ N_\mc V}\\
%     \text{s.t. }\quad  EX &= D_1+D_2,
%     \quad X \geq 0\\
%      X\1_{ N_S} &\leq C^{\mc E}(i,k),
%     \quad  \lvert  E\rvert X\1_{ N_S} \leq C^{\mc V}(i,k)\\
%     D_1^T \1_{N_\mc V} &= -D_2^T \1_{ N_\mc V},
%     \quad D_1 \leq  \Delta_1 M ,
%     \quad -D_2 \leq  -\Delta_2 M, 
%     \quad D_1 \geq 0 ,
%     \quad D_2 \leq 0 \,.
%  \end{split}
% \end{align}
% We claim $S^{\ast}(\mc N(0,0))$ as the \emph{undisturbed throughput} of $\mc N$ and $S^{\ast}(\mc N(ev,k))$ as the \emph{disturbed throughput} of $\mc N$ with moment $\mc N(ev,k)$ for any $ev \in \mc E \cup \mc V$ and $k =1,\cdots, k_{ev}$.
\begin{align}
 \label{eq:tp_mat}   
 \begin{split}
 S^{\ast}(\mc N(i,k)) &= \max_{X(i,k),D_1(i,k),D_2(i,k)}  \1^T_{ N_S} D_1(i,k)^T \1_{ N_\mc V}\\
    \text{s.t. }\quad  &EX(i,k) = D_1(i,k)+D_2(i,k),
    \quad X(i,k) \geq 0\\
    & X(i,k)\1_{ N_S} \leq C^{\mc E}(i,k),
    \quad  E_+ X(i,k)\1_{ N_S} \leq C^{\mc V}(i,k)\\
    & D_1(i,k)^T \1_{N_\mc V} = -D_2(i,k)^T \1_{ N_\mc V}\\
    & D_1(i,k) \leq  \Delta_1 M ,
    \quad -D_2(i,k) \leq  -\Delta_2 M, 
    \quad D_1(i,k) \geq 0 ,
    \quad D_2(i,k) \leq 0 \,.
 \end{split}
\end{align}
We claim $S^{\ast}(\mc N(0,0))$ as the \emph{undisturbed throughput} of $\mc N$ and $S^{\ast}(\mc N(ev,k))$ as the \emph{disturbed throughput} of $\mc N$ with moment $\mc N(ev,k)$ for any $ev \in \mc E \cup \mc V$ and $k =1,\cdots, k_{ev}$.
\end{definition}

The objective function of~\eqref{eq:tp_mat} is the largest possible sum of fulfilled O-D demands in the network, as computed in~\eqref{eq:n2D}. The constraints depict the flow conservation law as in~\eqref{eq:flow_cons}, the node capacity constraints, the link capacity constraints, and the constraints for demand matrices in~\eqref{eq:D_constr}.

As the linear program~\eqref{eq:tp_mat} is feasible and bounded, an optimal solution always exists. The theorem below then provides a necessary and sufficient condition to the optimal solution of~\eqref{eq:tp_mat}.
\begin{thm}
\label{thm:LP_sol}
$X(i,k), D_1(i,k), D_2(i,k)$ is an optimal solution of~\eqref{eq:tp_mat} if and only if there exists $U_1(i,k)\in \mathbb{R}^{{N_\mc V} \times {N_S}}$, $U_2(i,k)\in \mathbb{R}^{{N_\mc E} \times {N_S}}$, $U_3(i,k), U_4(i,k), U_5(i,k), U_6(i,k)\in \mathbb{R}^{{ N_\mc V} \times {N_S}}$, $\vec{h}_1(i,k)\in \mathbb{R}^{N_\mc E}$, $\vec{h}_2(i,k) \in \mathbb{R}^{N_\mc V}$, and $\vec{h}_3(i,k)\in \mathbb{R}^{N_S}$ such that the following constraints are held:
\begin{align}
\label{eq:primal_1}    & EX(i,k) = D_1(i,k)+D_2(i,k), \quad  
    X(i,k) \geq 0,\quad
    X(i,k)\1_{N_S} \leq C^{\mc E}(i,k),\quad 
    E_+ X\1_{N_S} \leq C^{\mc V}(i,k)\\
\label{eq:primal_2}    & D_1(i,k)^T \1_{ N_\mc V} = -D_2(i,k)^T \1_{ N_\mc V},  
    D_1(i,k) \leq  \Delta_1 M, 
    -D_2(i,k) \leq  -\Delta_2 M, 
    D_1(i,k) \geq 0, 
    D_2(i,k) \leq 0 \\
\label{eq:dual_1}    & E^T U_1(i,k) - U_2(i,k) + \big(\vec{h}_1(i,k) + {E_+}^T \vec{h}_2(i,k)\big)\1^T_{ N_S} =0\\
\label{eq:dual_2}    & -\1_{ N_\mc V}\1^T_{ N_S} + \1_{ N_\mc V}\vec{h}_3(i,k)^T - U_1(i,k) + U_3(i,k) - U_5(i,k) =0\\
\label{eq:dual_3}    & \1_{N_\mc V}\vec{h}_3(i,k)^T - U_1(i,k) - U_4(i,k) + U_6(i,k) =0\\
\label{eq:dual_4}    & U_2(i,k),\cdots, U_6(i,k) \geq 0, \quad \vec{h}_1(i,k), \vec{h}_2(i,k) \geq 0 \\
\label{eq:zero_duality_gap}    & \1^T_{N_S} D_1(i,k)^T \1_{ N_\mc V} = \vec{h}_1(i,k)^T C^{\mc E} +\vec{h}_2(i,k)^T C^{\mc V} + tr(U_3(i,k)^T \Delta_1 M) - tr(U_4(i,k)^T \Delta_2 M)\,.
\end{align}
% The condition~\eqref{eq:zero_duality_gap} can be substituted with the following conditions:
% \begin{equation}\label{eq:complementary_slackness}
%   \begin{split}
%     & tr(X^T U_2) = 0, \quad
%     \vec{h}^T_1 X \1_{ \lvert \mc J\rvert} = \vec{h}^T_1 C^{\mc E},\quad
%     \vec{h}^T_2 { \lvert E \rvert} X \1_{ \lvert \mc J\rvert} = \vec{h}^T_2 C^{\mc V}, \\
%     & tr(U_3^T D_1) = tr(U_3^T \Delta_1 M), \quad
%     tr(U_4^T D_2) = tr(U_4^T \Delta_2 M), \quad
%     tr(U_5^T D_1) =0, \quad
%     tr(U_6^T D_2) =0
%   \end{split}
% \end{equation}
\end{thm}
The proof of Theorem~\ref{thm:LP_sol} is presented in~\ref{appendix:LP_solution_pf}.

% To avoid the cumbersome notations, we sometimes drop the notation $(i,k)$ when it is clear from the context.

%The flow conservation law~\eqref{eq:tp_flowcons1}--\eqref{eq:tp_flowcons2} 

In the rest of the paper, we refer to the risk-aware air transportation network defined in Definition~\ref{def:network_ori} as the \emph{original network}, and all nodes and links in the original network as the original nodes and links.

\subsection{Extended Network Model}
\label{sec:extend_model}
The capacity constraints in~\eqref{eq:tp_mat} implies that, disruptions to the network can sometimes lead to reduced throughput and hence degraded performance of the network. Reserve capacity can manage the overflow demands and traffic in the original network during disruptions to certain extent. There are several potential forms of reserve capacity in the network, such as additional under-utilized space and infrastructure and reserved ``cushion'' space in the original network for contingency management. Here we focus on the addition of reserve infrastructure to the network by constructing a set of backup nodes/vertiports with relatively small capacities. We use such backup nodes and associated backup links for overflow traffic management when the original network is disturbed.

%To increase the throughput of the disturbed network during the events of disruption, we then.

To figure out where to optimally construct the backup vertiports, we postulate the existence of a set of backup locations (nodes) and their resulting new flight corridors (links). We then need to decide where to construct the backup vertiports among the candidate locations and their corresponding capacities.

We first introduce some basic concepts associated with the extended network and reserve capacity: backup nodes and links, and extended graph. 

\noindent\underline{Backup Nodes and Links:}
We denote the set of all available backup nodes $\mc V^b = \{v_{N_\mc V+1}, \cdots,v_{N_\mc V+N_{\mc V^b}}\}$ and the set of all available backup links $\mc E^b = \{v_{N_\mc E+1}, \cdots,v_{N_\mc E+N_{\mc E^b}}\}$ for some $N_{\mc V^b}, N_{\mc E^b} \in \mathbb{N}_{>0}$. Further, we let the index sets be $\mc J_{\mc V^b} = \{N_\mc V+1, \cdots, N_\mc V+N_{\mc V^b}\}$ and $\mc J_{\mc E^b} = \{N_\mc E+1, \cdots, N_\mc E+N_{\mc E^b}\}$.
%, the chosen set of backup locations (resp., corridors) as $\mc V^c$ (resp., $\mc E^c$), where $\mc V^c\subseteq \mc V^b$ and $\mc E^c = \{e \in \mc E^b \mid \tau_e = v\in \mc V,   \sigma_e = v\in \mc V^c\} \cup \{e \in \mc E^b \mid \sigma_e = v\in \mc V, \tau_e = v\in \mc V^c\}$. 

We denote the capacity of a candidate backup node $v \in \mc V^b$ as $C_{v} \geq 0$, where $C_{v} = 0$ implies that candidate location $v$ has no backup vertiport and $C_{v} > 0$ implies that candidate $v$ has a vertiport with capacity $C_{v}$. We denote the capacity of a backup link $e \in \mc E^b$ as $C_e$. The vector of capacity for all backup nodes (resp., links) is therefore $C^{\mc V^b} = \{C_{v_i}\}_{i\in \mc J_{\mc V^b}}$ (resp., $C^{\mc E^b} = \{C_{e_i}\}_{i\in \mc J_{\mc E^b}}$). Moreover, we make the following assumptions on the backup nodes and links.
%\textcolor{red}{revise this sentence later } Notice that we are not constructing a vertiport on each backup node (location), and hence some of the nodes and the corresponding links are having zero capacity.

\begin{assum}[Backup Nodes and Links]\
\label{assum:backup}
\begin{enumerate}
    \item For each backup node $v \in\mc V^b$, there exists $v_0 \in \mc V$, such that $(v,v_0),(v_0,v) \in \mc E^b$.
%    \item Given $v_1 \in\mc V$, $v_2 \in \mc V^b$, or $v_1 \in\mc V^b$, $v_2 \in \mc V$, then $(v_1, v_2) \in \mc E^b$ if and only if $(v_2, v_1) \in \mc E^b$.
    \item For each backup link $e = (\tau_e, \sigma_e) \in \mc E^b$, $(\sigma_e,\tau_e) \in \mc E^b$, only one node among $\tau_e, \sigma_e$ is a backup node, i.e., either $\tau_e \in\mc V$, $\sigma_e \in \mc V^b$, or $\tau_e \in\mc V^b$, $\sigma_e \in \mc V$.
    \item The capacity of a backup link aligns with and is always equal to the (only) backup node it connects to, i.e., for all $e\in \mc E^b$, $C_e = C_{\sigma_e}$ if ${\sigma_e} \in \mc V^b$ and $C_e = C_{\tau_e}$ if ${\tau_e} \in \mc V^b$.
    \item A backup link or node $ev \in \mc E^b \cup \mc V^b$ will never be disturbed, so that its capacity is always $C_{ev}$.
    % \item The demand will not emerge from the backup nodes in $\mc V^b$.
\end{enumerate}
\end{assum}

The approaches for identifying qualified candidate options for backup nodes and links vary according to the environment and infrastructure requirements and alternative travel route standards. The approach we adopt in the case study will be discussed in detail in Section~\ref{sec:casestudy}.

\noindent\underline{Extended Graph:}
We then let the set of all nodes be $\mc V^t = \mc V \cup \mc V^b$ and the set of all links be $\mc E^t = \mc E \cup \mc E^b$. 
As a result, the total number of nodes and links are $N_{\mc V^t} = N_\mc V+N_{\mc V^b}$ and $N_{\mc E^t} = N_\mc E+N_{\mc E^b}$, and the index set for nodes and links in the extended model are $\mc J_{\mc V^t} = \mc J_{\mc V}\cup \mc J_{\mc V^b}$ and $\mc J_{\mc E^t} = \mc J_{\mc E}\cup \mc J_{\mc E^b}$. We define the graph that contains all original and backup nodes and links, $\mc G^t = (\mc V^t, \mc E^t)$, as an \emph{extended graph} of an original network with graph $(\mc V, \mc E)$.

The backup nodes and links are solely for temporary usage, and the flights are allowed to travel through the backup options only if they are affected by the disruption on the original network. We therefore define the extended network based on a disturbed moment of network as follows.
\begin{definition} [Extended Air Transportation Network]\
    \label{def:extended_network}
% Given a disturbed air transportation network $\mc N = (\mc G,P_{dis},C^{\mc E},C^{\mc V},P^{\mc E},P^{\mc V})$ and a realization $\mc N(ev,k)$, where $\mc G = (\mc V, \mc E)$ where $ev \in \mc V \cup \mc E$.
% An \emph{extended network} of $\mc N$ is a tuple $\mc N^t = (\mc N(ev,k), \mc G^t,C^{\mc V_b})$, where $\mc G^t$ and $C^{\mc V_b}$ are the extended network graph and a choice of the extended network capacity defined as above. 
Given an air transportation network $\mc N =(\mc G,P_{dis},\allowbreak C^{\mc E}, C^{\mc V},P^{\mc E},P^{\mc V})$, where $\mc G = (\mc V, \mc E)$, % $P_{dis}$, $C^{\mc E}$, $C^{\mc V}$, $P^{\mc E}$, $P^{\mc V}$ are defined as in Section~\ref{sec:original_model},
and a disturbed moment of $\mc N$, $\mc N(ev,k) = (\mc G, \vec{C}^{\mc E}(ev,k),\vec{C}^{\mc V}(ev,k))$, where $ev \in \mc V \cup \mc E$ and $k=1,\cdots,k_{ev}$.
An \emph{extended air transportation network} of $\mc N(ev,k)$ is a tuple $\mc N^{ext}(ev,k) = (\mc N(ev,k), \mc V^b,  \mc E^b,C^{\mc V_b})$, where $\mc V^b,  \mc E^b$ and $C^{\mc V_b}$ are the sets of backup nodes and links, and the capacity vector of backup nodes defined as above. 
\end{definition}

Since the flights traveling through an undisturbed moment of the network will not use the backup nodes and links, we further assign $\mc N^{ext}(0,0) := \mc N(0,0)$.

\begin{assum}[Rerouting under Disruption]\
\label{assum:reroute}
Given an original air transportation network $\mc N$ with network graph $\mc G = (\mc V, \mc E)$, the disturbed network $\mc N(ev,k)= (\mc G, \vec{C}^{\mc E}(ev,k),\vec{C}^{\mc V}(ev,k))$, and the extended network $\mc N^{ext}(ev,k) = (\mc N(ev,k), \mc V^b, \mc E^b,C^{\mc V_b})$, the flight routing through the network must follow the rules below.

%Given an extended air transportation network $\mc N^{ext}(ev,k) = (\mc N(ev,k), \mc V^b,  \mc E^b,C^{\mc V_b})$ of a disturbed moment of network, $\mc N(ev,k)= (\mc G, \vec{C}^{\mc E}(ev,k),\vec{C}^{\mc V}(ev,k))$, where $ev \in \mc V \cup \mc E$ and $k=1,\cdots,k_{ev}$, the flights traveling through the extended network must follow the rules below.
%original network  $\mc N = (\mc G,P_{dis},C^{\mc E},C^{\mc V},P^{\mc E},P^{\mc V})$, where $\mc G = (\mc V, \mc E)$ and the backup nodes and arcs $\mc V^b$ and $\mc E^b$, the flights will follow the following rules:
    \begin{enumerate}
    \item If $e \in \mc E$ is disturbed, in addition to the links and nodes in the original network $\mc N$, the affected flights may travel through a backup node $v \in \mc V^b$ in close proximity and its associated connecting links $(\tau_e, v), (v,\sigma_e) \in \mc E^b$, and $(\tau_e, v), (v,\sigma_e)$ qualifies to provide an alternative path for $e$ only if the ratio $r = (d_{\tau_e, v} + d_{v,\sigma_e})/d_{\tau_e, \sigma_e}$ is within $[\underline{d}_1, \overline{d}_2]$ for some $\underline{d}_1, \overline{d}_2 \in \mathbb{R}_{>0}$ and $1<\underline{d}_1 \leq \overline{d}_2 $, where $d_{\cdot,\cdot}$ represents the Euclidean distance between the physical locations of the two vertiports. % The flights still need to stop at each vertiport along its route.
    \item If $v \in \mc V$ is disturbed, in addition to the links and nodes in the original network $\mc N$, the affected flights may travel through a nearby backup node $v' \in \mc V^b$ that satisfies $(v, v') \in \mc E^b$, and the connecting links $(v, v'), (v', v)$. % The flights traveling through the node $v$ but not any backup node must stop at $v$ before continue its trip, while the flights traveling through the node $v$ and any connecting backup node can stop at one and only one node among $v$ and all backup nodes $v$ connects to.
    \item If $v \in \mc V$ is disturbed, the flights for O-D pair $(v,v_1)$ for any $v_1 \neq v \in \mc V$ may take off from any backup node $v' \in \mc V^b$ that satisfies $(v, v') \in \mc E^b$ instead of $v$.
    \item If $v \in \mc V$ is disturbed, the flights for O-D pair $(v_2,v)$ for any $v_2 \neq v \in \mc V$ may land at any backup node $v' \in \mc V^b$ that satisfies $(v, v') \in \mc E^b$ instead of $v$.
    \item If the network is not disturbed, then the flights will not travel through any backup node or link.
%    \item If $e \in \mc E$ is disturbed, $\tau_e$ qualifies to provide an alternative path for $A \rightarrow B$ if the ratio $r = (d_{A \rightarrow C} + d_{C \rightarrow B})/d_{A \rightarrow B}$ is within $[1.02, 1.5]$ where $d_{\cdot \rightarrow \cdot }$ gives the distance between the nodes.
\end{enumerate}

\end{assum}

Next, we consider the throughput of the extended networks $\mc N^{ext}(ev,k) = (\mc N(ev,k), \mc V^b,  \mc E^b,C^{\mc V_b})$, where $\mc N(ev,k) = (\mc G,C^{\mc E}(ev,k),C^{\mc V}(ev,k))$ and $\mc G = (\mc V, \mc E)$. 
Similar to that in the original case in Section~\ref{sec:original_model}, we first introduce the incidence matrices, demand indicating matrices, and capacity matrices specifically for the extended network model. We will discuss the cases when the disturbed element $ev$ is a link or a node separately, i.e., when $ev = e^a\in \mc E$ or when $ev = v^a \in \mc V$. 

\noindent \underline{Case 1: A Link is Disturbed ($ev = e^a\in \mc E$).}
When a link $e^a \in \mc E^b$ is disturbed such that $c_{e^a} = c_{e^a,k}$ for any $k=1,\cdots,k_{e^a}$, we define an $N_{\mc V^t} \times N_{\mc E^t}$ matrix $E'(e^a)$ such that
\begin{align}
\label{eq:E_e}
[E'(e^a)]_{ij} = \begin{cases}  [E]_{ij} &\text{ if } i \in \mc J_{\mc V}, j \in \mc J_{\mc E}\\
                                -1 &\text{ if }  j \in \mc J_{\mc E^b}, v_i = \tau_{e^a} = \tau_{e_j}\\
                                1 &\text{ if }  j \in \mc J_{\mc E^b}, \tau_{e_j} = \tau_{e^a}, v_i = \sigma_{e_j}\\
                             1 &\text{ if } j \in \mc J_{\mc E^b}, v_i =\sigma_{e^a} = \sigma_{e_j}\\
                             -1 &\text{ if } j \in \mc J_{\mc E^b}, \sigma_{e_j} =\sigma_{e^a}, v_i = \tau_{e_j} \\
                             0 &\text{ otherwise} \,.
    \end{cases}
\end{align}
$E'(e^a)$ defined in~\eqref{eq:E_e} is a reformulated incidence matrix of the extended graph $\mc G^t$ when $e^a$ is disturbed, where we only preserve the columns for links used for rerouting according to Assumption~\ref{assum:reroute} in addition to those in the original graph, while leaving the rest of columns zero. 
We then divide $E(e^a)$ into the block matrices $E$, $\mathbf{0}^{{N_{\mc V^b}} \times {N_\mc E }}$, $E^{b,1}(e^a) \in \mathbb{R}^{{N_{\mc V}} \times {N_{\mc E^b} }}$, and $E^{b,2}(e^a)\in \mathbb{R}^{{N_{\mc V^b}} \times {N_{\mc E^b} }}$ such that
\begin{equation}
\label{eq:E_e_block}   
E'(e^a) =
\left(
  \begin{array}{c|c}
   E & E^{b,1}(e^a) \\
  \hline
  \mathbf{0}^{{N_{\mc V^b}} \times {N_\mc E }} & E^{b,2}(e^a)
\end{array} \right) \,.
\end{equation}

%, where $E^{b,1}(e^a)$ is an ${{N_{\mc V}} \times {N_{\mc E^b} }}$ matrix such that $[E^{b,1}(e^a)]_{ij} = [E(e^a)]_{ij}$ for all $i \in \mc J_{\mc V},j\in \mc J_{\mc E^b}$ and $E^{b,2}(e^a)$ is an ${{N_{\mc V^b}} \times {N_{\mc E^b} }}$ matrix such that $[E^{b,2}(e^a)]_{ij} = [E(e^a)]_{ij}$ for all $i \in \mc J_{\mc V^b},j\in \mc J_{\mc E^b}$.

According to Assumption~\ref{assum:reroute}, one can compute the capacity vector for all rerouting backup links flights may travel through when $e^a$ is disturbed, denoted as $C^{\mc E^b}(e^a)$, with the capacity vector for all backup nodes $C^{\mc V^b}$ and $E^{b,2}(e^a)$ defined in~\eqref{eq:E_e_block}, such that 
\begin{equation}
\label{eq:cap_e_add}
    C^{\mc E^b}(e^a) = \Big[{ E^{b,2}(e^a)_+} \odot \big(C^{\mc V^b} \1^T_{N_{\mc E^b} } \big)\Big]^T \cdot \1_{N_{\mc V^b}} \,.
\end{equation}
%The capacity vector for all nodes, denoted as $C^{\mc V^t}(e^a)\in \mathbb{R}^{N_{\mc V^t}}$ is concatenated by $C^{\mc V}_0$ and $C^{\mc V^b}$, and the capacity vector for all links, $C^{\mc E^t}(e^a)\in \mathbb{R}^{N_{\mc E^t}}$, is concatenated by $C^{\mc E}(e^a,k)$ and $C^{\mc E^b}(e^a,k)$, i.e.,
We then construct the capacity vectors for all nodes and links, $C^{\mc V^t}(e^a)\in \mathbb{R}^{N_{\mc V^t}}$ and $C^{\mc E^t}(e^a)\in \mathbb{R}^{N_{\mc E^t}}$, by concatenating the capacity vectors of the original nodes and links and those of the backups:
\begin{equation}
    \label{eq:cap_e_ext}
    C^{\mc V^t}(e^a,k) =         
\left(
  \begin{array}{c}
  C^{\mc V}_0  \\
  \hline
  C^{\mc V^b} 
\end{array} \right), \quad
C^{\mc E^t}(e^a,k) =\left(
  \begin{array}{c}
  C^{\mc E}(e^a,k)  \\
  \hline
  C^{\mc E^b}(e^a) 
\end{array} \right) \,.
\end{equation}

The destination and origin indicating matrices, $\Delta_1'(e^a)$ and $\Delta_2'(e^a)$, are now matrices of size $N_{\mc V^t}  \times {N_S} $, such that
\begin{equation}
\label{eq:Delta_e}   
\Delta_1'(e^a) = \left(
  \begin{array}{c}
  \Delta_1  \\
  \hline
  \mathbf{0}^{N_{\mc V^b} \times {N_S} } 
\end{array} \right),
\Delta_2'(e^a) = \left(
  \begin{array}{c}
  \Delta_2  \\
  \hline
  \mathbf{0}^{{N_{\mc V^b}} \times {N_S} } 
\end{array} \right)
\end{equation}
%, and hence both $\Delta_1(e^a)$ and $\Delta_2(e^a)$ are matrices of size $N_{\mc V^t}  \times {N_S} $.
We denote the amount of demand that the extended network $\mc N^{ext}(e^a,k)$ will fulfill for any O-D pair $s_{\ell} \in S$ per unit of time as $n'_{\ell}(e^a,k) \geq 0 $ for all $\ell \in \mc J_S$.
We then define $D_1'(e^a,k), D_2'(e^a,k) \in \mathbb{R}^{N_{\mc V^t} \times N_S}$ such that $[D_1'(e^a,k)]_{i\ell} = [\Delta_1'(e^a)]_{i\ell}\cdot n_\ell'(e^a,k)$ and $[D_2'(e^a,k)]_{i\ell} = [\Delta_2'(e^a)]_{i\ell}\cdot n_\ell'(e^a,k)$, and hence
\begin{equation}
\label{eq:n2D_e}
    \sum_{\ell =1}^{N_S} n_{\ell}'(e^a,k)= \1^T_{ N_S} D_1'(e^a,k)^T \1_{ N_\mc V}\, ,
\end{equation}
while given some large enough real number $M$, $D_1'(e^a,k)$, $D_2'(e^a,k)$, $\Delta_1'(e^a)$, and $\Delta_2'(e^a)$ must satisfy
\begin{align}
\label{eq:D_constr_e}
\begin{split}    
    & D_1'(e^a,k)^T \1_{N_{\mc V^t}} = -D_2'(e^a,k)^T \1_{ N_{\mc V^t}},\\
    & D_1'(e^a,k) \leq  \Delta_1'(e^a) M,
    \quad  -D_2'(e^a,k) \leq  -\Delta_2'(e^a) M, 
    \quad D_1'(e^a,k) \geq 0,
    \quad  D_2'(e^a,k) \leq 0 \,.
\end{split}
\end{align}

We let $X(e^a,k) \in \mathbb{R}^{N_{\mc E^t} \times N_S}$ be the flow matrix for all links and demands, with each entry indicating the amount of flights assigned to any demand traveling through any link per unit of time. Following the flow conservation law, we require
\begin{equation}
\label{eq:flow_cons_e}
    E'(e^a)X'(e^a,k) = D_1'(e^a,k)+D_2'(e^a,k), 
    \quad X'(e^a,k) \geq 0 \,.
\end{equation}

\noindent \underline{Case 2: Node is Disturbed ($ev = v^a \in \mc V$).}
When a node $v^a\in \mc V$ is disturbed, and $c_{v^a} = c_{v^a,k}$ for any $k=1,\cdots,k_{v^a}$, any backup node connected directly to $v^a$ -- the backup node connected to $v^a$ by a single backup link -- works as a complement of $v^a$ according to Assumption~\ref{assum:reroute}. As a result, the extended network, $\mc N^{ext}(v^a,k)$, works the same as adding additional capacity to the disturbed moment of network, $\mc N(v^a,k)$, yielding the same throughput. % as we will always able to distribute the flights flying toward the additional capacity to the backup nodes and links.

The additional capacity is in fact the sum of capacities of all backup nodes that the flights may reroute to when $v^a$ is disturbed. To compute this sum, we first define $\Delta^{adj}$ as an $N_\mc V \times  N_{\mc V^b}$ matrix indicating all backup nodes adjacent to any original node in $\mc V$, that is, 
\begin{align}
    [\Delta^{adj}]_{i,j} = \begin{cases}  1 &\text{ if } (v_i,v_{N_\mc V+j}) \in \mc E^b\\
                            0 &\text{ otherwise }\,.              
    \end{cases}
\end{align}
The additional capacity that the node $v^a$ obtains from the adjacent nodes, denoted as $C^{\mc V^b}(v^a) \in \mathbb{R}^{N_{\mc V}}$, can then be computed as
\begin{equation}
\label{eq:cap_v_add}
%    C^{\mc V^b}(v^a) = \vec{e}_{N_\mc V}(v^a) \odot (\Delta^{adj} \cdot C^{\mc V^b} ) \,,
    C^{\mc V^b}(v^a) = \vec{e}_{N_\mc V}(i) \odot (\Delta^{adj} \cdot C^{\mc V^b} ) \,,
\end{equation}  
where $i \in \mc J_{\mc V}$ and $v^a = v_i$, indicating that only the disturbed node may get the additional capacity.
% where $\vec{e}_{N_\mc V}(v^a)$ is the elementary vector of length $N_\mc V$ with $1$ on the $i$'th element if $v^a = v_i$ for some $i \in \mc J_{\mc V}$ and $0$ elsewhere, indicating that only the disturbed node may get the additional capacity.

%According to Assumption~\ref{assum:reroute}, flights can only travel through the backup nodes adjacent to the disturbed node $v^a$, i.e., the backup nodes connected to $v^a$ by a single backup link, and the connecting backup links. By Assumption~\ref{assum:backup}, we can infer that the throughput of the extended network will be the same as the momentary network $\mc N(v^a,k)$ with additional capacity at $v_a$, as we will always able to distribute the flights flying toward the additional capacity to the backup nodes and links. 
We then assign the incidence matrix, destination and origin indicating matrices, the capacity vectors for all nodes and links as follows, 
% $E(v^a) = E$, $\Delta_1(v^a) = \Delta_1$, $\Delta_2(v^a) = \Delta_2$, 
\begin{align}
\begin{split}
\label{eq:ext_v}    
 & E'(v^a) = E, \quad \Delta_1'(v^a) = \Delta_1, \quad  \Delta_2'(v^a) = \Delta_2, \\
 & C^{\mc E^t}(v^a,k) = C^{\mc E}_0\\
 & C^{\mc V^t}(v^a,k) = C^{\mc V}(v^a,k) + C^{\mc V^b}(v^a) \,.
\end{split}
\end{align}

Mirroring that in case 1, we denote the amount of demand the extended network $\mc N^{ext}(v^a,k)$ will fulfill for any O-D pair $s_{\ell} \in S$ per unit of time as $n'_{\ell}(v^a,k) \geq 0 $ for all $\ell \in \mc J_S$ and define $D_1'(v^a,k), D_2'(v^a,k) \in \mathbb{R}^{N_{\mc V^t} \times N_S}$ such that $[D_1'(v^a,k)]_{i\ell} = [\Delta_1'(v^a)]_{i\ell}\cdot n_\ell'(v^a,k)$ and $[D_2'(v^a,k)]_{i\ell} = [\Delta_2'(v^a)]_{i\ell}\cdot n_\ell'(v^a,k)$, and hence
\begin{equation}
\label{eq:n2D_v}
    \sum_{\ell =1}^{N_S} n_{\ell}'(v^a,k)= \1^T_{ N_S} D_1'(v^a,k)^T \1_{ N_\mc V}\, ,
\end{equation}
while given some large enough real number $M$, $D_1'(v^a,k)$, $D_2'(v^a,k)$, $\Delta_1'(v^a)$, and $\Delta_2'(v^a)$ must satisfy
\begin{align}
\label{eq:D_constr_v}
\begin{split}    
    & D_1'(v^a,k)^T \1_{N_{\mc V}} = -D_2'(v^a,k)^T \1_{ N_{\mc V}},\\
    & D_1'(v^a,k) \leq  \Delta_1'(v^a) M,
    \quad  -D_2'(v^a,k) \leq  -\Delta_2'(v^a) M, 
    \quad D_1'(v^a,k) \geq 0,
    \quad  D_2'(v^a,k) \leq 0 \,.
\end{split}
\end{align}

Again, we let $X'(v^a,k) \in \mathbb{R}^{N_{\mc E} \times N_S}$ be the flow matrix for all links and demands for the extended network, with each entry indicating the amount of flights assigned to any demand traveling through any link per unit of time. The flow conservation law requires that
\begin{equation}
\label{eq:flow_cons_v}
    E'(v^a)X'(v^a,k) = D_1'(v^a,k)+D_2'(v^a,k), 
    \quad X'(v^a,k) \geq 0 \,.
\end{equation}
Notice that, although we omit the extended links in this formulation of the extended network when $v^a$ is disturbed, we are able to infer the amount of flows that travel through the backup links from the flow matrix $X'(v^a,k)$, while such mapping is not unique.

\begin{example}
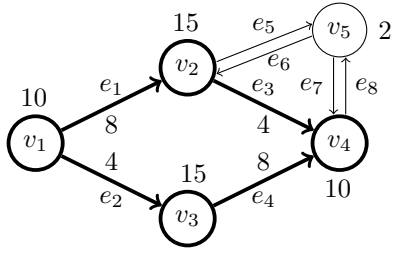
\begin{figure}
    \centering    
    \begin{tikzpicture}
    
    \node[draw, line width=0.5mm, circle] (1) at (0, 0) {$v_1$};
    \node[text width=1cm] at (0.3,0.6) {$10$};
    
    \node[draw,line width=0.5mm, circle] (2) at (2, 1) {$v_2$};
    \node[text width=1cm] at (2.3,1.6) {$15$};

    \node[draw,line width=0.5mm, circle] (3) at (2, -1) {$v_3$};
    \node[text width=1cm] at (2.4,-0.4) {$15$};

    \node[draw,line width=0.5mm, circle] (4) at (4, 0) {$v_4$};
    \node[text width=1cm] at (4.3,-0.6) {$10$};

    \node[draw, circle] (5) at (4, 1.5) {$v_5$};
    \node at (4.6,1.5) {$2$};

    % \node[draw, circle] (6) at (0, -1.5) {$v_6$};
    % \node[text width=1cm] at (4.3,-0.6) {$10$};

    \draw[->, line width=0.5mm] (1) --  node[above] {$e_1$} node[below] {$8$}(2);
    \draw[->, line width=0.5mm] (1) --  node[below] {$e_2$} node[above] {$4$}(3);
    \draw[->, line width=0.5mm] (2) --  node[above] {$e_3$} node[below] {$4$}(4);
    \draw[->, line width=0.5mm] (3) --  node[below] {$e_4$} node[above] {$8$}(4);
     % \draw[->]  (0,0) -- (1,0);
    \draw[->] ([yshift = 0.8mm]2.east) -- node[above] {$e_5$} ([yshift = 0.8mm]5.west);
    \draw[->] ([yshift = -0.8mm]5.west) -- node[below, xshift=2mm, yshift=1.3mm] {$e_6$} ([yshift = -0.8mm]2.east);

    \draw[->] ([xshift = -0.8mm]5.south) -- node[left] {$e_7$} ([xshift = -0.8mm]4.north);
    \draw[->] ([xshift = 0.8mm]4.north) -- node[right] {$e_8$} ([xshift = 0.8mm]5.south);
    % \draw[->] (3) -- node[below,xshift=-2mm, yshift=1mm] {$[5,6]$} (5);
    % \draw[->] (4) -- node[below,xshift=3mm, yshift=1mm] {$[1,5]$} (5); 
    % \draw[->] (5) -- node[above,xshift=3mm, yshift=-1mm] {$[3,6]$} (7);
    \end{tikzpicture}
    \caption{The notional original air transportation network (bolded) with additional 1 backup node and 4 backup links.}
    \label{fig:ex_ext_network}
\end{figure}

We illustrate our concepts in the extended air transportation network model with Fig.~\ref{fig:ex_ext_network}, which adds a potential backup node ${v_5}$ to the original risk-aware network as in Fig.~\ref{fig:ex_ori_network} with $C_{v_5}=2$. We assume that $\mc V$, $\mc E$, $S$, $C^\mc V$, $C^\mc E$, $P^\mc V$, and $P^\mc E$ remains the same as those in the original network in Example~\ref{ex:ori}. In the extended network we have $\mc V^b = \{v_5\}$ and $\mc E^b = \{e_5,e_6,e_7,e_8\}$.

When node $v_4$ is disturbed such that $c_{v_4} = c_{v_4,1} = 5$, the flights are now allowed to travel through the node $v_5$ and the links $e_7, e_8$, and hence we have $E'(v_4) = E$, $\Delta_1'(v_4) = \Delta_1$, $\Delta_2'(v_4) = \Delta_2$, $C^{\mc E^t}(v_4,1) = C^{\mc E}_0$, $C^{\mc V^t}(v_4,1) = [10,15,15,7]^T$ according to~\eqref{eq:cap_v_add}--\eqref{eq:ext_v}.

When link $e_3$ is disturbed such that $c_{e_3} = c_{e_3,1} = 2$, the flights are now allowed to travel through the node $v_5$ and the links $e_5,e_7$, and according to~\eqref{eq:E_e}--\eqref{eq:Delta_e}, we can compute
\begin{equation*}
E'(e_3) =  \begin{bmatrix}
  -1 & -1 &  0 &  0 & 0 & 0 & 0 & 0  \\
   1 &  0 & -1 &  0 & -1 & 0 & 0 & 0  \\
   0 &  1 &  0 & -1 & 0 & 0 & 0 & 0  \\
   0 &  0 &  1 &  1 & 0 & 0 & 1 & 0  \\
   0 &  0 &  0 &  0 & 1 & 0 & -1 & 0  
\end{bmatrix},  \quad
\Delta_1'(e_3) =  \begin{bmatrix}
   0 & 0 & 0 \\
   1 & 0 & 0 \\
   0 & 0 & 0 \\
   0 & 1 & 1 \\
   0 & 0 & 0 
\end{bmatrix}, \quad
\Delta_2'(e_3) =  \begin{bmatrix}
   -1 & -1 & 0 \\
    0 &  0 &  0 \\
    0 &  0 & -1 \\
    0 &  0 & 0 \\
   0 & 0 & 0 
\end{bmatrix}\,,
\end{equation*}

\begin{equation*}
C^{\mc E^t}(e_3,1) = [8,4,2,8,2,2,2,2]^T, \quad  C^{\mc V^t}(e_3,1) = [10,15,15,10,2]^T \,.
\end{equation*}

We then provide two possible flow matrices and demand matrices that conforms to constraints~\eqref{eq:D_constr_e}--\eqref{eq:flow_cons_e}, \eqref{eq:D_constr_v}--\eqref{eq:flow_cons_v}, when node $v_4$ is disturbed and when link $e_3$ is disturbed as below,
\begin{equation*}
X'(v_4, 1) =   \begin{bmatrix}
    5 & 0 & 0 & 0\\
    3 & 2 & 3 & 2 \\
    0 & 0 & 0 & 2
\end{bmatrix}^T, \quad
X'(e_3,1) =  \begin{bmatrix}
    5 & 0 & 0 & 0 & 0 & 0 & 0 & 0\\
    3 & 3 & 2 & 3 & 1 & 0 & 1 & 0\\
    0 & 0 & 0 & 4 & 0 & 0 & 0 & 0
\end{bmatrix}^T \,.
\end{equation*}

\begin{equation*}
D_1'(v_4, 1) =   \begin{bmatrix}
    0 & 0 & 0 \\
    5 & 0 & 0 \\
    0 & 0 & 0 \\
    0 & 5 & 2 
\end{bmatrix},\quad
D_1'(e_3,1) =  \begin{bmatrix}
    0 & 0 & 0 \\
    5 & 0 & 0 \\
    0 & 0 & 0 \\
    0 & 6 & 4 \\
    0 & 0 & 0 
\end{bmatrix} \,.
\end{equation*}

\end{example}

% the capacity vector of all nodes will be the same as that of the original network except the capacity of the affected node $v^a$, as we will add the capacity of the backup nodes the flights are allowed to travel through according to Assumption~\ref{assum:reroute} to the capacity of $v^a$,
%  Then $C^{\mc V^t}(v^a)$ is a vector of length ${\lvert \mc V \rvert}$, and 
% \begin{equation}
%     C^{\mc V^t}(v^a) = C^{\mc V}(v^a) + C^{\mc V^b}(v^a) \,,
% \end{equation}
% where we let the 
% \begin{equation}
%     C^{\mc V^b}(v^a) = \vec{e}_{N_\mc V}(v^a) \odot (\Delta^{adj} \cdot C^{\mc V^b} ) \,,
% \end{equation}  
% where $\vec{e}_{N_\mc V}(v^a)$ is the elementary vector of length $N_\mc V$ with $1$ on the $i$'th element if $v^a = v_i$ for some $i \in \mc J_{\mc V}$ and $0$ elsewhere.

%We define the tuple $\mc N^{ext}(ev,k) = (\mc N(ev,k), \mc G^t, E(ev), \Delta_1, \Delta_2, C^{\mc E^t}(ev), C^{\mc V^t}(ev))$ as a \emph{Partial Extended Network}, where $\mc N(ev,k), \mc G^t, E(ev), \Delta_1, \Delta_2, C^{\mc E^t}(ev), C^{\mc V^t}(ev)$ are the disturbed network, extended network graph, partial incidence matrix, destination and origin indicator matrices, and link and node capacity matrices defined as above for all $ev\in \mc E \cup \mc V$ and $k=1,\cdots,k_{ev}$. Moreover, we let $\mc N^{ext}(0,0) = \mc N(0,0)$

Similar to the throughput computation for the disturbed or undisturbed moment of the original network in~\eqref{eq:tp_mat}, we follow the capacity and flow reservation constraints, and hence the throughput of $\mc N^{ext}(ev,k)$ can be computed with the optimization formulation below.
\begin{align}
 \label{eq:tp_ext}   
 \begin{split}
 S^{\ast}(\mc N^{ext}(ev,k))& = \max_{X'(ev,k),D_1'(ev,k),D_2'(ev,k)}  \1^T_{ N_S} D_1'(ev,k)^T \1_{ N_{\mc V,ev}}\\
    \text{s.t. }\quad  &E(ev)'X(ev,k)' = D_1'(ev,k)+D_2'(ev,k),
    \quad X'(ev,k) \geq 0\\
    & X'(ev,k)\1_{N_S} \leq C^{\mc E^t}(ev,k),
    \quad  E(ev)_+ X(ev,k) \1_{N_S} \leq C^{\mc V^t}(ev,k),\\
    % & D_1'(ev,k)^T \1_{N_\mc V^t} = -D_2'(ev,k)^T \1_{ N_\mc V^t} \quad \text{ if } ev \in \mc E,\\
    % & D_1'(ev,k)^T \1_{N_\mc V} = -D_2'(ev,k)^T \1_{ N_\mc V} \quad \text{ if } ev \in \mc V, \\
    & D_1'(ev,k)^T \1_{N_{\mc V,ev}} = -D_2'(ev,k)^T \1_{ N_{\mc V,ev}}, \\
    & D_1'(ev,k) \leq  \Delta_1'(ev) M ,
    \quad -D_2'(ev,k) \leq  -\Delta_2'(ev) M, 
    \quad D_1'(ev,k) \geq 0 ,
    \quad D_2'(ev,k) \leq 0 \,,
 \end{split}
\end{align}
where we let $\1_{ N_{\mc V,ev}} =\1_{ N_{\mc V^t}}$, $X'(ev,k) \in \mathbb{R}^{N_{\mc E^t} \times N_S}$, $D_1'(ev,k), D_1'(ev,k) \in \mathbb{R}^{N_{\mc V^t} \times N_S}$ if $ ev \in \mc E$; $\1_{ N_{\mc V,ev}} =\1_{ N_{\mc V}}$, $X'(ev,k) \in \mathbb{R}^{N_\mc E \times N_S}$, $D_1(ev,k), D_1(ev,k) \in \mathbb{R}^{N_\mc V \times N_S}$ if $ ev \in \mc V$. Moreover, since we let $\mc N^{ext}(0,0) = \mc N(0,0)$, then $S^{\ast}(\mc N^{ext}(0,0)) = S^{\ast}(\mc N(0,0))$.

To avoid cumbersome notation, we drop the notation $(ev,k)$ when $ev$ and $k$ are clear from context.

The following proposition then characterizes the optimal solution of the linear program~\eqref{eq:tp_ext} in the same fashion as in Theorem~\ref{thm:LP_sol}.
\begin{prop}
\label{thm:LP_sol_ext}
$X'(ev,k), D_1'(ev,k), D_2'(ev,k)$ is an optimal solution of~\eqref{eq:tp_ext} if and only if there exists $u(ev,k)= \big( X'(ev,k),D_1'(ev,k), D_2'(ev,k),U_1'(ev,k),\dots, U_6'(ev,k), \vec{h}'_1(ev,k), \vec{h}'_2(ev,k),\vec{h}'_3(ev,k) \big)$ where 
\begin{enumerate}
    \item if $ev \in \mc E$, $U_1'(ev,k)\in \mathbb{R}^{{N_{\mc V^t}} \times {N_S}}$, $U_2'(ev,k)\in \mathbb{R}^{{N_{\mc E^t}} \times {N_S}}$, $U_3'(ev,k),\cdots, U_6'(ev,k)\in \mathbb{R}^{{N_{\mc V^t}} \times {N_S}}$, $\vec{h}'_1(ev,k)\in \mathbb{R}^{N_{\mc E^b}}$, $\vec{h}_2'(ev,k) \in \mathbb{R}^{N_{\mc V^t}}$, and $\vec{h}'_3(ev,k)\in \mathbb{R}^{N_S}$;
    \item if $ev \in \mc V$, $U_1'(ev,k)\in \mathbb{R}^{{N_\mc V} \times {N_S}}$, $U_2'(ev,k)\in \mathbb{R}^{{N_\mc E} \times {N_S}}$, $U_3'(ev,k),\cdots, U_6'(ev,k)\in \mathbb{R}^{{ N_\mc V} \times {N_S}}$, $\vec{h}_1'(ev,k)\in \mathbb{R}^{N_\mc E}$, $\vec{h}_2'(ev,k) \in \mathbb{R}^{N_\mc V}$, and $\vec{h}_3'(ev,k)\in \mathbb{R}^{N_S}$,
\end{enumerate}
such that the following conditions are held:
\begin{align}
\label{eq:LP_ext}
\begin{split}    
&  E'(ev)X'(ev,k) = D_1'(ev,k)+D_2'(ev,k), \\
&   X'(ev,k) \geq 0,\quad
    X'(ev,k)\1_{N_S} \leq C^{\mc E^t}(ev,k),\quad 
     E'(ev)_+ X'(ev,k)\1_{ N_S} \leq C^{\mc V^t}(ev,k),\\
&   D_1'(ev,k)^T \1_{N_{\mc V^t}} = -D_2'(ev,k)^T \1_{ N_{\mc V^t}} \quad \text{ if } ev \in \mc E,\\
&   D_1'(ev,k)^T \1_{N_\mc V} = -D_2'(ev,k)^T \1_{ N_\mc V} \quad \text{ if } ev \in \mc V, \\
&    D_1'(ev,k) \leq  \Delta_1'(ev) M, \quad
    -D_2'(ev,k) \leq  -\Delta_2'(ev) M, \quad
    D_1'(ev,k) \geq 0, \quad
    D_2'(ev,k) \leq 0, \\
    & E'(ev)^T U_1'(ev,k) - U_2'(ev,k) + (\vec{h}'_1(ev,k) + { E'(ev)_+}^T \vec{h}'_2(ev,k))\1^T_{N_S} =0,\\
    & -\1_{ N_{\mc V^t}}\1^T_{N_S} + \1_{N_{\mc V^t}}\vec{h}'_3(ev,k)^T - U_1'(ev,k) + U_3'(ev,k) - U_5'(ev,k) =0 \quad \forall ev \in \mc E,\\
    & -\1_{  N_{\mc V}}\1^T_{N_S} + \1_{  N_{\mc V}}\vec{h}'_3(ev,k)^T - U_1'(ev,k) + U_3'(ev,k) - U_5'(ev,k) =0 \quad \forall ev \in \mc V,\\
    & \1_{  N_{\mc V^t}}\vec{h}'_3(ev,k)^T - U_1'(ev,k) - U_4'(ev,k) + U_6'(ev,k) =0 \qquad \forall ev \in \mc E,\\
    & \1_{  N_{\mc V}}\vec{h}'_3(ev,k)^T - U_1'(ev,k) - U_4'(ev,k) + U_6'(ev,k) =0 \qquad \forall ev \in \mc V,\\
    & U_2'(ev,k),U_3'(ev,k),U_4'(ev,k), U_5'(ev,k), U_6'(ev,k) \geq 0, \quad \vec{h}'_1(ev,k), \vec{h}'_2(ev,k) \geq 0, \\
\end{split}
\end{align}
and 
\begin{align}
\begin{split}    
\label{eq:zero_duality_gap2}    
&\1^T_{N_S} D_1'(ev,k)^T \1_{ N_{\mc V^t}} = \vec{h}'_1(ev,k)^T C^{\mc E^t}(ev,k) +\vec{h}'_2(ev,k)^T C^{\mc V^t}(ev,k) + tr(U_3(ev,k)^T \Delta_1'(ev) M)\\
& \hspace{20em}- tr(U_4'(ev,k)^T \Delta_2'(ev) M) \qquad \forall ev \in \mc E,\\
&\1^T_{N_S} D_1'(ev,k)^T \1_{N_{\mc V}} = \vec{h}'_1(ev,k)^T C^{\mc E^t}(ev,k) +\vec{h}'_2(ev,k)^T C^{\mc V^t}(ev,k) + tr(U'_3(ev,k)^T \Delta_1'(ev) M) \\
& \hspace{20em} - tr(U_4'(ev,k)^T \Delta_2'(ev) M) \qquad \forall ev \in \mc V,\\
\end{split}
\end{align}
\end{prop}
We omit the proof of Proposition~\ref{thm:LP_sol_ext} as it follows closely to that of Theorem~\ref{thm:LP_sol}. As~\eqref{eq:tp_ext} differs from~\eqref{eq:tp_mat} only by several parameters, to show the the results in~\eqref{eq:LP_ext}, one just need to substitute the corresponding parameters.

% \begin{remark}
% The optimal solutions for~\eqref{eq:tp_ext} holds the same property as that in Theorem~\ref{thm:LP_sol} if we substitute the variables with those defined in Section~\ref{sec:extend_model}. Moreover, we name the set of constraints~\eqref{eq:primal_1}--\eqref{eq:zero_duality_gap} as 
% \end{remark}

% As a result, all $u(ev,k)$ that satisfies~\eqref{eq:primal_1}--\eqref{eq:zero_duality_gap} is an optimal solution to~\eqref{eq:tp_ext}.

\section{Optimal Backup Vertiport Planning}
\label{sec:selection}

We now introduce a mathematical optimization framework to plan the capacity and location of backup nodes for a risk-aware air transportation network, aiming at maximizing the expected throughput of the resulted extended network during all possible disruptions together with total budget and cost-effective considerations. 

At each backup location $v_i \in \mc V^b$, where $i \in \mc J_{\mc V^b}$, we either build a backup vertiport with some positive capacity or choose to not building anything. We construct a monotonically increasing sequence $\{C_{v_i,m}\}_{m=1}^{K_Z}$ to represent a list of capacities, where $K_Z \in \mc \mathbb{N}_{>0}$ and $K_Z >1$, $C_{v_i,1} = 0$, and $0< C_{v_i,2} <\cdots <  C_{v_i,K_Z}$. We assume that, to build a backup vertiport at location $v_i$, we can only choose among its corresponding viable capacities $C_{v_i} = C_{v_i,m}$ for some $m = 2, \cdots,K_Z$, and we denote the cost for building a vertiport with capacity $C_{v_i,m}$ at $v_i$ as $F_{v_i,m}>0$. Similarly, if we choose not to build a backup vertiport at $v_i$, then $C_{v_i} = C_{v_i,1} = 0$, and the resulting construction cost is $F_{v_i,1} =0$. We define the $N_{\mc V^b} \times K_Z$ potential capacity matrix as $C^{b, all}$, where $[C^{b, all}]_{i- N_{\mc V},m} = C_{v_{i},m}$ for all $i\in \mc J_{\mc V^b}$ and $m= 1,\cdots,K_Z$; also, we define the $N_{\mc V^b} \times K_Z$ cost matrix for all backup nodes as $F$, where $[F]_{i- N_{\mc V},m} = F_{v_i,m}$ for all $i\in \mc J_{\mc V^b}$ and $m= 1,\cdots,K_Z$.

We let $z_{ij}$ be the binary capacity-selecting variable for all $i = 1, \cdots, N_{\mc V^b}$ and $j = 1,\cdots, K_Z$: if we are building a vertiport at location $v_{i+N_{\mc V}}$ with capacity $C_{v_{i+N_{\mc V}},m}$ for some $m = 2,\cdots,K_Z$, then $z_{im} = 1$ while $z_{ij} = 0$ for all $j\neq m$; if we are not building a vertiport at location $v_{i+N_{\mc V}}$, then $z_{i1} = 1$ while $z_{ij} = 0$ for all $j\neq 1$. We further let $Z$ be the $N_{\mc V^b} \times K_Z$ capacity-selecting decision matrix such that $[Z]_{ij} = z_{ij}$.

Notice that we can only choose one capacity for each backup vertiport, and hence $Z$ needs to satisfy:
\begin{equation}
    \label{eq:z_constr}
    Z \in \{0,1\}^{{N_{\mc V^b }} \times K_Z}, \quad
    Z \1_{K_Z} = \1_{N_{\mc V^b }}\,.
\end{equation}

Given any $Z$ that satisfies~\eqref{eq:z_constr}, we can then compute the capacity vector of all backup nodes as
\begin{equation}
    \label{eq:c_V_b}
    C^{\mc V^b} = (C^{b,all} \odot Z) \1_{K_Z}  \,.
\end{equation}

With the candidate locations of backup vertiports and the selection variable $Z$, we next formulate the problem for selecting the optimal locations and capacities of backup vertiports as a mixed-integer problem.

\subsection{Mixed-Integer Problem for UAM Network Design}
\label{sec:method}
Given any original risk-aware network, the extended graph, and the chosen backup node capacities, we evaluate the performance of the extended network through the \emph{expected throughput} of the network when nodes and links are disturbed, $\mathbb{E} [S^{\ast}(\mc N^{ext}(ev,k))]$, which can be computed as 
\begin{align}
\label{eq:obj_exp}
\begin{split}
\mathbb{E} [S^{\ast}(\mc N^{ext}(ev,k))] &= \sum_{i =1}^{N_{\mc E}} \sum_{ k = 1}^{k_{e_i}} S^{\ast}(\mc N^{ext}(e_i,k)) p_{e_i,k} + \sum_{j=1}^{N_{\mc V}} \sum_{ k = 1}^{k_{v_j}} S^{\ast}(\mc N^{ext}(v_j,k))p_{v_j,k} + (1-P_{dis})S^{\ast}(\mc N(0,0)) \,.
% & = (1-P_{dis})S^{\ast}(\mc N(0,0)) + \max_{\mc U,Z}  \Big(\max_{\mc U}\sum_{i =1}^{N_{\mc E}}\sum_{ k = 1}^{k_{e_i}} p_{e_i,k} \1^T_{N_S} D_1(e_i,k)^T \1_{ N_{\mc V^t}}  \\
% & \qquad \qquad + \sum_{j=1}^{N_{\mc V}} \sum_{ k = 1}^{k_{v_j}} p_{v_j,k} \1^T_{N_S} D_1(v_j,k)^T \1_{ N_\mc V} - w (Z\odot F)\1_{K_Z} \Big) 
\end{split}
\end{align}

As budget is always a concern in the infrastructure-constructing problem, we additionally require that the total constructing cost cannot exceed a given budget level, $\overline{F}$, i.e.,
\begin{equation}
\label{eq:budget}
\1_{N_S}^T(Z\odot F)\1_{K_Z} \leq \overline{F} \,.
\end{equation}

The optimization problem below will choose the capacity-selecting decision matrix $Z$. It chooses the optimal $Z$ that maximizes the weighted difference of expected throughput defined in~\eqref{eq:obj_exp} and the total construction cost $\1_{N_S}^T(Z\odot F)\1_{K_Z}$. The weighting parameter $w$ can be interpreted as a valuation factor that communicates between network throughput and the cost of reaching to a certain throughput level. A smaller $w$ indicates that the planner assigns a higher valuation to the per unit increase in network throughput. Given a $w >0$, we consider the decision-making optimization problem as follows:
\begin{align}
\label{eq:obj}
\begin{split}
    \max_{Z} & \Big (\mathbb{E} [S^{\ast}(\mc N^{ext}(ev,k))] - w \1_{N_S}^T(Z\odot F)\1_{K_Z} \Big)\\
    \text{s.t.} \quad & Z \in \{0,1\}^{{N_{\mc V^b }} \times K_Z}, \quad
    Z \1_{K_Z} = \1_{N_{\mc V^b} } , \\
    & \1_{N_S}^T(Z\odot F)\1_{K_Z} \leq \overline{F}
\end{split}
\end{align}

We denote a feasible solution to problem~\eqref{eq:tp_ext} where we compute $S^{\ast}(\mc N^{ext}(ev,k))$ as a tuple $u(ev,k)$, where 
\begin{align*}
    u(ev,k) &= \big( X'(ev,k),D_1'(ev,k), D_2'(ev,k),U_1'(ev,k),\dots, U_6'(ev,k), \vec{h}_1'(ev,k), \vec{h}_2'(ev,k),\vec{h}_3'(ev,k) \big)
\end{align*}
and let $\mc U = \{u(ev,k) \mid ev\in \mc E\cup \mc V, k = 1,\cdots, k_{ev}\}$.

With~\eqref{eq:obj_exp} and Proposition~\ref{thm:LP_sol_ext}, we then reformulate the optimization problem~\eqref{eq:obj} as 
\begin{align}
\label{eq:obj_split}
\begin{split}
 (1-P_{dis})S^{\ast}(\mc N(0,0)) &  +  \max_{\mc U,Z}  \Big(\sum_{i =1}^{N_{\mc E}}\sum_{ k = 1}^{k_{e_i}} p_{e_i,k} \1^T_{N_S} D_1(e_i,k)^T \1_{ N_{\mc V^t}} + \sum_{j=1}^{N_{\mc V}} \sum_{ k = 1}^{k_{v_j}} p_{v_j,k} \1^T_{N_S} D_1(v_j,k)^T \1_{ N_\mc V}  \\
&   \hspace{22em} - w \1_{N_S}^T (Z\odot F)\1_{K_Z}\Big) \\
\text{s.t.} \quad   
& C^{\mc V^b} = (Z\odot C^{b,all}) \1_{K_Z}, \\
& u(ev,k) \text{ satisfies~\eqref{eq:LP_ext}--\eqref{eq:zero_duality_gap2} } \quad \forall~ev\in \mc E\cup \mc V, k = 1,\cdots,k_{ev} \\
&   Z \in \{0,1\}^{{N_{\mc V^b }} \times K_Z}, \quad
    Z \1_{K_Z} = \1_{N_{\mc V^b} }, \quad
    \1_{N_S}^T(Z\odot F)\1_{K_Z} \leq \overline{F} \,.
\end{split}
\end{align}

The problem~\eqref{eq:obj_split} is a Mixed Integer Program (MIP) with binary variable $Z$. It includes the throughput computation conditions in~\eqref{eq:LP_ext}--\eqref{eq:zero_duality_gap2}, the vertiport-constructing conditions in~\eqref{eq:z_constr}, and the budget constraint in~\eqref{eq:budget}.

Unfortunately, in~\eqref{eq:zero_duality_gap2}, both $C^{\mc V^t}(v)$ (defined in~\eqref{eq:cap_v_add}--\eqref{eq:ext_v}) and $C^{\mc E^t}(e)$ (defined in~\eqref{eq:cap_e_add}--\eqref{eq:cap_e_ext}) are functions of $C^{\mc V^b}$ (defined in~\eqref{eq:c_V_b}), a term with the binary variable $Z$. Therefore, the constraints of problem~\eqref{eq:obj_split} contains the bilinear terms of variables $Z, \vec{h}'_1(e,k), \vec{h}'_2(e,k)$, and $\vec{h}'_2(v,k)$: $\vec{h}'_1(e,k)^T C^{\mc E^t}(e,k)$, $\vec{h}'_2(e,k)^T C^{\mc V^t}(e)$, and $\vec{h}'_2(v,k)^T C^{\mc V^t}(v,k)$. In the next subsection, we are then converting these terms into a series of linear terms of variables.

\subsection{Mixed-Integer Linear Problem for UAM Network Design}
\label{sec:MILP}

We are able to get rid of the bilinear terms  $\vec{h}'_1(e,k)^T C^{\mc E^t}(e,k)$, $\vec{h}'_2(e,k)^T C^{\mc V^t}(e)$, and $\vec{h}'_2(v,k)^T C^{\mc V^t}(v,k)$, using the two propositions below.
\begin{prop}
\label{prop:bigM}
Given $Z \in \{0,1\}^{{N_{\mc V^b }} \times K_Z}$, $\Lambda, \Lambda' \in \mathbb{R}^{{N_{\mc V^b }}\times K_Z}$, we assume $Z\1_{K_Z} = \1_{N_{\mc V^b }}$ and $0\leq \Lambda \leq M_{\Lambda}$ for some $M_{\Lambda} > 0$. Then $\Lambda'$ satisfies the conditions~\eqref{eq:bigM_1}--\eqref{eq:bigM_2}  below if and only if $\Lambda' = \Lambda \odot Z$:
\begin{align}
    \label{eq:bigM_1} 0 &\leq \Lambda' \leq  M_{\Lambda}Z\\
    \label{eq:bigM_2} 0 &\leq \Lambda-\Lambda' \leq  M_{\Lambda}(\1_{N_{\mc V^b }} \1_{K_Z}^T - Z) \,.
\end{align}
\end{prop}
\begin{proof}
The proposition is a matrix version of the big M method.
First, if $\Lambda' = \Lambda \odot Z$, then $[\Lambda']_{ij} = 0$ if $[Z]_{ij} =0$, and $[\Lambda']_{ij} = [\Lambda]_{ij}$ if $[Z]_{ij} =1$. Since $0\leq \Lambda \leq M_{\Lambda}$, then~\eqref{eq:bigM_1} and~\eqref{eq:bigM_2} are satisfied.

On the other side, \eqref{eq:bigM_1} indicates that $\Lambda' =0$ if $[Z]_{ij} =0$, and ~\eqref{eq:bigM_2} indicates that $[\Lambda']_{ij} = [\Lambda]_{ij} \leq M_{\Lambda}$ if $[Z]_{ij} =1$, and hence $\Lambda' = \Lambda \odot Z$.
\end{proof}

\begin{prop}
\label{prop:bilinear}
Let $\vec{h}_{12} \in \mathbb{R}^{N_{\mc E^b}}$, $\vec{h}_{22} \in \mathbb{R}^{N_{\mc V^b}}$, and $\vec{h}_{2} \in \mathbb{R}^{N_{\mc V}}$, such that $\vec{h}_{11},\vec{h}_{21},\vec{h}_{2} \geq 0$. Given $Z \in \{0,1\}^{{N_{\mc V^b}}\times K_Z}$, and let $v\in \mc V$, $e \in \mc E$, then
\begin{enumerate}
    \item If $\Lambda = C^{b,all}\odot( \mc E^{b,2}_+ \cdot \vec{h}_{12} \1^T_{K_Z})$ and $\Lambda' = \Lambda \odot Z$, then $\1^T_{N_{\mc V^b}} \Lambda' \1_{K_Z} = \vec{h}^T_{12} C^{\mc E^b}(e)$;
    \item If $\Lambda = C^{b,all}\odot(\vec{h}_{22} \1^T_{K_Z})$ and $\Lambda' = \Lambda \odot Z$, then $\1^T_{N_{\mc V^b}} \Lambda' \1_{K_Z} = \vec{h}^T_{22} C^{\mc V^b}$;
    \item If $\Lambda = C^{b,all}\odot((\Delta^{adj})^T (\vec{e}_{N_\mc V}(v) \odot \vec{h}_{2}) \1^T_{K_Z})$ and $\Lambda' = \Lambda \odot Z$, then $\1^T_{N_{\mc V^b}} \Lambda' \1_{K_Z} = \vec{h}^T_2 C^{\mc V^b}(v)$.
\end{enumerate}
\end{prop}
Proposition~\ref{prop:bilinear} is straightforward by matrix calculations, and hence we omit the proof here.

With Proposition~\ref{prop:bigM} and~\ref{prop:bilinear} and some large enough $M_{\Lambda}$, we can therefore rewrite the constraint~\eqref{eq:zero_duality_gap2}, substituting the bilinear terms with linear terms and conditions.

For all $e\in \mc E$, $k = 1,\cdots,k_e$, we let $\vec{h}'_{11}(e,k) \in \mathbb{R}^{N_{\mc E}}$, $\vec{h}'_{12}(e,k) \in \mathbb{R}^{N_{\mc E^b}}$, $\vec{h}'_{21}(e,k) \in \mathbb{R}^{N_{\mc V}}$, $\vec{h}'_{22}(e,k) \in \mathbb{R}^{N_{\mc V^b}}$, such that 
$\vec{h}'_1(e,k) =\left(
    \begin{array}{c}
    \vec{h}'_{11}(e,k)  \\
    \hline
    \vec{h}'_{12}(e,k) 
    \end{array} \right)$ and
$\vec{h}'_2(e,k) = \left(
    \begin{array}{c}
    \vec{h}'_{21}(e,k)  \\
    \hline
    \vec{h}'_{22}(e,k)
    \end{array} \right)$. We then have
% \begin{align}
% \begin{split}
% \label{eq:h_split}
% %[\vec{h}_{11}(e,k)]_{ij} = [\vec{h}_{1}(e,k)]_{ij}
%  \vec{h}'_1(e,k) =\left(
%     \begin{array}{c}
%     \vec{h}'_{11}(e,k)  \\
%     \hline
%     \vec{h}'_{12}(e,k) 
%     \end{array} \right) \,, \quad
%     \vec{h}'_2(e,k) = \left(
%     \begin{array}{c}
%     \vec{h}'_{21}(e,k)  \\
%     \hline
%     \vec{h}'_{22}(e,k)
%     \end{array} \right) \,,
% \end{split}
% \end{align}
\begin{align}
\label{eq:zero_gap_e}
\begin{split}
    \1^T_{N_S} D_1'(e,k)^T \1_{ N_{\mc V^t}} &= \vec{h}'_1(e,k)^T C^{\mc E^t}(e,k)+\vec{h}'_2(e,k)^T C^{\mc V^t}(e) \\
 & \hspace{10em}   + tr(U_3'(e,k)^T \Delta_1'(e) M) - tr(U_4'(e,k)^T \Delta_2'(e) M) \\
& = \vec{h}_{11}(e,k)^T C^{\mc E}(e,k)+\vec{h}'_{12}(e,k)^T C^{\mc E^b}(e) + \vec{h}'_{21}(e,k)^T C^{\mc V}_0 + \vec{h}'_{22}(e,k)^T C^{\mc V^b}  \\
&  \hspace{10em} + tr(U_3'(e,k)^T \Delta_1'(e) M) - tr(U_4'(e,k)^T \Delta_2'(e) M) \\
& =\vec{h}'_{11}(e,k)^T C^{\mc E}(e,k)+\1^T_{N_{\mc V^b}} \Lambda_1'(e,k) \1_{K_Z} + \vec{h}'_{21}(e,k)^T C^{\mc V}_0 + \1^T_{N_{\mc V^b}} \Lambda'_2(e,k) \1_{K_Z}\\
&  \hspace{10em} + tr(U_3'(e,k)^T \Delta_1'(e) M) - tr(U_4'(e,k)^T \Delta_2'(e) M)  \,,
\end{split}
\end{align}
where
\begin{align}
\begin{split}
\label{eq:lambda_e}
&\Lambda_1(e,k) = C^{b,all}\odot(  E^{b,2}_+ \cdot \vec{h}'_{12}(e,k) \1^T_{K_Z}), \quad
0 \leq \Lambda'_1(e,k) \leq  M_{\Lambda} Z, \\
& \Lambda_2(e,k) = C^{b,all}\odot(\vec{h}'_{22}(e,k) \1^T_{K_Z}), \quad
0 \leq \Lambda'_2(e,k) \leq  M_{\Lambda}Z, \\
& 0 \leq \Lambda_1(e,k)-\Lambda'_1(e,k) \leq  M_{\Lambda}(\1_{N_{\mc V^b}} \1_{K_Z}^T - Z),\quad
0 \leq \Lambda_2(e,k)-\Lambda'_2(e,k) \leq  M_{\Lambda}(\1_{N_{\mc V^b}} \1_{K_Z}^T - Z) \,.
\end{split}
\end{align}

Similarly,  for all $v\in \mc V$ and $k = 1,\cdots,k_v$, 
\begin{align}
\begin{split}
\label{eq:zero_gap_v}
\1^T_{N_S} D_1'(v,k)^T \1_{ N_{\mc V}} & = \vec{h}'_1(v,k)^T C^{\mc E}_0 +\vec{h}'_2(v,k)^T C^{\mc V^t}(v,k) \\
& \hspace{10em} + tr(U_3'(v,k)^T \Delta_1'(v) M) - tr(U_4'(v,k)^T \Delta_2'(v) M) \\
& =\vec{h}'_1(v,k)^T C^{\mc E}_0 +\vec{h}'^T_2(v,k) C^{\mc V}(v^a,k) +\1^T_{N_{\mc V^b}} \Lambda'_3(v,k) \1_{K_Z}  \\
& \hspace{10em} + tr(U_3'(v,k)^T \Delta_1'(v) M) - tr(U_4'(v,k)^T \Delta_2'(v) M)  \,,
\end{split}
\end{align}
where we let
\begin{align}
\begin{split}
\label{eq:lambda_v}
&\Lambda_3(v,k) = C^{b,all}\odot((\Delta^{adj})^T (\vec{e}_{N_{\mc V}} \odot \vec{h}'_{2}(v,k)) \1^T_{K_Z}), \quad
0 \leq \Lambda'_3(v,k) \leq  M_{\Lambda}Z, \\
& 0 \leq \Lambda_3(v,k)-\Lambda'_3(v,k) \leq  M_{\Lambda}(\1_{N_{\mc V^b}} \1_{K_Z}^T - Z) \,.
\end{split}
\end{align}

%The problem~\eqref{eq:obj_milp} can let a mixed integer linear program (MILP), where all terms are linear with the variables 
We now let $\eta(ev,k)$ be the tuple of solutions,
\begin{align}
\begin{split}
    \eta(ev,k) &= \big(X(ev,k),D_1'(ev,k), D_2'(ev,k),U_1'(ev,k),\dots, U_6'(ev,k), \vec{h}'_1(ev,k), \vec{h}'_2(ev,k),\vec{h}'_3(ev,k),\\
    & \hspace{10em} \Lambda_1(ev,k),\Lambda_2(ev,k),\Lambda_3(ev,k),\Lambda'_1(ev,k),\Lambda'_2(ev,k),\Lambda'_3(ev,k)\big) \,,
\end{split}
\end{align}
and let $\mc H = \{\eta(ev,k) \mid ev\in \mc E\cup \mc V, k = 1,\cdots, k_{ev}\}$ be the set of all solutions.

% $u(ev,k) = \big( X(ev,k),D_1(ev,k), D_2(ev,k),U_1(ev,k),\dots, U_6(ev,k), \vec{h}_1(ev,k), \cdots,\vec{h}_3(ev,k) \big)$
% and let $\mc U = \{u(ev,k) \mid ev\in \mc E\cup \mc V, k = 1,\cdots, k_{ev}\}$.

We rewrite the set of equations with bi-linear terms into the linear conditions as in~\eqref{eq:zero_gap_e}--\eqref{eq:lambda_v}, and hence convert the problem~\eqref{eq:obj_split} into a MILP as follows.
\begin{align}
\label{eq:obj_milp}
\begin{split}
(1-P_{dis})S^{\ast}(\mc N(0,0))  +  &\max_{\mc H,Z}  \Big(\sum_{i =1}^{N_{\mc E}}\sum_{ k = 1}^{k_{e_i}} p_{e_i,k} \1^T_{N_S} D_1(e_i,k)^T \1_{ N_{\mc V^t}}  + \sum_{j=1}^{N_{\mc V}} \sum_{ k = 1}^{k_{v_j}} p_{v_j,k} \1^T_{N_S} D_1(v_j,k)^T \1_{ N_\mc V}\\
& \hspace{23em}- w \1_{N_S}^T(Z\odot F)\1_{K_Z} \Big) \\
\text{s.t.} \quad   
% & \mc H = \{\eta(ev,k) \mid ev\in \mc E\cup \mc V, k = 1,\cdots, k_{ev}\}\\
& \eta(ev,k) \text{ satisfies~\eqref{eq:LP_ext} } \quad \forall ev\in \mc E\cup \mc V, k = 1,\cdots,k_{ev} ,  \\
& \eta(e,k) \text{ satisfies~\eqref{eq:zero_gap_e}--\eqref{eq:lambda_e} } \quad \forall e\in \mc E, k = 1,\cdots,k_{e} , \\
& \eta(v,k) \text{ satisfies~\eqref{eq:zero_gap_v}--\eqref{eq:lambda_v} } \quad \forall v\in \mc E, k = 1,\cdots,k_{v} , \\
&  Z \in \{0,1\}^{{N_{\mc V^b }} \times K_Z}, \quad
    Z \1_{K_Z} = \1_{N_{\mc V^b} }, \quad
    \1_{N_S}^T(Z\odot F)\1_{K_Z} \leq \overline{F} \,.
\end{split}
\end{align}

We are able to solve such MILP~\eqref{eq:obj_milp} with some known solvers, such as Gurobi. In the next section, we show how we apply this optimization approach to three networks with different sizes and network topologies.

\section{Case Study}\label{sec:casestudy}

In this section, we apply the proposed network design and optimization approach described in Sections~\ref{sec:network} and \ref{sec:selection} to design UAM networks for three representative real-world cases in the U.S.---Milwaukee, Atlanta, and Dallas--Fort Worth metropolitan areas. We first introduce the selection and setup of the case study. On the performance evaluation of the network design outcomes, we utilize both design visualizations and quantitative metrics\footnote{The code used in the case study can be found in: https://github.com/QinshuangCoolWei/Risk-aware-Urban-Air-Mobility-Network-Design-with-Overflow-Redundancy.git.}. 

\subsection{General Information and Setup}

In this case study, we apply the proposed approach on three representative cities in the U.S.---Milwaukee, Atlanta, and Dallas--Fort Worth metropolitan areas--- to reflect network diversity in multiple aspects. Table~\ref{tbl:usecases} provides a summary of the basic information of the three use cases. On the city scale, this set of cities include medium (Milwaukee), large (Atlanta), and mega (Dallas--Fort Worth) metropolitan areas. More importantly, the UAM networks of these three metropolitan areas vary in topological structure. Figure~\ref{fig:Setup} shows visualizations of the three original UAM network structures. We design the original UAM networks based on information such as urban structure and travel demand forecast between locations throughout the city. The design of the original UAM network is assumed to be business-driven. In contrast to safety-driven design, a business-driven design, such as network for on-demand UAM operations, seeks to maximize the revenue of the system. Because this work focuses on a network infrastructure expansion problem which chooses the location and capacity of backup vertiports that premise on given original network and potential backup nodes, detailed processes of obtaining the original UAM networks are not included here. In fact, the original networks in this problem can be any design outcomes discussed in Section~\ref{sec:litreview}. 

We first set the values of the undisturbed capacities for all nodes and links in the networks as saved in the Github file. For each node or link $ev$, we then set four discrete values of disturbed capacities at $0, 0.25\cdot C_{ev}, 0.5\cdot C_{ev}$, and $0.75\cdot C_{ev}$, with conditional probability given that $ev$ is disturbed as $0.05,0.1,0.15$, and $0.7$, respectively. We assume that there is always exactly one node or link disturbed in the network, and that each of them is disturbed with equal probability.

\begin{table}[h!]
\centering
\caption{Summary statistics of the use cases for three metropolitan areas}
\label{tbl:usecases}
\begin{tabular}{l|c|c|c}
\hline
                                      & Milwaukee & Atlanta              & Dallas--Fort Worth \\ \hline
Approx. Area Size (miles$^2$)            & 630       & 1,350                 & 3,500              \\
No. of Regular Nodes                  & 7         & 11                   & 15                \\ 
No. of O-D Pairs                      & 12        & 58                   & 64                \\ 
No. of Candidate Backup Nodes        & 24        & 26                   & 45                \\ 
Topology                              & Star      & Hybrid mesh and star & Multi star        \\ \hline
\end{tabular}
\end{table}

The Milwaukee network with 7 regular nodes and 12 O-D pairs is a simplified version to highlight its star topology. The Atlanta network with 11 regular nodes and 58 O-D pairs has a hybrid mesh and star topology, where the Atlanta downtown is a clear hub in the network, while other sub-centers such as Marietta, Sandy Springs, and Doraville, have frequent traffic flows between them as well. The Dallas--Fort Worth network with 15 regular nodes and 64 O-D pairs, also the largest example in this study, has a unique multi star topology. In this area, other than the ``twin cities'' Dallas downtown and Fort Worth downtown, the Dallas Fort Worth International Airport (DFW) and Arlington in between are the other two potential hubs in the UAM network.

The blue asterisks in Figure~\ref{fig:Setup} depict the locations of the candidate backup nodes in each of the city cases. Because the locations of the candidate backup nodes for building small vertiports are restricted by several factors, this problem is naturally a discrete facility location problem, instead of a continuous one. Specifically, we identify these candidate locations with the following rules:
\begin{enumerate}
    \item The geographic conditions must permit. The candidate locations must avoid lakes, rivers, mountains, forests, ranches, restricted areas, etc.
    \item They are not deep inside large residential areas. This is to mitigate the negative societal impacts such as community noise and privacy, which are keys to the integration of UAM. 
    \item They are not deep inside large factories and industrial campuses. At the current stage, we assume that such areas are not appropriate for building backup vertiports. 
    \item The current infrastructure at the candidate locations can potentially be repurposed into small vertiports. Some examples are parking lots, open spaces, rooftops, etc.
    \item The candidate locations are not too far outside of the convex hull of the regular nodes. With limited budget, locations outside of this convex hull are highly likely to become the ``inefficient'' choices.
    \item In addition to the constraints above, the entire candidate set must provide an overall comprehensive coverage of the entire metropolitan area. 
\end{enumerate}

\begin{figure}[t!]
	\centering
        \includegraphics[width=0.3\textwidth]{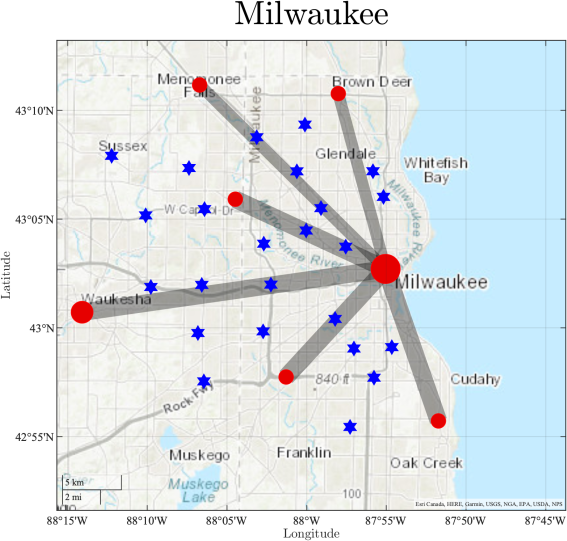}
        \hspace*{0.05cm}
        \includegraphics[width=0.3\textwidth]{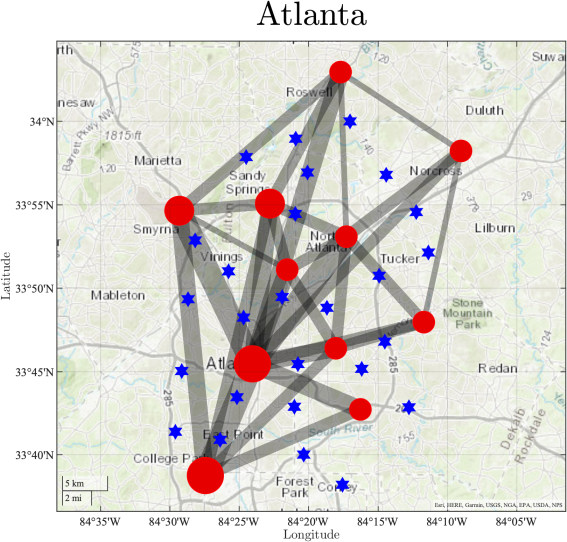}
        \hspace*{0.05cm}
        \includegraphics[width=0.3\textwidth]{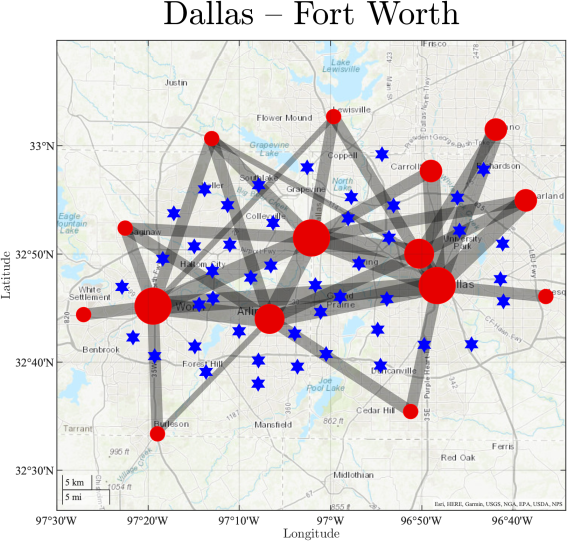}\\
	\caption{Visualizations of the original business-driven UAM network designs with candidate locations for backup nodes in Milwaukee (left), Atlanta (middle), and Dallas--Fort Worth (right) metropolitan areas.}
	\label{fig:Setup}
\end{figure}

In the following experiments, we apply the proposed formulation and optimization approaches to add reserve capacities to the original UAM networks. In particular, we examine the resulting network designs under different budget constraints and valuation settings. Note that we normalize the node and link capacities, as well as the budget and costs for constructing the backup vertiports to relative scales. Node and link capacities are normalized within range 0 to 10. The largest nodes, such as Atlanta downtown and Dallas downtown, have the maximum capacity 10; most of the nodes in the regular network has capacity between 4 and 8. In this case study, the set of O-D pairs is the same as the set of links. As for the alternative path, only the ``single stop'' scenario is considered according to Assumption~\ref{assum:reroute}, i.e., one can replace a direct path $A \rightarrow B$ with $A \rightarrow C \rightarrow B$, where $C$ is a backup node. Furthermore, to be more practical, $C$ qualifies to provide an alternative path for $A \rightarrow B$ if the ratio $r = (d_{A \rightarrow C} + d_{C \rightarrow B})/d_{A \rightarrow B}$ is within $[1.02, 1.5]$ where $d_{\cdot \rightarrow \cdot }$ gives the distance between the nodes. If $r < 1.02$, the backup node is too close to the disturbed link; if $r > 1.5$, the additional travel distance is deemed too long. Because of the design philosophy to avoid large, under-utilized infrastructures, the capacity of a selected backup node is limited to 1 or 2. The costs for building small (capacity 1) and large (capacity 2) vertiports are 4 and 6, respectively. The total budget level may differ with the scale of the original UAM network. In this study we investigate a budget range of 0 to 150. 

We select the locations and capacities for the backup vertiports among the given set of candidate locations by solving the MILP~\eqref{eq:obj_milp}, and then assess the results with two forms of evaluation. First, the design visualizations present the locations and sizes of the selected backup nodes, as well as the resulting backup links for each design setting. Then, using quantitative metrics, we evaluate the design outcome vs. budget level through three crucial aspects of the redundancy of a UAM network: \textit{capacity} (throughput enhancement), \textit{diversity} (travel alternative diversity), and \textit{coverage} (maxi-minimal distance to landing) for contingency management purposes. 

\subsection{Design Visualizations}

\begin{figure}[htbp]
	\centering
        \includegraphics[width=0.3\textwidth]{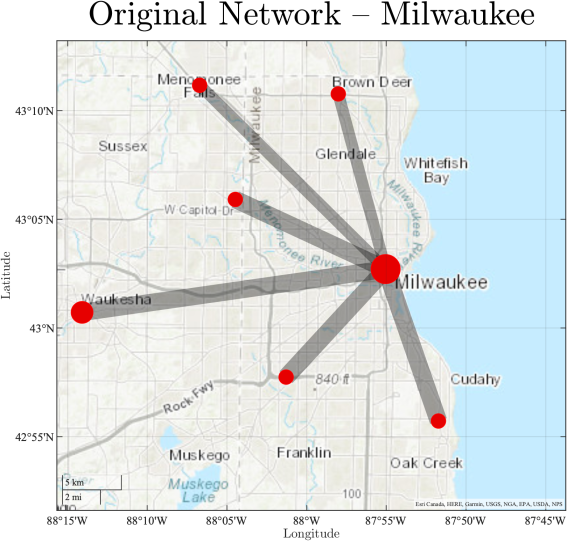}
        \hspace*{0.05cm}
        \includegraphics[width=0.3\textwidth]{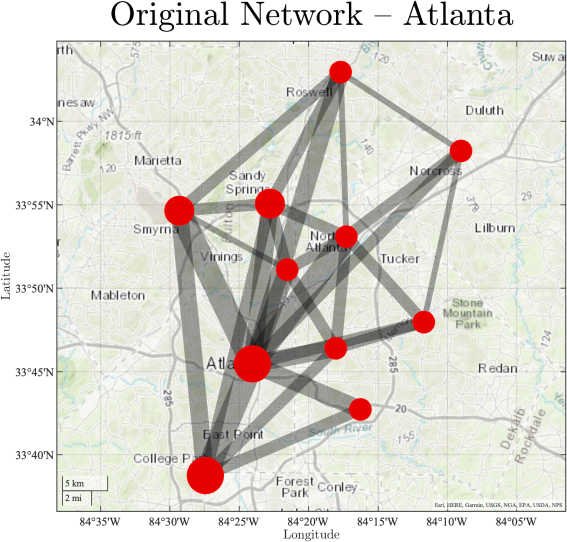}
        \hspace*{0.05cm}
        \includegraphics[width=0.3\textwidth]{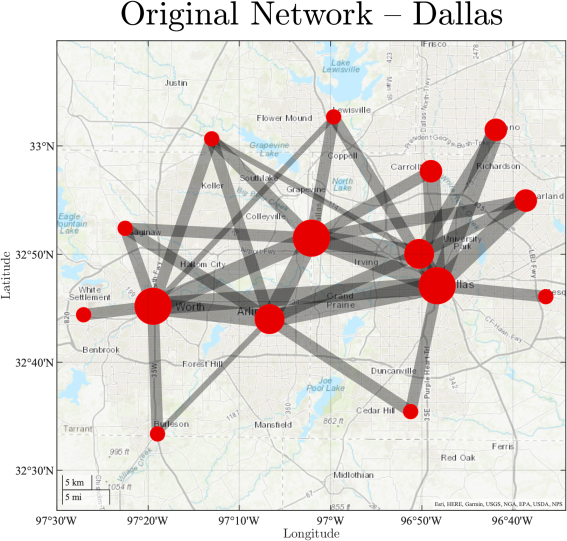}\\
        \vspace*{0.15cm}
        \includegraphics[width=0.3\textwidth]{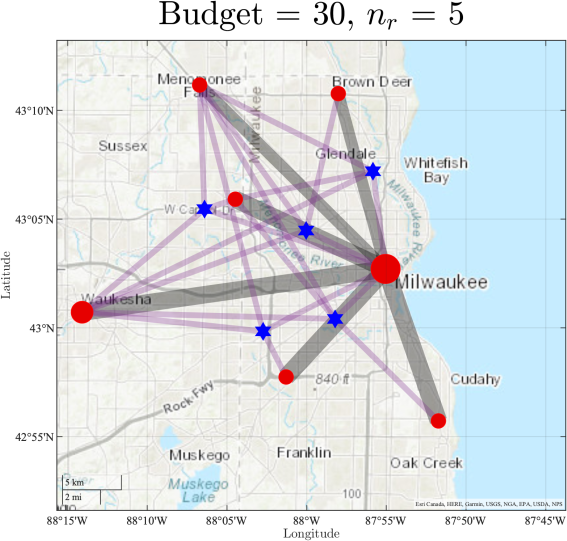}
        \hspace*{0.05cm}
        \includegraphics[width=0.3\textwidth]{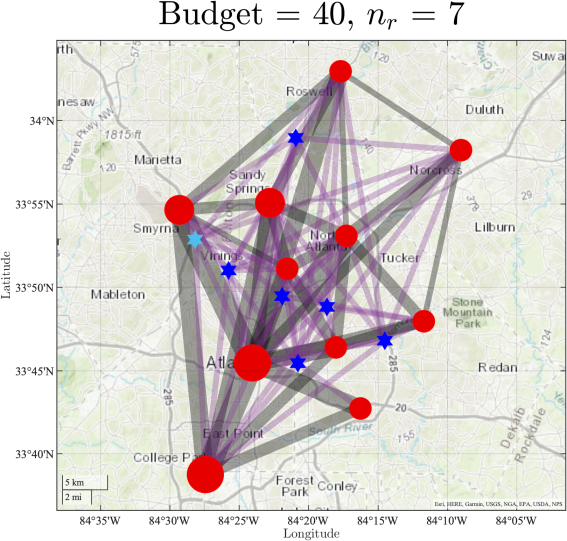}
        \hspace*{0.05cm}
	\includegraphics[width=0.3\textwidth]{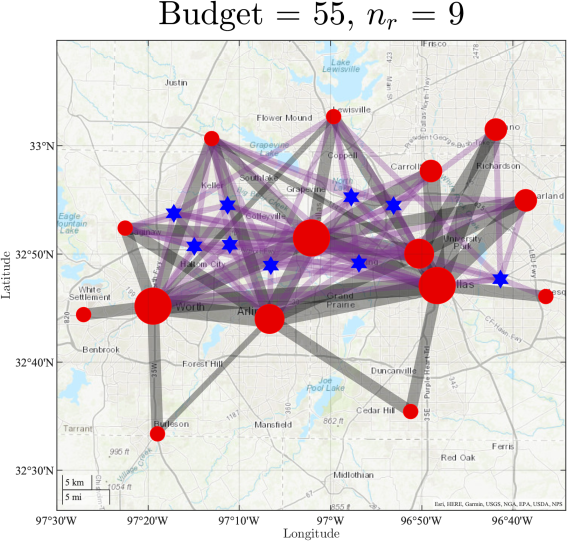}\\
        \vspace*{0.15cm}
        \includegraphics[width=0.3\textwidth]{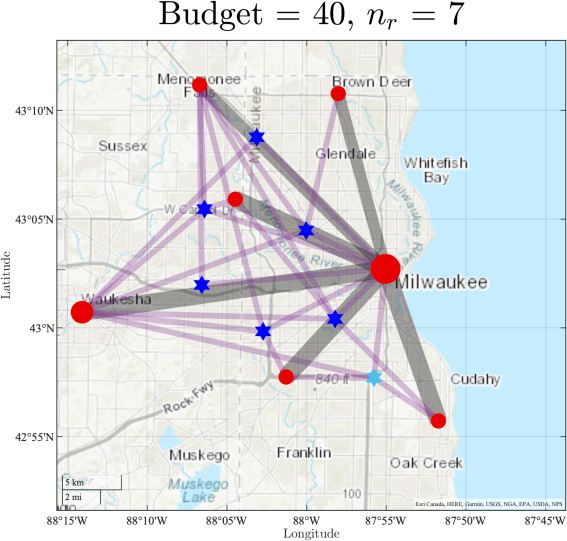}
        \hspace*{0.05cm}
        \includegraphics[width=0.3\textwidth]{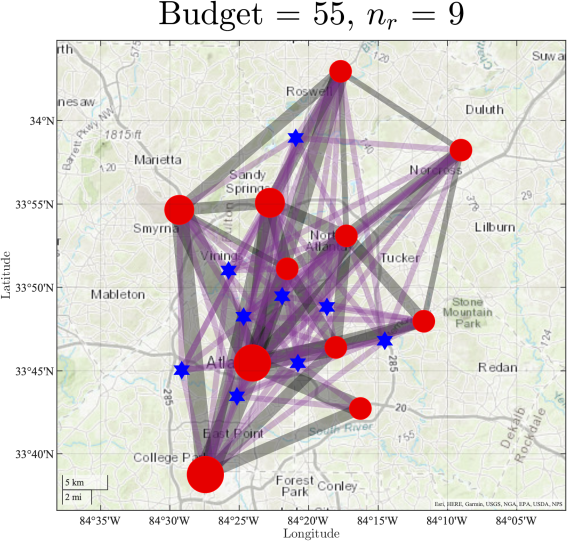}
        \hspace*{0.05cm}
	\includegraphics[width=0.3\textwidth]{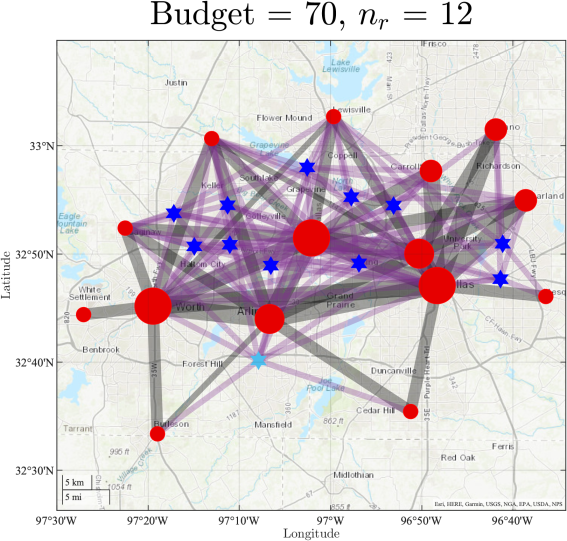}\\
        \vspace*{0.15cm}
        \includegraphics[width=0.3\textwidth]{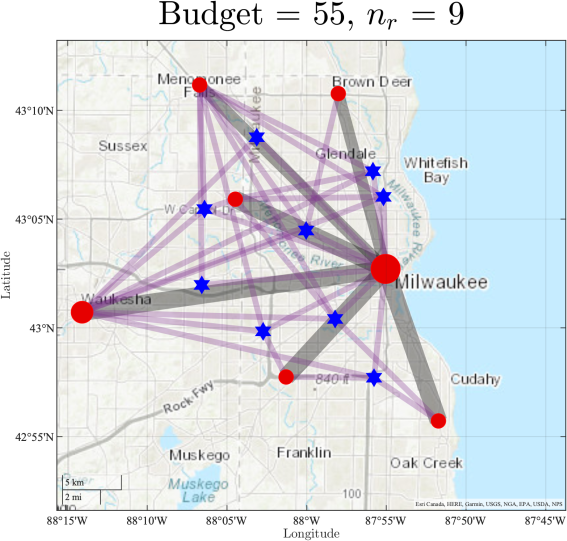}
        \hspace*{0.05cm}
        \includegraphics[width=0.3\textwidth]{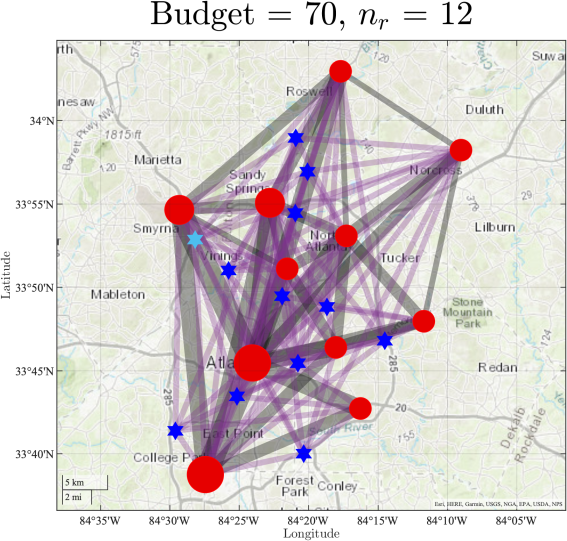}
        \hspace*{0.05cm}
        \includegraphics[width=0.3\textwidth]{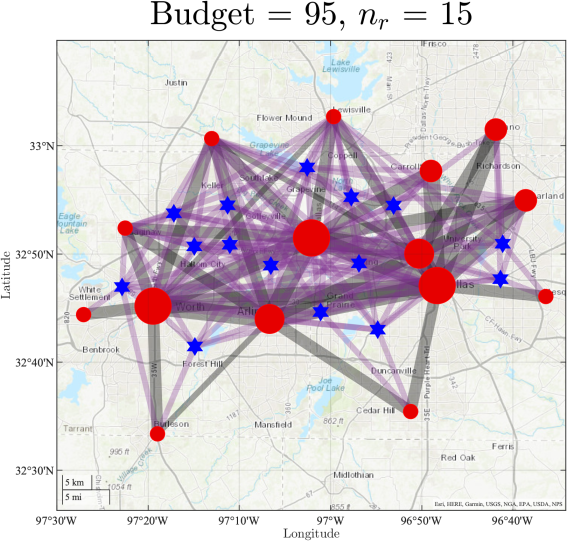}\\
	\caption{Visualizations of UAM network designs at different budget levels for Milwaukee (left), Atlanta (middle), and Dallas--Fort Worth (right) metropolitan areas. For each city, the budget increases from top to bottom, resulting in networks with increased number of backup vertiports and flight corridors. The red circle and blue asterisk denote regular and backup nodes, respectively. Notice that only the selected backup nodes with nonzero capacity are presented in this figure. The gray and purple lines denote regular and backup links, respectively. The size/width of a node/link indicates its capacity.}
	\label{fig:Result1}
\end{figure}

Figure~\ref{fig:Result1} displays, for each city, network designs at three different budget levels. The first row of Figure~\ref{fig:Result1} shows the original networks without any reserve capacity for comparison. %The original networks without any reserve capacity are also shown in the first row of Figure~\ref{fig:Result1} for comparison. 
We represent the backup nodes and links with blue asterisks and purple lines, respectively. One can distinguish the small (capacity 1) and large (capacity 2) backup nodes through the light and dark blue asterisks, respectively. The capacity of a backup link is constrained by the capacity of the backup node it connects with. Note that the backup links in Figure~\ref{fig:Result1} only reflect the choices for alternative travel paths, i.e., a purple line connects original node $A$ with backup node $C$ when $C$ serves as a transit point for at least one O-D pair involving $A$. The backup node $C$ can still serve as a backup landing side for original node $B$ (under disruption) even if $B$ and $C$ are not connected by a purple link in the figure. 

We make some overall observations across all three city cases. First, within the budget level ranges shown in Figure~\ref{fig:Result1}, the numbers of backup nodes and links chosen by our mathematical program increase with the level of budget, which aligns with the intuition. Second, more backup nodes also result in improved coverage over the entire network. Third, when the capacity of a backup node is constrained to 1 or 2, the design process always choose to build as many large nodes (with capacity 2) as possible. At certain budget levels, such as 40 and 70, some small backup nodes (with capacity 1) are chosen as well. Fourth, since the objective is to maximize the total network throughput, the placement of the backup vertiports largely depends on the spatial distribution of O-D pairs in the network. Especially at a small budget level, our method tends to choose locations that can provide alternative travel path for more O-D pairs.

Next, we discuss the result of each individual city case. 
The Milwaukee network has the most straightforward topological structure among the three cases. In this star network, 6 nodes in the Milwaukee area only connect with the central node in Milwaukee downtown. When the budget level is 30, a set of 5 backup nodes can provide alternative travel path for every O-D pair and diverse landing options with good overall coverage. The longest O-D pair in the middle benefits from its ``central'' location and has the most alternative travel options. As the budget level increases, the set of backup nodes can provide more travel alternative options near the periphery as well as better overall coverage. Through the Milwaukee case, we can also see that a backup node will not provide alternative travel path for an O-D pair if its location is too close to the relevant link. 
The Atlanta network has a hybrid mesh and star structure. When the budget level is 40, the optimization program selects 7 backup nodes (6 large and 1 small). Although this set of 7 backup nodes can already provide alternative travel path for every O-D pair in the network, their locations are mostly concentrated in the central region where the traffic is dense. With more budget, more backup nodes from the southern region emerges and the overall coverage becomes more balanced. 
The Dallas--Fort Worth network has a unique multi star structure. In addition, most O-D pairs and traffic distribute in the northern region of the network. As a result, when the budget level is 55, a relatively low budget level considering the scale and complexity of the Dallas--Fort Worth network, all 9 constructed backup vertiports locate in the upper part of the network. When the budget level increases to 70, a set of 12 backup vertiports includes one in the lower part of the network, yet a few O-D pairs still has to operate without alternative travel path. When the budget level is 95, the 15 selected backup vertiports can now provide a highly comprehensive coverage over the entire network. In the meantime, every O-D pair in the network has at least one alternative travel path. In general, it takes more infrastructures and resources to provide comprehensive reserve capacity for a complex network like Dallas--Fort Worth compared to a simple network like Milwaukee.

\begin{figure}[h!]
	\centering
        \includegraphics[width=0.3\textwidth]{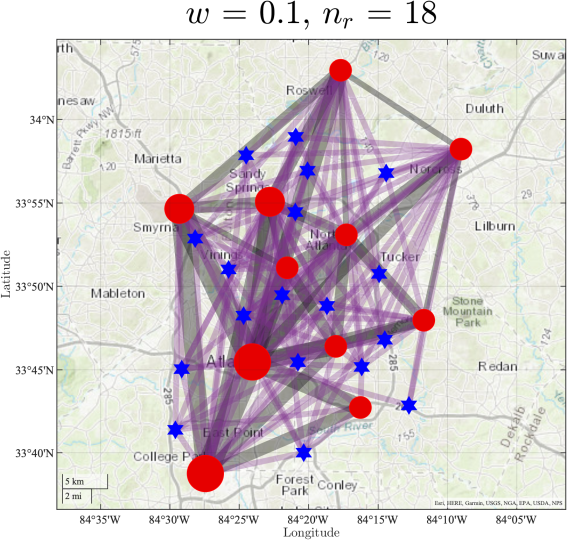}
        \hspace*{0.05cm}
        \includegraphics[width=0.3\textwidth]{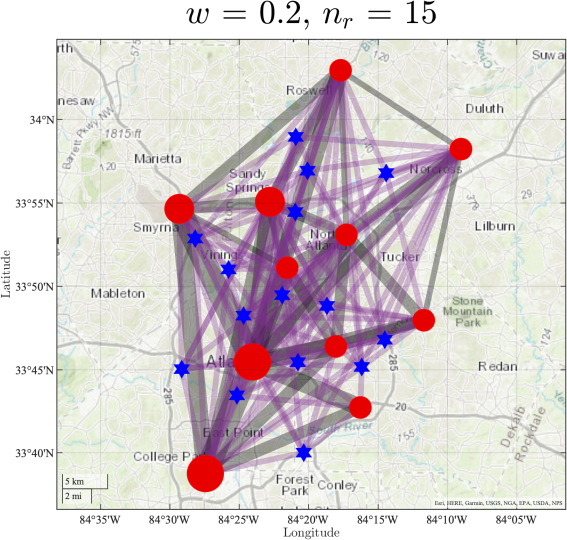}
        \hspace*{0.05cm}
        \includegraphics[width=0.3\textwidth]{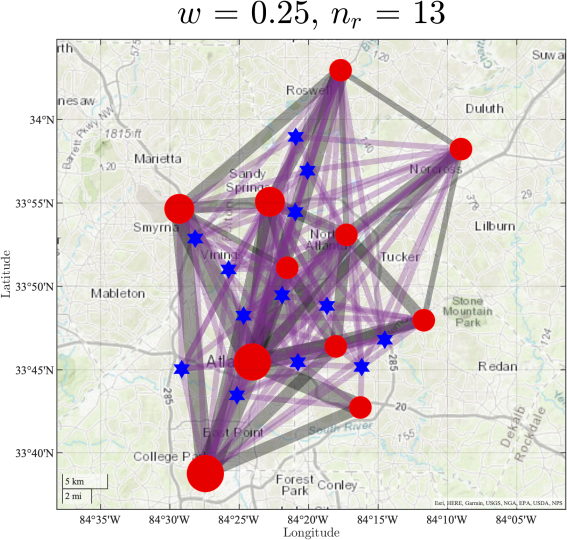}\\
        \vspace*{0.15cm}
        \includegraphics[width=0.3\textwidth]{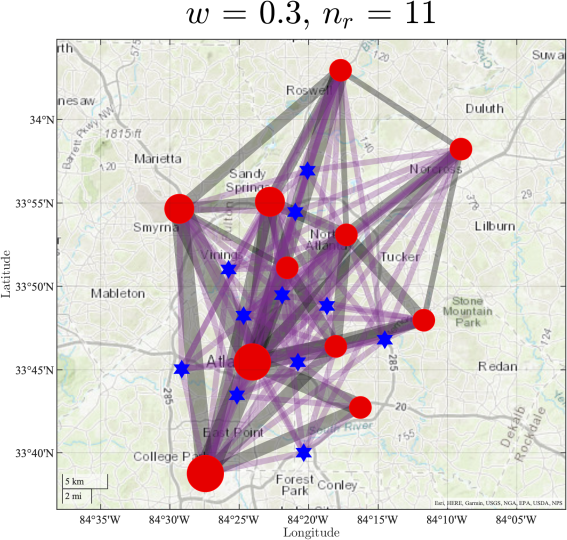}
        \hspace*{0.05cm}
        \includegraphics[width=0.3\textwidth]{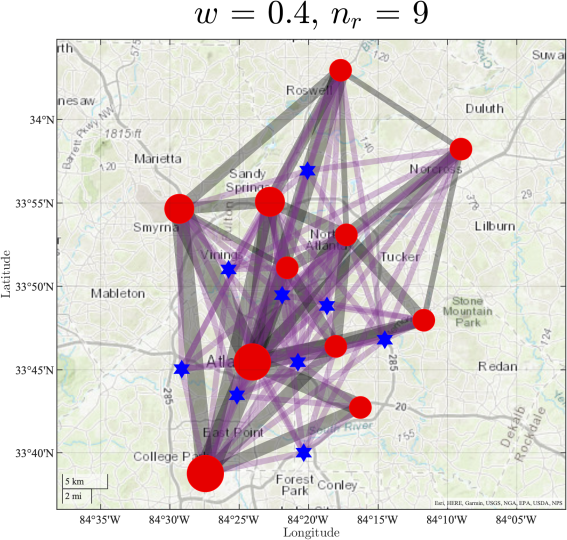}
        \hspace*{0.05cm}
        \includegraphics[width=0.3\textwidth]{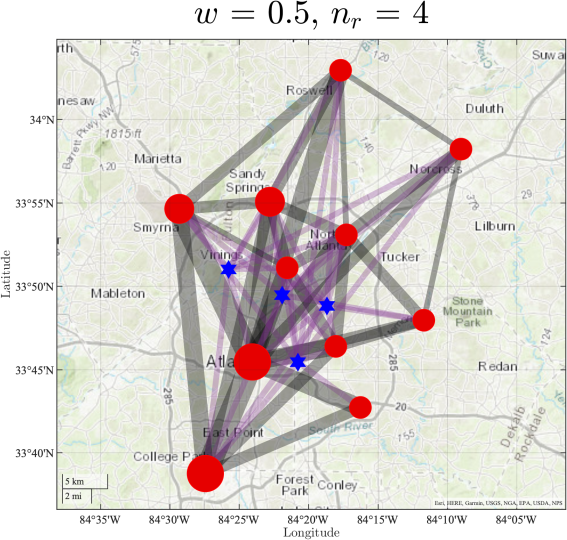}
	\caption{Visualizations of the Atlanta metropolitan area UAM network designs at different valuation ($w$) levels while fixing the budget level at 150.}
	\label{fig:Resultw}
\end{figure}

Other than the budget level, another dimension to visually compare different design outcomes is the valuation parameter. We interpret the weighting parameter $w$ in the objective function of the MILP~(Eqn. \eqref{eq:obj_milp}) associated with the construction cost as the valuation of the additional expected network throughput brought by the additional infrastructures (vertiport, charging station, etc.) built on the backup nodes. A smaller $w$ means that we give a higher value to the per unit increase in expected network throughput. Therefore, at a fixed (high) budget level, a smaller $w$ value should result in the building of more backup nodes in the network. For the sake of brevity, we demonstrate the impact of $w$ through an illustrative example using the Atlanta network. Figure~\ref{fig:Resultw} shows the Atlanta network design outcomes at six different $w$ values while fixing the budget level at 150, the highest budget in this study. As $w$ increases from 0.1 to 0.5 (the valuation of network throughput enhancement decreases), the number of selected backup nodes decreases from 18 to 4. In the meantime, the pattern of selection at different backup node numbers matches with the observations in Figure~\ref{fig:Result1}, similar to a reduced budget manner. This indicates that, if the system designer or policy maker places a low valuation on the per unit increase in expected network throughput, then the capacity-selecting program will choose to build only a small number of backup nodes and infrastructures even with ample budget. This result again confirms that $w$ is a crucial factor in the formulation especially for cost-benefit analysis.

\subsection{Quantitative Metric Comparisons}

In addition to the network visualizations, we also quantitatively assess how the increasing budget may enhance the redundancy of a UAM network. Specifically, we evaluate at different budget levels the three most critical angles of a UAM network design’s redundancy: capacity, diversity, and coverage. To begin with, we utilize network throughput enhancement, which is also the design objective, to evaluate the increment of the UAM network's capacity with the backup infrastructures. Figure~\ref{fig:Exp1} shows the network throughput enhancement results for all three city cases when the budget level is between 0 and 150. 

Formally, we define the total network throughput enhancement $\delta$ as the sum of difference between the throughput with and without backup infrastructures for all possibilities, and define the expected network throughput enhancement $\overline{\delta}$ as the expectation of $\delta$ across the network, i.e., 
\begin{align}
    \delta &= \sum_{i =1}^{N_{\mc E}} \sum_{ k = 1}^{k_{e_i}} \Big(S^{\ast}(\mc N^{ext}(e_i,k)) -\mc N(e_i,k)\Big)\cdot \frac{p_{e_i,k}}{P^{e_i}_{dis}} + \sum_{j=1}^{N_{\mc V}} \sum_{ k = 1}^{k_{v_j}} \Big(S^{\ast}(\mc N^{ext}(v_j,k)) - \mc N(v_j,k)\Big)\cdot \frac{p_{v_i,k}}{P^{v_i}_{dis}} \\%+ (1-P_{dis})S^{\ast}(\mc N(0,0))  \\
   \overline{\delta} &= \sum_{i =1}^{N_{\mc E}} \sum_{ k = 1}^{k_{e_i}} \Big(S^{\ast}(\mc N^{ext}(e_i,k)) -\mc N(e_i,k)\Big)\cdot {p_{e_i,k}} + \sum_{j=1}^{N_{\mc V}} \sum_{ k = 1}^{k_{v_j}} \Big(S^{\ast}(\mc N^{ext}(v_j,k)) - \mc N(v_j,k)\Big)\cdot p_{v_i,k} \,.\end{align}
In the case study, we assume that each link and node in the original network will be disturbed with equal probability, i.e., $P^{e_i}_{dis} = P^{v_j}_{dis} = P_{dis}/(N_{\mc E} +N_{\mc V})$.
% The expected network throughput enhancement $\overline{\delta}$ can be interpreted as the average throughput enhancement of the network, when a node or link in the original network is disturbed.

The top plot of Figure~\ref{fig:Exp1} shows how the expected network throughput enhancement $\overline{\delta}$ grows with budget level. We can observe two important features from the expected network throughput enhancement result. 
The first key observation is that, for all three city cases, the optimal design (and the resulting network throughput enhancement) remains the same when the budget level exceeds a certain threshold. For the current use case and setting, this threshold is 60 for Milwaukee, an 105 for both Atlanta and Dallas--Fort Worth. This indicates that under the default setting of the valuation parameter $w$, it is not necessary to make additional investments beyond this point, even with more budget available. The maximum throughput enhancement a network achieves can depend on many factors, such as the network structure, topology, distribution of O-D pairs, and the capacities of the original nodes and links. 
The second key observation is that, for a network with larger scale, it takes more investments and backup infrastructure to achieve a certain level of average throughput enhancement. For example, in the top plot of Figure~\ref{fig:Exp1}, it takes budget level 50 for the Milwaukee UAM network to reach to an expected network throughput enhancement $\overline{\delta} =0.4$. To achieve the same level of average throughput enhancement, the required budget levels for the Atlanta and Dallas--Fort Worth UAM networks are 80 and 105, respectively. 

Other than the average network throughput enhancement, the bottom plot of Figure~\ref{fig:Exp1} shows the growth of the total network throughput enhancement $\delta$ with budget level. %The total network throughput enhancement is the expected network throughput enhancement multiplied by the number of nodes and links of the network ($N_{\mc V} + N_{\mc E}$).
Now, the total network throughput enhancement $\delta$ can reflect not only the growth pattern, but also the scale of the network. One can see that the total throughput enhancements $\delta$'s of the Atlanta and Dallas--Fort Worth UAM networks with backup vertiports and links are approximately three times of that of the Milwaukee UAM network. 

\begin{figure}[h!]
	\centering
        \includegraphics[width=0.95\textwidth]{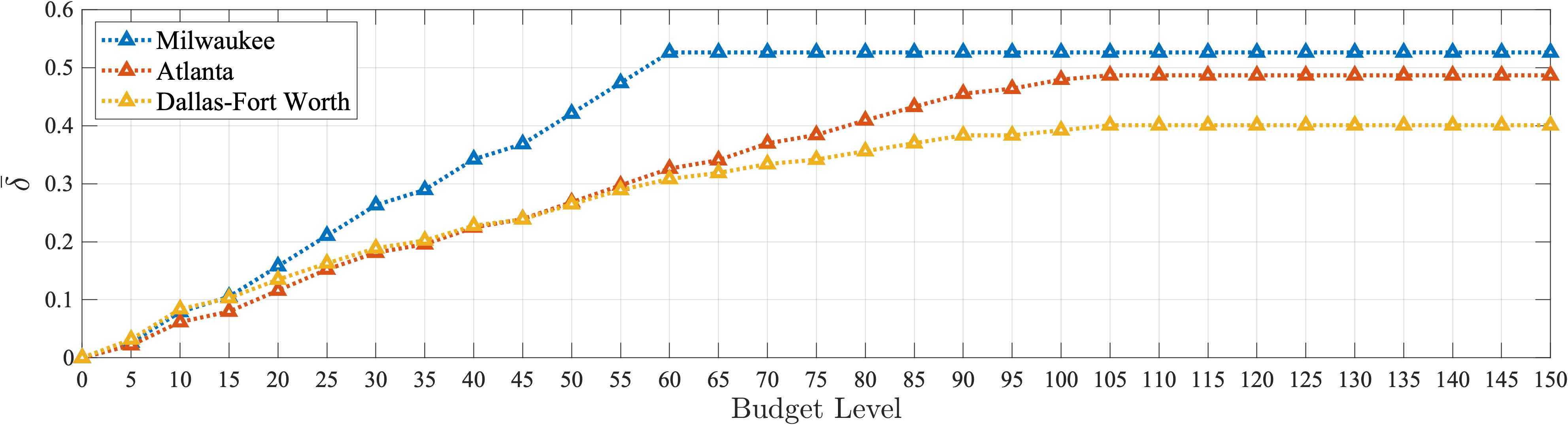}\\
        \vspace*{0.2cm}
        \includegraphics[width=0.95\textwidth]{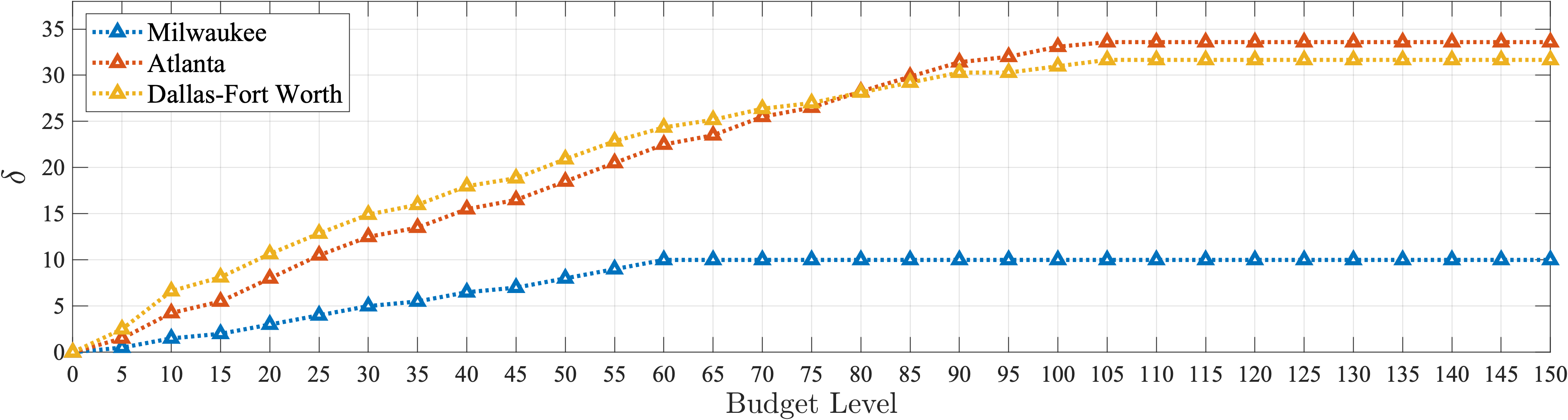}
	\caption{Expected network throughput enhancements $\overline{\delta}$ (top) and total network throughput enhancements $\delta$ (bottom) vs. budget level for Milwaukee, Atlanta, and Dallas--Fort Worth metropolitan areas.}
	\label{fig:Exp1}
\end{figure}

Compared to the network throughput enhancement, the next two quantitative metrics describing diversity and coverage are more of the ``byproducts'' of the design outcome. They are not the objectives of this design process. However, they reflect important aspects of an air transportation network's redundancy and are expected to be improved with the addition of reserve infrastructures. Therefore, we want to investigate how these two aspects are reinforced with increased budget. A table which includes the median and worst case results of these two aspects can be found in Table~\ref{tbl:diversitycoverage} in the Appendix.

\begin{figure}[h!]
	\centering
        \includegraphics[width=0.95\textwidth]{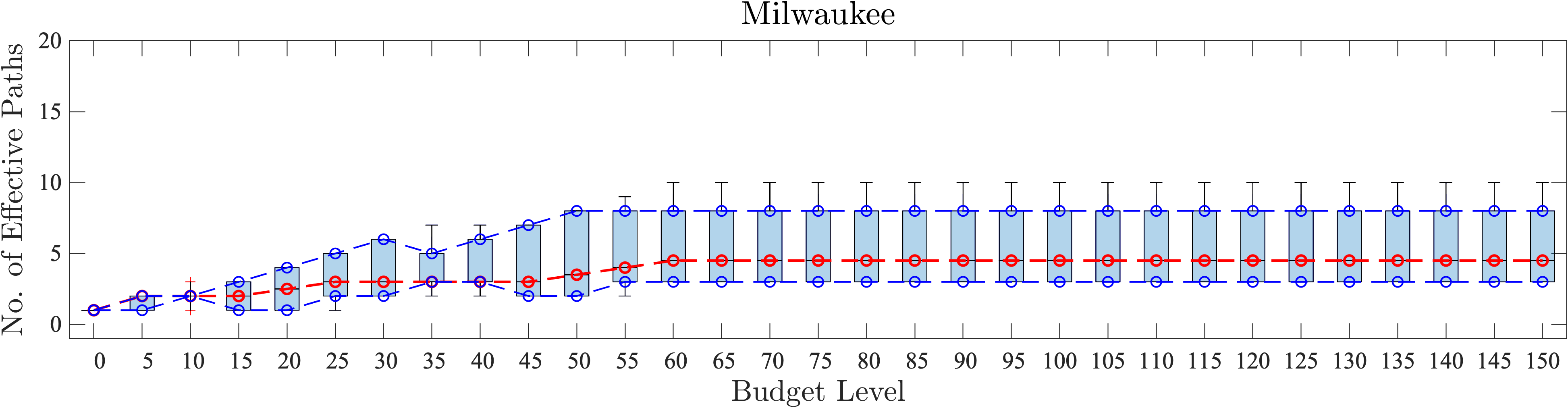}\\
        \vspace*{0.15cm}
        \includegraphics[width=0.95\textwidth]{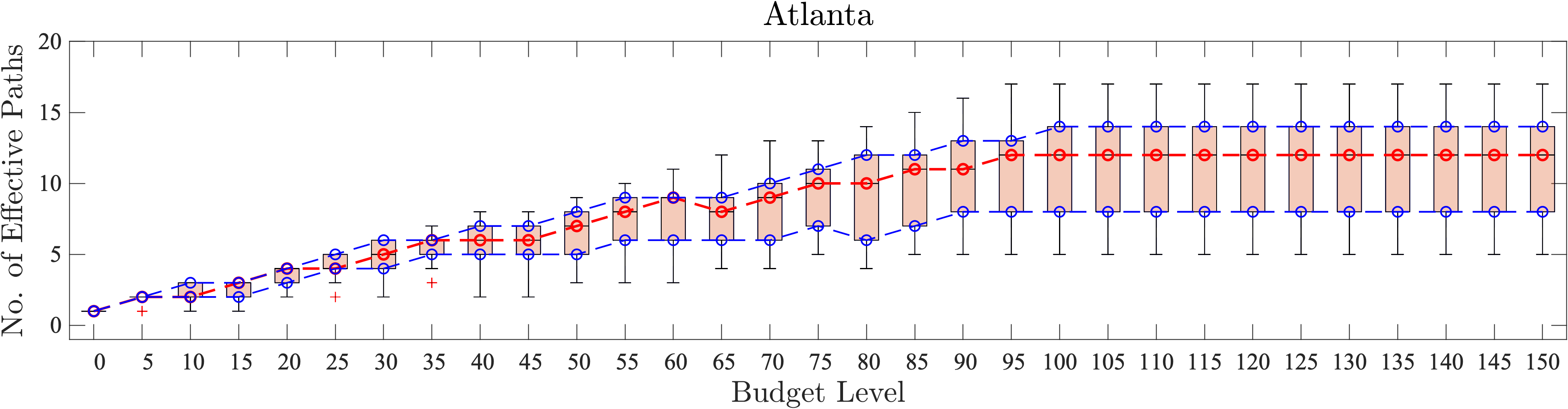}\\
        \vspace*{0.15cm}
        \includegraphics[width=0.95\textwidth]{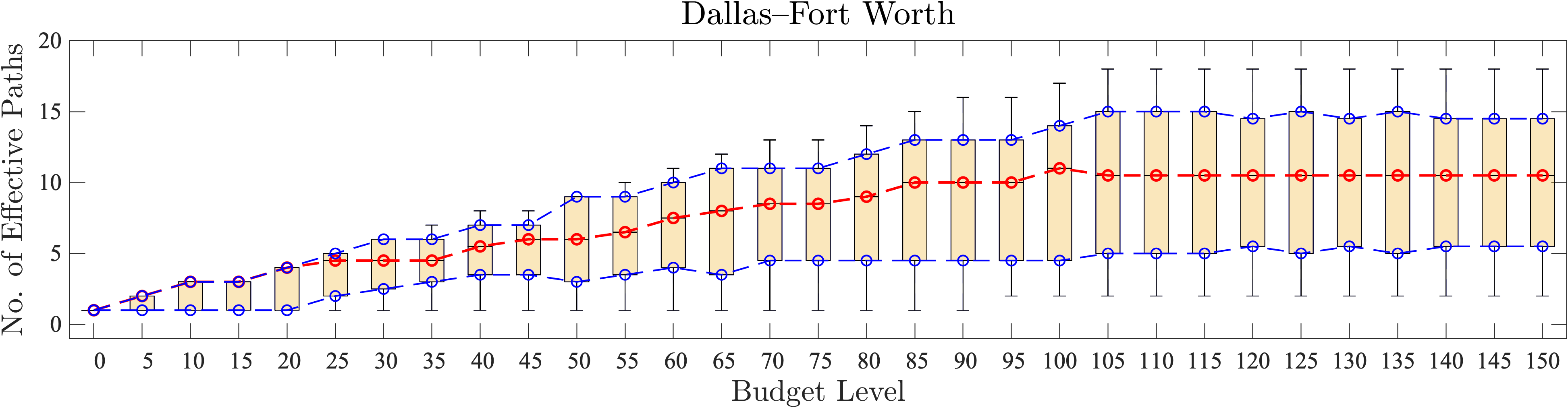}
	\caption{Distribution of travel alternative diversity (number of effective connections between a specific O-D pair) vs. budget level for Milwaukee, Atlanta, and Dallas--Fort Worth metropolitan areas. We use boxplot to display the distribution of travel alternative diversity among all O-D pairs of a network at each budget level.}
	\label{fig:Exp2}
\end{figure}

Figure~\ref{fig:Exp2} demonstrates the travel alternative diversity of the extended networks. The travel alternative diversity of a specific O-D pair is the number of effective connections/paths between the origin and destination. At each budget level, we use boxplot to display the distribution of travel alternative diversity among all O-D pairs of a network. All three subplots of Figure~\ref{fig:Exp2} have the same range on the y-axis. In the original network without any backup infrastructure, each O-D pair has exactly one effective path (link). For the Milwaukee UAM network, the largest design (design with the most backup nodes) has at least 3 effective paths for any O-D pair, while the median is 4.5. At the budget level of 30, a design with 5 reserve nodes (also shown in Figure~\ref{fig:Result1}) can ensure that each O-D pair has at least 2 effective paths, while the median is 3. For a small UAM network like Milwaukee, this is a good level of redundancy on travel alternative diversity. In the largest design of the Atlanta UAM network, the minimum and median numbers of effective paths are 5 and 12, respectively. At the budget level of 50, a design with 8 reserve nodes can ensure that each O-D pair has at least 3 effective paths, while the median is 7. For the Dallas--Fort Worth network, the minimum and median numbers of effective paths in the largest design are 2 and 10.5, respectively, achieved at the budget level of 95. Like analyzed previously, the travel alternative diversity varies among O-D pairs because of the special topological structure and distribution of O-D pairs in the Dallas--Fort Worth network. Since the robustness in travel alternative diversity is not an objective in the design process, the ``worst case'' is improved relatively slowly. This observation motivates a potential robust design approach for future work, which will be further discussion in Section~\ref{sec:remarks}. In the largest design, from Milwaukee to Atlanta to Dallas--Fort Worth the variability in travel alternative diversity increases, which is a joint effect of network scale and topological structure.

Finally, we use maximum landing distance to observe how the addition of the backup nodes may improve the network's overall coverage. In contingency management, a better coverage means that it is easier for an en-route flight to find a nearest vertiport for landing. Although the configuration of an eVTOL aircraft gives it more flexibility to land at a suitable place in the city in case of emergency, being able to land at a vertiport can greatly benefit the continuity of the operations and reduce extra costs. The maximum landing distance of an O-D pair is defined as follows: let $x$ be a location along the link $e$ which connects the O-D pair, and $\mc V^t$ be the set of all (original and backup) vertiports in the current UAM network, the minimum landing distance from $x$ to any vertiport in the network is $\min_{v \in \mc V^t} d(x,v)$, where we use Euclidean distance for the distance function $d(x,v)$. Then, the maximum landing distance along the link of this O-D pair is obtained by the max-min problem $\max_{x \in e} \min_{v \in \mc V^t} d(x,v)$. This metric is an indicator of the worst landing case over the entire flight path. Although this demanding criterion for evaluating the UAM network's coverage is not a design objective in this work, we can still see its improvement with added backup vertiports.

\begin{figure}[h!]
	\centering
        \includegraphics[width=0.95\textwidth]{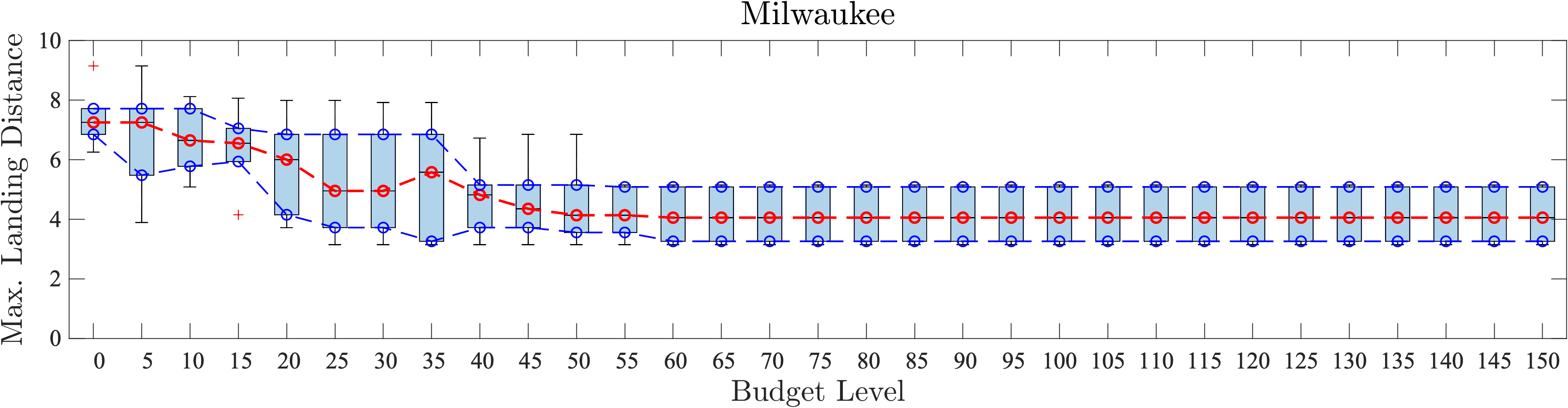}\\
        \vspace*{0.15cm}
        \includegraphics[width=0.95\textwidth]{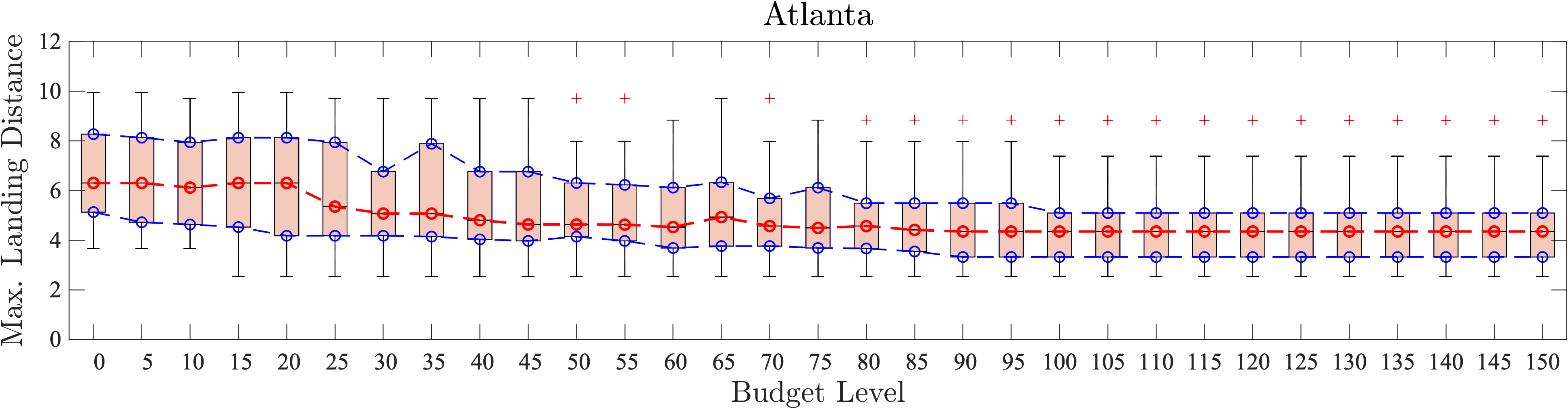}\\
        \vspace*{0.15cm}
        \includegraphics[width=0.95\textwidth]{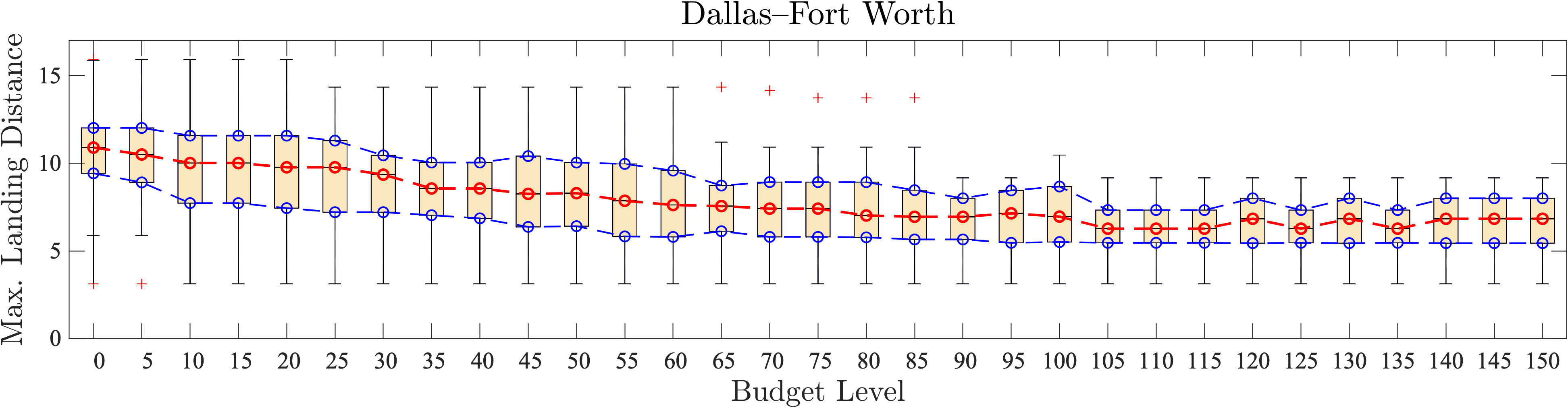}
	\caption{Distribution of maximum landing distance (maximum distance to the nearest landing site along route between a specific O-D pair) in kilometers vs. budget level for Milwaukee, Atlanta, and Dallas--Fort Worth metropolitan areas.}
	\label{fig:Exp3}
\end{figure}

Similar to Figure~\ref{fig:Exp2}, Figure~\ref{fig:Exp3} presents at each budget level the distribution of maximum landing distance among all O-D pairs of a network. In Figure~\ref{fig:Exp3}, although the variation patterns differ between the three city cases due to different network features, the overall improvement trend is clear. For the Milwaukee network, without any backup vertiport the maximum and median values are 9.1 and 7.3 kilometers, respectively. With budget level 30 and 5 backup vertiports, these two values are improved to 7.9 and 5. When the budget level is 55, the largest network with 9 backup vertiports achieves significant improvement in the coverage---the maximum and median are reduced to 5.2 and 4.1, respectively. For the Atlanta network, the median value is reduced from 6.3 in the original network to 4.3 when the budget level exceeds 100. The variability of the distribution also reduces significantly, which is a signal of overall reinforcement in coverage. However, we also notice that the coverage on the ``worst case'' O-D pair is still limited, as the maximum value is only slightly reduced from 10 to 8.8 in the largest design. In fact, there is one O-D pair in the Atlanta network that the set of backup vertiports are still not significantly closer to the mid point of the trip. 
For the Dallas network, the coverage is also steadily improved with the budget level. As the budget level increases from 0 to 85, the median of maximum landing distance decreases significantly from 10.9 to 6.9, while the 25-th and 75-th percentiles are reduced by similar amounts as well. During this budget range, the worst case is reduced from 15.9 to 13.7. Nevertheless, as the budget level changes from 85 to 90, the worst case reduces significantly from 13.7 to 9.2. This single additional reserve vertiport turns out to make a substantial difference on the worst case. Therefore, this design would be preferred if coverage is a secondary consideration in this process. 
Overall, due to the nature of this travel path coverage criterion, our backup vertiports can already achieve considerable improvement. Although one can further improve the coverage if it becomes an objective in the optimization set-up, yet that can only happen at the cost of network throughput and travel alternative diversity. We will further discuss this potential possibility and trade-off study next in Section~\ref{sec:remarks}.

\section{Remarks}\label{sec:remarks}

% Limitations and Future Work
% 1) robust design considerations, and multi-objective
% 2) travel time considerations

Overall, this work raises the new concept of UAM network design with reserve capacity for better operational guarantees under disruption and proposes an optimization framework to optimally deploy such temporarily utilized capacity under limited resources. The case study results also meet the general design expectations and demonstrate the effectiveness of the design approach. While adding value to the current literature in several ways, like most similar projects in the literature, this work can still be further improved in multiple aspects. In this section, we briefly discuss the limitations of this study and corresponding future avenues. 

In the first place, an alternative robust design approach may add redundancy ``more evenly'' throughout the entire UAM network. Current approach seeks design schemes which maximize the total throughput (largest possible sum of O-D flows) of the network. Consequently, under limited budget, the placement of reserve vertiports is prioritized to locations that are likely to ``serve more demands''. While this logic is valid, a potential drawback is that, under certain network topological structures, some proportion of the network could be under-served by the reserve functionality. One such example is the travel diversity distribution result in Figure~\ref{fig:Exp2}. At the budget level of 25, while many (over 25\%) O-D pairs in the Milwaukee and Dallas--Fort Worth networks already have as many as 5 effective travel paths, at least one O-D pair still only has 1 effective path. The same observation and variability in reserve functionality among O-D pairs can also be seen in the maximum landing distance result in Figure~\ref{fig:Exp3}. A robust design formulation optimizing the ``worst case'' among O-D pairs is a potential solution to such imbalance issues. This can be achieved by an alternative objective function with min-max component. However, such a robust design formulation also has disadvantages. Because it neglects the spatial distribution and density of travel demands, some backup vertiports could be even more under-utilized, while others are being over-utilized or even becoming insufficient. Since operational efficiency is a crucial consideration in this planning process, it would be fruitful to explore some trade-offs between robustness and efficiency through formulations that are similar or equivalent to multi-criteria optimization. 

Second, incorporating congestion-aware features in the design of air transportation network can further improve the overall performance of the network. Congestion has been widely explored in ground transportation literature, but not as well explored in air transportation network. Too many flights in the corridors or vertiports comparing to the corresponding capacities can cause congestion, so that the travel time of the flights will increase with the aggregated flights in the network or disturbed network capacity. The proposed design of the risk-aware air transportation network does not consider congestion, while the disturbed network may worsen the congestion in the network because of the capacity drop. As travel time is also one of the main concerns of the transportation systems besides throughput, adding related components into the current problem formulation may provide a more comprehensive analysis to this infrastructure-construction problem. 

A third potential avenue is the simultaneous design of both network redundancy and proactive contingency plans. Operations planning under disruptions is another closely related topic once the UAM network design with reserve capacity becomes reality. A coordination of network design and operations planning can further capture operational dynamics (e.g., centralized planning, vehicle routing, battery charging) in the network optimization process. Based on the current framework, additional components such as queuing theory will further benefit UAM contingency planning in disruption scenarios. 

Fourth, although the current solving time for the optimization problem \eqref{eq:obj_milp} is reasonable given that it is a network design problem instead of a traffic assignment problem, and thus will not be computed too many times over a single network, exploring some decomposition approaches may still improve the efficiency of our algorithm and enhance the scalability of the design problem.

Finally, it will be fruitful to model more complex disruption scenarios. In this work we assume that one node or link is disturbed in each disruption scenario and consider all possible disruption scenarios in an independent and uniform way. One can further consider: (1) when some disruption scenarios are more likely to happen than others, and (2) certain disruption scenarios can happen simultaneously. In the first case, one can assign more differentiated probability distributions to disruption scenarios. In the second case, one can utilize a correlation matrix to capture the dependencies among the individual disruption scenarios. The proposed optimization framework can incorporate these disruption scenarios through introducing the additional components mentioned above.

\section{Conclusion}\label{sec:conclusion}

In this paper, we proposed a novel concept of air transportation network with reserve capacity and developed an optimization framework to plan the locations and capacities of backup vertiports for increasing the network's redundancy for contingency management under disruption. We designed a risk-aware original air transportation network model for on-demand urban mobility service, and a corresponding extended network model with additional reserve capacity (backup nodes and links) under disruption. To maximize the expected throughput of the extended network, we formulate this network design problem into a bi-level and bilinear optimization problem and show that we can obtain an optimal solution for the design problem by solving a MILP. The resulting design is a network scheme with two types of permanent infrastructures -- the original vertiports for on-demand mobility service, and a handful of smaller backup vertiports for overflow and contingency management during disruption. 

In the case study, we applied the proposed methodology on three real-world metropolitan areas with different sizes and network topological structures. We observed that the proposed method chooses more backup vertiports under either a higher budget level or a higher valuation of the network throughput enhancement. In addition, a larger network requires more backup vertiports to attain a certain level of expected throughput enhancement throughout the network. The existence of the backup vertiports can also reinforce the network from the perspectives of travel alternative diversity and emergency landing coverage.
% Although the placement of backup vertiports is dominated by the distribution of O-D pairs and traffic flow in the network, their existence can also reinforce the network from the perspectives of travel alternative diversity and emergency landing coverage. The variation patterns of redundancy metrics vs. budget level can help determine the required investments to achieve specific redundancy objectives. 
Some future avenues of this research include robust designs with balanced enhancement among O-D pairs, incorporation of congestion-aware components, exploration of decomposition approaches, modeling of complex disruption scenarios, and coordination between network design and proactive contingency plans. Overall, the proposed concept and solution can serve as a new paradigm for air transportation network and infrastructure design.

\section*{Acknowledgement}

This work was sponsored by the National Aeronautics and Space Administration (NASA) University Leadership Initiative (ULI) program under project ``Autonomous Aerial Cargo Operations at Scale'', via grant number 80NSSC21M071 to the University of Texas at Austin. The authors are grateful to NASA project technical monitors and project partners for their support. Any opinions, findings, conclusions, or recommendations expressed in this material are those of the authors and do not necessarily reflect the views of the project sponsor.

% References with bibTeX database:

\newpage

\bibliographystyle{model5-names}\biboptions{authoryear}
% \bibliography{Biblio.bib}

\newpage

\appendix

\section{Proof of Theorem~\ref{thm:LP_sol}}
\label{appendix:LP_solution_pf}
%\section*{Appendix}

Since $i,k$ are fixed when proceeding with the problem, we therefore drop all the notations $(i,k)$ in this proof. As~\eqref{eq:tp_mat} is a linear program, we first derive the dual problem of~\eqref{eq:tp_mat}, and then find the optimal solutions to~\eqref{eq:tp_mat} using the necessary and sufficient optimality conditions in linear programming.

The corresponding Lagrangian can be defined as 
\begin{align}
\begin{split}
\label{eq:lagrangian}
&\mc L(X,D_1,D_2,U_1,U_2,\cdots,U_6,\vec{h}_1,\vec{h}_2,\vec{h}_3) = 
-\1^T_{ N_S} D_1^T \1_{N_{\mc V}} \\
& \qquad\qquad +tr(U_1^T E X)-tr(U_1^T D_1)-tr(U_1^T D_2)-tr(U_2^T X)\\
& \qquad\qquad + \vec{h}_1^T (X\1_{N_S} - C^{\mc E}) + \vec{h}_2^T ( E_{+} X\1_{ N_S} - C^{\mc V}) +\vec{h}_3^T D_1^T \1_{ N_{\mc V}} + \vec{h}_3^T D_2^T \1_{ N_{\mc V}} \\
& \qquad\qquad + tr(U_3^T D_1) - tr(U_3^T \Delta_1 M)  - tr(U_4^T D_2) + tr(U_4^T \Delta_1 M) - tr(U_5^T D_1) + tr(U_6^T D_2)
\end{split}
\end{align}

The dual of linear program~\eqref{eq:tp_mat} is then
\begin{align}
\begin{split}
\label{eq:dual_prob}    
\max_{X,D_1,D_2,U_1,U_2,\cdots,U_6,\vec{h}_1,\vec{h}_2,\vec{h}_3} & G(U_1,U_2,\cdots,U_6,\vec{h}_1,\vec{h}_2,\vec{h}_3)\\
:= \max_{X,D_1,D_2,U_1,U_2,\cdots,U_6,\vec{h}_1,\vec{h}_2,\vec{h}_3} &\Big(\min_{X,D_1,D_2}\mc L(X,D_1,D_2,U_1,U_2,\cdots,U_6,\vec{h}_1,\vec{h}_2,\vec{h}_3)\Big)\\
\text{s.t.} \quad & U_2,U_3,U_4, U_5, U_6 \geq 0, \quad \vec{h}_1, \vec{h}_2 \geq 0
\end{split}
\end{align}

Trace is invariant under cyclic permutations. Also, based on matrix calculus, let $A_1, A_2, A_3$ be $n_1 \times n_2$, $n_2 \times n_3$, and $n_3 \times n_4$ matrices, respectively, where $n_1, n_2, n_3, n_4 \in \mathbb{N}$, then
\begin{align*}
    \frac{\partial }{\partial A_2} (A_1 A_2) = A_1^T, \quad \frac{\partial }{\partial A_2} (A_1 A_2 A_3) = A_1^T A_3^T \,.
\end{align*}

Moreover, for any scalar $c \in \mathbb{R}$, $ c = tr(c) = c^T$.

We then derive the following:
\begin{align}
\label{eq:partial_X}
\begin{split}    
& \quad~ \frac{\partial }{\partial X} \mc L(X,D_1,D_2,U_1,U_2,\cdots,U_6,\vec{h}_1,\vec{h}_2,\vec{h}_3) \\
&=  \frac{\partial }{\partial X} tr\Big((U_1^T E - U_2^T +  \1_{ N_S}\vec{h}_1^T + \1_{N_S}\vec{h}_2^T  E_+ )X\Big) \\
&= E^T U_1 - U_2 + (\vec{h}_1+E_+^T \vec{h}_2)\1_{ N_S}^T \,.
\end{split}
\end{align}

Similarly,
\begin{align}
\label{eq:partial_D1}
\begin{split}    
& \quad~ \frac{\partial }{\partial D_1} \mc L(X,D_1,D_2,U_1, U_2,\cdots,U_6,\vec{h}_1,\vec{h}_2,\vec{h}_3) \\
&= \frac{\partial }{\partial D_1} (-\1^T_{N_S} D_1^T \1_{N_{\mc V}}+\vec{h}_3^T D_1^T \1_{ N_{\mc V}}) + \frac{\partial }{\partial D_1} tr(-U_1^T D_1+U_3^T D_1 -U_5^T D_1)\\
&= \frac{\partial }{\partial D_1} tr(-\1^T_{N_{\mc V}} D_1 \1_{N_S}+\1_{N_{\mc V}}^T D_1 \vec{h}_3) + \frac{\partial }{\partial D_1} tr(-U_1^T D_1+U_3^T D_1 -U_5^T D_1)\\
& = -\1_{N_{\mc V}} \1^T_{N_S} + \1_{N_{\mc V}} \vec{h}_3^T - U_1 +U_3-U_5 \,,
\end{split}
\end{align}

and 

\begin{align}
\begin{split}    
& \quad~ \frac{\partial }{\partial D_2} \mc L(X,D_1,D_2,U_1,U_2,\cdots,U_6,\vec{h}_1,\vec{h}_2,\vec{h}_3) \\
&= \frac{\partial }{\partial D_2} tr(-U_1^T D_2 - U_4^T D_2+U_6^T D_2) + \frac{\partial }{\partial D_2} \vec{h}_3^T D_2^T \1_{N_{\mc V}}\\
& = \1_{ N_{\mc V}}\vec{h}_3^T - U_1 - U_4 + U_6\,.
\end{split}
\end{align}

Since the Lagrangian~\eqref{eq:lagrangian} is a linear function of $X,D_1,D_2$, then the minimum of Lagrangian subject to the variables $X,D_1,D_2$ is equal to the Lagrangian if and only if~\eqref{eq:lagrangian} is not changing with values of $X,D_1,D_2$, that is, $G(U_1,U_2,\cdots,U_6,\vec{h}_1,\vec{h}_2,\vec{h}_3) = \mc L(U_1,U_2,\cdots,U_6,\vec{h}_1,\vec{h}_2,\vec{h}_3)$ if and only if $\frac{\partial }{\partial A} \mc L(X,D_1,D_2, \\ U_1,U_2,\cdots,U_6,\vec{h}_1,\vec{h}_2,\vec{h}_3) = 0$, for $A = X, D_1, D_2$. 
We then convert the dual program~\eqref{eq:dual_prob} into the following equivalently:
\begin{align}
\label{eq:dual_prob2}    
\begin{split}
\max_{X,D_1,D_2,U_1,U_2,\cdots,U_6,\vec{h}_1,\vec{h}_2,\vec{h}_3} & -\vec{h}_1^T C^{\mc E} - \vec{h}_2^T C^{\mc V} - tr(U_3^T \Delta_1 M) +tr(U_4^T \Delta_2 M) \\
   \text{s.t.}  \quad & E^T U_1 - U_2 + (\vec{h}_1+E_+^T \vec{h}_2)\1_{N_S}^T =0\\
& -\1_{N_{\mc V}} \1^T_{N_S} + \1_{N_{\mc V}} \vec{h}_3^T - U_1 +U_3-U_5 =0\\
& \1_{ N_{\mc V}}\vec{h}_3^T - U_1 - U_4 + U_6 =0\\
& U_2,U_3,U_4, U_5, U_6 \geq 0, \quad \vec{h}_1, \vec{h}_2 \geq 0
\end{split}
\end{align}

We then conclude that, $X,D_1,D_2$ (resp., $U_1,U_2,\cdots,U_6,\vec{h}_1,\vec{h}_2,\vec{h}_3$) are optimal solutions to the original problem~\eqref{eq:tp_mat} (resp., the dual problem~\eqref{eq:dual_prob2}) if and only if the primal and dual feasibility conditions~\eqref{eq:primal_1}--\eqref{eq:dual_4} and the zero duality gap optimality condition~\eqref{eq:zero_duality_gap} are satisfied.

\newpage

\section{Additional Case Study Results}

\begin{table}[h!]
\scriptsize
\centering
\caption{Information on the median and worst case of travel alternative diversity and maximum landing distance vs. budget level for all three city cases}
\label{tbl:diversitycoverage}
\begin{tabular}{c|cccccc|cccccc}
\hline \hline
\multirow{3}{*}{\begin{tabular}[c]{@{}c@{}}Budget\\ Level\end{tabular}} & \multicolumn{6}{c|}{Travel Alternative Diversity}                                                                                                        & \multicolumn{6}{c}{Maximum Landing Distance}                                                                                                             \\ \cline{2-13} 
                                                                        & \multicolumn{2}{c|}{Milwaukee}                         & \multicolumn{2}{c|}{Atlanta}                           & \multicolumn{2}{c|}{Dallas--Fort Worth} & \multicolumn{2}{c|}{Milwaukee}                         & \multicolumn{2}{c|}{Atlanta}                            & \multicolumn{2}{c}{Dallas--Fort Worth} \\ \cline{2-13} 
                                                                        & \multicolumn{1}{c|}{Median} & \multicolumn{1}{c|}{Min} & \multicolumn{1}{c|}{Median} & \multicolumn{1}{c|}{Min} & \multicolumn{1}{c|}{Median}    & Min   & \multicolumn{1}{c|}{Median} & \multicolumn{1}{c|}{Max} & \multicolumn{1}{c|}{Median} & \multicolumn{1}{c|}{Max}  & \multicolumn{1}{c|}{Median}   & Max    \\ \hline
0                                                                       & \multicolumn{1}{c|}{1.0}    & \multicolumn{1}{c|}{1.0} & \multicolumn{1}{c|}{1.0}    & \multicolumn{1}{c|}{1.0} & \multicolumn{1}{c|}{1.0}       & 1.0   & \multicolumn{1}{c|}{7.3}    & \multicolumn{1}{c|}{9.1} & \multicolumn{1}{c|}{6.3}    & \multicolumn{1}{c|}{10.0} & \multicolumn{1}{c|}{10.9}     & 15.9   \\ \hline
5                                                                       & \multicolumn{1}{c|}{2.0}    & \multicolumn{1}{c|}{1.0} & \multicolumn{1}{c|}{2.0}    & \multicolumn{1}{c|}{1.0} & \multicolumn{1}{c|}{2.0}       & 1.0   & \multicolumn{1}{c|}{7.3}    & \multicolumn{1}{c|}{9.1} & \multicolumn{1}{c|}{6.3}    & \multicolumn{1}{c|}{10.0} & \multicolumn{1}{c|}{10.5}     & 15.9   \\ \hline
10                                                                      & \multicolumn{1}{c|}{2.0}    & \multicolumn{1}{c|}{1.0} & \multicolumn{1}{c|}{2.0}    & \multicolumn{1}{c|}{1.0} & \multicolumn{1}{c|}{3.0}       & 1.0   & \multicolumn{1}{c|}{6.7}    & \multicolumn{1}{c|}{8.1} & \multicolumn{1}{c|}{6.1}    & \multicolumn{1}{c|}{9.7}  & \multicolumn{1}{c|}{10.0}     & 15.9   \\ \hline
15                                                                      & \multicolumn{1}{c|}{2.0}    & \multicolumn{1}{c|}{1.0} & \multicolumn{1}{c|}{3.0}    & \multicolumn{1}{c|}{1.0} & \multicolumn{1}{c|}{3.0}       & 1.0   & \multicolumn{1}{c|}{6.6}    & \multicolumn{1}{c|}{8.1} & \multicolumn{1}{c|}{6.3}    & \multicolumn{1}{c|}{10.0} & \multicolumn{1}{c|}{10.0}     & 15.9   \\ \hline
20                                                                      & \multicolumn{1}{c|}{2.5}    & \multicolumn{1}{c|}{1.0} & \multicolumn{1}{c|}{4.0}    & \multicolumn{1}{c|}{2.0} & \multicolumn{1}{c|}{4.0}       & 1.0   & \multicolumn{1}{c|}{6.0}    & \multicolumn{1}{c|}{8.0} & \multicolumn{1}{c|}{6.3}    & \multicolumn{1}{c|}{10.0} & \multicolumn{1}{c|}{9.8}      & 15.9   \\ \hline
25                                                                      & \multicolumn{1}{c|}{3.0}    & \multicolumn{1}{c|}{1.0} & \multicolumn{1}{c|}{4.0}    & \multicolumn{1}{c|}{2.0} & \multicolumn{1}{c|}{4.5}       & 1.0   & \multicolumn{1}{c|}{5.0}    & \multicolumn{1}{c|}{8.0} & \multicolumn{1}{c|}{5.4}    & \multicolumn{1}{c|}{9.7}  & \multicolumn{1}{c|}{9.8}      & 14.4   \\ \hline
30                                                                      & \multicolumn{1}{c|}{3.0}    & \multicolumn{1}{c|}{2.0} & \multicolumn{1}{c|}{5.0}    & \multicolumn{1}{c|}{2.0} & \multicolumn{1}{c|}{4.5}       & 1.0   & \multicolumn{1}{c|}{5.0}    & \multicolumn{1}{c|}{7.9} & \multicolumn{1}{c|}{5.1}    & \multicolumn{1}{c|}{9.7}  & \multicolumn{1}{c|}{9.4}      & 14.4   \\ \hline
35                                                                      & \multicolumn{1}{c|}{3.0}    & \multicolumn{1}{c|}{2.0} & \multicolumn{1}{c|}{6.0}    & \multicolumn{1}{c|}{3.0} & \multicolumn{1}{c|}{4.5}       & 1.0   & \multicolumn{1}{c|}{5.6}    & \multicolumn{1}{c|}{7.9} & \multicolumn{1}{c|}{5.1}    & \multicolumn{1}{c|}{9.7}  & \multicolumn{1}{c|}{8.6}      & 14.4   \\ \hline
40                                                                      & \multicolumn{1}{c|}{3.0}    & \multicolumn{1}{c|}{2.0} & \multicolumn{1}{c|}{6.0}    & \multicolumn{1}{c|}{2.0} & \multicolumn{1}{c|}{5.5}       & 1.0   & \multicolumn{1}{c|}{4.8}    & \multicolumn{1}{c|}{6.7} & \multicolumn{1}{c|}{4.8}    & \multicolumn{1}{c|}{9.7}  & \multicolumn{1}{c|}{8.6}      & 14.4   \\ \hline
45                                                                      & \multicolumn{1}{c|}{3.0}    & \multicolumn{1}{c|}{2.0} & \multicolumn{1}{c|}{6.0}    & \multicolumn{1}{c|}{2.0} & \multicolumn{1}{c|}{6.0}       & 1.0   & \multicolumn{1}{c|}{4.4}    & \multicolumn{1}{c|}{6.9} & \multicolumn{1}{c|}{4.6}    & \multicolumn{1}{c|}{9.7}  & \multicolumn{1}{c|}{8.3}      & 14.4   \\ \hline
50                                                                      & \multicolumn{1}{c|}{3.5}    & \multicolumn{1}{c|}{2.0} & \multicolumn{1}{c|}{7.0}    & \multicolumn{1}{c|}{3.0} & \multicolumn{1}{c|}{6.0}       & 1.0   & \multicolumn{1}{c|}{4.1}    & \multicolumn{1}{c|}{6.9} & \multicolumn{1}{c|}{4.6}    & \multicolumn{1}{c|}{9.7}  & \multicolumn{1}{c|}{8.3}      & 14.4   \\ \hline
55                                                                      & \multicolumn{1}{c|}{4.0}    & \multicolumn{1}{c|}{2.0} & \multicolumn{1}{c|}{8.0}    & \multicolumn{1}{c|}{3.0} & \multicolumn{1}{c|}{6.5}       & 1.0   & \multicolumn{1}{c|}{4.1}    & \multicolumn{1}{c|}{5.2} & \multicolumn{1}{c|}{4.6}    & \multicolumn{1}{c|}{9.7}  & \multicolumn{1}{c|}{7.9}      & 14.4   \\ \hline
60                                                                      & \multicolumn{1}{c|}{4.5}    & \multicolumn{1}{c|}{3.0} & \multicolumn{1}{c|}{9.0}    & \multicolumn{1}{c|}{3.0} & \multicolumn{1}{c|}{7.5}       & 1.0   & \multicolumn{1}{c|}{4.1}    & \multicolumn{1}{c|}{5.2} & \multicolumn{1}{c|}{4.5}    & \multicolumn{1}{c|}{8.8}  & \multicolumn{1}{c|}{7.6}      & 14.4   \\ \hline
65                                                                      & \multicolumn{1}{c|}{4.5}    & \multicolumn{1}{c|}{3.0} & \multicolumn{1}{c|}{8.0}    & \multicolumn{1}{c|}{4.0} & \multicolumn{1}{c|}{8.0}       & 1.0   & \multicolumn{1}{c|}{4.1}    & \multicolumn{1}{c|}{5.2} & \multicolumn{1}{c|}{4.9}    & \multicolumn{1}{c|}{9.7}  & \multicolumn{1}{c|}{7.6}      & 14.4   \\ \hline
70                                                                      & \multicolumn{1}{c|}{4.5}    & \multicolumn{1}{c|}{3.0} & \multicolumn{1}{c|}{9.0}    & \multicolumn{1}{c|}{4.0} & \multicolumn{1}{c|}{8.5}       & 1.0   & \multicolumn{1}{c|}{4.1}    & \multicolumn{1}{c|}{5.2} & \multicolumn{1}{c|}{4.6}    & \multicolumn{1}{c|}{9.7}  & \multicolumn{1}{c|}{7.4}      & 14.2   \\ \hline
75                                                                      & \multicolumn{1}{c|}{4.5}    & \multicolumn{1}{c|}{3.0} & \multicolumn{1}{c|}{10.0}   & \multicolumn{1}{c|}{5.0} & \multicolumn{1}{c|}{8.5}       & 1.0   & \multicolumn{1}{c|}{4.1}    & \multicolumn{1}{c|}{5.2} & \multicolumn{1}{c|}{4.5}    & \multicolumn{1}{c|}{8.8}  & \multicolumn{1}{c|}{7.4}      & 13.7   \\ \hline
80                                                                      & \multicolumn{1}{c|}{4.5}    & \multicolumn{1}{c|}{3.0} & \multicolumn{1}{c|}{10.0}   & \multicolumn{1}{c|}{4.0} & \multicolumn{1}{c|}{9.0}       & 1.0   & \multicolumn{1}{c|}{4.1}    & \multicolumn{1}{c|}{5.2} & \multicolumn{1}{c|}{4.6}    & \multicolumn{1}{c|}{8.8}  & \multicolumn{1}{c|}{7.0}      & 13.7   \\ \hline
85                                                                      & \multicolumn{1}{c|}{4.5}    & \multicolumn{1}{c|}{3.0} & \multicolumn{1}{c|}{11.0}   & \multicolumn{1}{c|}{5.0} & \multicolumn{1}{c|}{10.0}      & 1.0   & \multicolumn{1}{c|}{4.1}    & \multicolumn{1}{c|}{5.2} & \multicolumn{1}{c|}{4.4}    & \multicolumn{1}{c|}{8.8}  & \multicolumn{1}{c|}{6.9}      & 13.7   \\ \hline
90                                                                      & \multicolumn{1}{c|}{4.5}    & \multicolumn{1}{c|}{3.0} & \multicolumn{1}{c|}{11.0}   & \multicolumn{1}{c|}{5.0} & \multicolumn{1}{c|}{10.0}      & 1.0   & \multicolumn{1}{c|}{4.1}    & \multicolumn{1}{c|}{5.2} & \multicolumn{1}{c|}{4.3}    & \multicolumn{1}{c|}{8.8}  & \multicolumn{1}{c|}{6.9}      & 9.2    \\ \hline
95                                                                      & \multicolumn{1}{c|}{4.5}    & \multicolumn{1}{c|}{3.0} & \multicolumn{1}{c|}{12.0}   & \multicolumn{1}{c|}{5.0} & \multicolumn{1}{c|}{10.0}      & 2.0   & \multicolumn{1}{c|}{4.1}    & \multicolumn{1}{c|}{5.2} & \multicolumn{1}{c|}{4.3}    & \multicolumn{1}{c|}{8.8}  & \multicolumn{1}{c|}{7.1}      & 9.2    \\ \hline
100                                                                     & \multicolumn{1}{c|}{4.5}    & \multicolumn{1}{c|}{3.0} & \multicolumn{1}{c|}{12.0}   & \multicolumn{1}{c|}{5.0} & \multicolumn{1}{c|}{11.0}      & 2.0   & \multicolumn{1}{c|}{4.1}    & \multicolumn{1}{c|}{5.2} & \multicolumn{1}{c|}{4.3}    & \multicolumn{1}{c|}{8.8}  & \multicolumn{1}{c|}{7.0}      & 10.5   \\ \hline
105                                                                     & \multicolumn{1}{c|}{4.5}    & \multicolumn{1}{c|}{3.0} & \multicolumn{1}{c|}{12.0}   & \multicolumn{1}{c|}{5.0} & \multicolumn{1}{c|}{10.5}      & 2.0   & \multicolumn{1}{c|}{4.1}    & \multicolumn{1}{c|}{5.2} & \multicolumn{1}{c|}{4.3}    & \multicolumn{1}{c|}{8.8}  & \multicolumn{1}{c|}{6.3}      & 9.2    \\ \hline
110                                                                     & \multicolumn{1}{c|}{4.5}    & \multicolumn{1}{c|}{3.0} & \multicolumn{1}{c|}{12.0}   & \multicolumn{1}{c|}{5.0} & \multicolumn{1}{c|}{10.5}      & 2.0   & \multicolumn{1}{c|}{4.1}    & \multicolumn{1}{c|}{5.2} & \multicolumn{1}{c|}{4.3}    & \multicolumn{1}{c|}{8.8}  & \multicolumn{1}{c|}{6.3}      & 9.2    \\ \hline
115                                                                     & \multicolumn{1}{c|}{4.5}    & \multicolumn{1}{c|}{3.0} & \multicolumn{1}{c|}{12.0}   & \multicolumn{1}{c|}{5.0} & \multicolumn{1}{c|}{10.5}      & 2.0   & \multicolumn{1}{c|}{4.1}    & \multicolumn{1}{c|}{5.2} & \multicolumn{1}{c|}{4.3}    & \multicolumn{1}{c|}{8.8}  & \multicolumn{1}{c|}{6.3}      & 9.2    \\ \hline
120                                                                     & \multicolumn{1}{c|}{4.5}    & \multicolumn{1}{c|}{3.0} & \multicolumn{1}{c|}{12.0}   & \multicolumn{1}{c|}{5.0} & \multicolumn{1}{c|}{10.5}      & 2.0   & \multicolumn{1}{c|}{4.1}    & \multicolumn{1}{c|}{5.2} & \multicolumn{1}{c|}{4.3}    & \multicolumn{1}{c|}{8.8}  & \multicolumn{1}{c|}{6.8}      & 9.2    \\ \hline
125                                                                     & \multicolumn{1}{c|}{4.5}    & \multicolumn{1}{c|}{3.0} & \multicolumn{1}{c|}{12.0}   & \multicolumn{1}{c|}{5.0} & \multicolumn{1}{c|}{10.5}      & 2.0   & \multicolumn{1}{c|}{4.1}    & \multicolumn{1}{c|}{5.2} & \multicolumn{1}{c|}{4.3}    & \multicolumn{1}{c|}{8.8}  & \multicolumn{1}{c|}{6.3}      & 9.2    \\ \hline
130                                                                     & \multicolumn{1}{c|}{4.5}    & \multicolumn{1}{c|}{3.0} & \multicolumn{1}{c|}{12.0}   & \multicolumn{1}{c|}{5.0} & \multicolumn{1}{c|}{10.5}      & 2.0   & \multicolumn{1}{c|}{4.1}    & \multicolumn{1}{c|}{5.2} & \multicolumn{1}{c|}{4.3}    & \multicolumn{1}{c|}{8.8}  & \multicolumn{1}{c|}{6.8}      & 9.2    \\ \hline
135                                                                     & \multicolumn{1}{c|}{4.5}    & \multicolumn{1}{c|}{3.0} & \multicolumn{1}{c|}{12.0}   & \multicolumn{1}{c|}{5.0} & \multicolumn{1}{c|}{10.5}      & 2.0   & \multicolumn{1}{c|}{4.1}    & \multicolumn{1}{c|}{5.2} & \multicolumn{1}{c|}{4.3}    & \multicolumn{1}{c|}{8.8}  & \multicolumn{1}{c|}{6.3}      & 9.2    \\ \hline
140                                                                     & \multicolumn{1}{c|}{4.5}    & \multicolumn{1}{c|}{3.0} & \multicolumn{1}{c|}{12.0}   & \multicolumn{1}{c|}{5.0} & \multicolumn{1}{c|}{10.5}      & 2.0   & \multicolumn{1}{c|}{4.1}    & \multicolumn{1}{c|}{5.2} & \multicolumn{1}{c|}{4.3}    & \multicolumn{1}{c|}{8.8}  & \multicolumn{1}{c|}{6.8}      & 9.2    \\ \hline
145                                                                     & \multicolumn{1}{c|}{4.5}    & \multicolumn{1}{c|}{3.0} & \multicolumn{1}{c|}{12.0}   & \multicolumn{1}{c|}{5.0} & \multicolumn{1}{c|}{10.5}      & 2.0   & \multicolumn{1}{c|}{4.1}    & \multicolumn{1}{c|}{5.2} & \multicolumn{1}{c|}{4.3}    & \multicolumn{1}{c|}{8.8}  & \multicolumn{1}{c|}{6.8}      & 9.2    \\ \hline
150                                                                     & \multicolumn{1}{c|}{4.5}    & \multicolumn{1}{c|}{3.0} & \multicolumn{1}{c|}{12.0}   & \multicolumn{1}{c|}{5.0} & \multicolumn{1}{c|}{10.5}      & 2.0   & \multicolumn{1}{c|}{4.1}    & \multicolumn{1}{c|}{5.2} & \multicolumn{1}{c|}{4.3}    & \multicolumn{1}{c|}{8.8}  & \multicolumn{1}{c|}{6.8}      & 9.2    \\ \hline \hline
\end{tabular}
\end{table}

\end{document}